\documentclass[sigconf,nonacm,review=false,authorsperrow=4]{acmart}

\AtBeginDocument{%
  \providecommand\BibTeX{{%
    \normalfont B\kern-0.5em{\scshape i\kern-0.25em b}\kern-0.8em\TeX}}}

\setcopyright{acmlicensed}
\copyrightyear{2026}
\acmYear{2026}
\acmDOI{XXXXXXX.XXXXXXX}
\acmConference[SIGMOD '26]{ACM International Conference on Management of Data}{May 31--June 5, 2026}{Bengaluru, India}
\acmISBN{978-1-4503-XXXX-X/2018/06}

\usepackage{amsmath} 
\usepackage{balance} 
\usepackage{booktabs} 
\usepackage{caption}
\usepackage{comment}
\usepackage{enumitem}
\setlist[itemize]{leftmargin=2.2em}

\usepackage{fancyvrb} 
\usepackage{graphicx} 
\usepackage{makecell} 
\usepackage{multirow} 
\usepackage{soul}
\usepackage{subfigure} 
\usepackage{xcolor}
\usepackage{xspace} 
\usepackage{bm}
\usepackage{utfsym}
\usepackage{cancel}
\usepackage{pifont}
\usepackage{marginnote}
\usepackage{mparhack}

\usepackage[linesnumbered,ruled,vlined]{algorithm2e}
\SetArgSty{textup}
\SetKw{And}{\textbf{and}}
\SetKw{Or}{\textbf{or}}
\SetKw{To}{\textbf{to}}

\newcounter{algoline}
\newcommand\Numberline{\refstepcounter{algoline}\nlset{\thealgoline}}
\AtBeginEnvironment{algorithm}{\setcounter{algoline}{0}}

\usepackage{tikz}

\newcommand{\vldbrevision}[1]{\textcolor{black}{#1}}
\usepackage{xcolor}
\newenvironment{vldbrevisionenv}
  {\begingroup\color{black}}
  {\endgroup}

\theoremstyle{definition}
\newtheorem{definition}{Definition}

\newtheorem{theorem}{Theorem}
\theoremstyle{remark}

\newcommand{\vb}[1]{\mathbf{#1}}

\newcommand\dbname{\ensuremath{\textsf{NeurBench}}\xspace}

\newcommand{\paragraphtitle}[1]{\vspace{0.1em}\noindent\textbf{#1.}}

\renewcommand{\subparagraph}[1]{\vspace{0.1em}
\noindent\textit{\underline{#1.}}}

\usepackage{enumitem}
\setlist[itemize]{leftmargin=1em}
\setlist[enumerate]{leftmargin=1em}

\newif\ifextended\extendedfalse

\newcommand{\extended}[1]{\ifextended#1\else\relax\fi}

\DeclareMathOperator*{\argmax}{arg\,max}
\DeclareMathOperator*{\argmin}{arg\,min}

\usepackage{multicol}

\usepackage{wrapfig}
\usepackage{listings}
\usepackage[most]{tcolorbox}

\newtcolorbox[auto counter]{mybox}[1][]{%
breakable,
enhanced,
sharp corners,
colback=white,
fonttitle=\bfseries,
enlarge bottom at break by=5mm,
enlarge top at break by=5mm,
overlay first={%
    \draw[black, line width=0.5mm](frame.south west)--(frame.south east);
    \node[anchor=north east] at (frame.south east) {continued on next page};
    },
overlay middle={%
    \draw[black, line width=0.5mm](frame.south west)--(frame.south east);
    \draw[black, line width=0.5mm](frame.north west)--(frame.north east);
    \node[anchor=north east] at (frame.south east) {continued on next page};
    \node[anchor=south west] at (frame.north west) {continued from next page};
    },
overlay last={%
    \draw[black, line width=0.5mm](frame.north west)--(frame.north east);
    \node[anchor=south west] at (frame.north west) {continued from next page};},
#1
}

  {\list{}{\leftmargin=0.15in\rightmargin=0.15in}\item[]}%
  {\endlist}

\usepackage{multicol}
\newtcolorbox{myquote}[1][]{
    colback=black!10,
    colframe=black!10,
    notitle,
    sharp corners,
    enhanced,
    breakable,
    left=2pt,
    right=2pt,
    top=2pt,
    bottom=2pt,
    ignore nobreak,
}

\definecolor{metacolor}{HTML}{0072B2} 
\definecolor{R1color}{HTML}{800080} 
\definecolor{R2color}{HTML}{008866} 
\definecolor{R3color}{HTML}{D55E00} 

\newtcolorbox{myquoteMeta}[1][]{
    colback=black!3,
    colframe=black!3,
    notitle,
    sharp corners,
    borderline west={1.5pt}{0pt}{blue},
    enhanced,
    breakable,
    left=2pt,
    right=2pt,
    top=0pt,
    bottom=0pt,
    ignore nobreak,
}

\newtcolorbox{myquoteR1}[1][]{
    colback=black!3,
    colframe=black!3,
    notitle,
    sharp corners,
    borderline west={1.5pt}{0pt}{R1color},
    enhanced,
    breakable,
    left=2pt,
    right=2pt,
    top=0pt,
    bottom=0pt,
    ignore nobreak,
}

\newtcolorbox{myquoteR2}[1][]{
    colback=black!3,
    colframe=black!3,
    notitle,
    sharp corners,
    borderline west={1.5pt}{0pt}{R2color},
    enhanced,
    breakable,
    left=2pt,
    right=2pt,
    top=0pt,
    bottom=0pt,
    ignore nobreak,
    #1
}

\newtcolorbox{myquoteR3}[1][]{
    colback=black!3,
    colframe=black!3,
    notitle,
    sharp corners,
    borderline west={1.5pt}{0pt}{R3color},
    enhanced,
    breakable,
    left=2pt,
    right=2pt,
    top=0pt,
    bottom=0pt,
    ignore nobreak,
}

\newcounter{expmsgctr}
\newcommand{\expmessage}[1]{%
  \refstepcounter{expmsgctr}%
  \noindent\textbf{Observation:}~#1%
}

\usepackage{tcolorbox}
\definecolor{mycolor}{rgb}{0.122, 0.435, 0.698}

\begin{document}








\title[\textsf{NeurBench}: A Benchmark Suite for Learned Database Components with Drift Modeling]{\textsf{NeurBench}: A Benchmark Suite for Learned Database Components with Drift Modeling}









\author{
Zhanhao Zhao$^1$, Haotian Gao$^1$, Naili Xing$^1$, Lingze Zeng$^1$,  Meihui Zhang$^2$, \\ Gang Chen$^3$, Manuel Rigger$^1$, Beng Chin Ooi$^3$
}

\affiliation{
\fontsize{10}{10}\textit{$^1$ National University of Singapore}
\qquad\fontsize{10}{10}\textit{$^2$ Beijing Institute of Technology}
\qquad	\fontsize{10}{10}\textit{$^3$ Zhejiang University} \\
\fontsize{9}{9}{\{zhzhao, gaohaotian, xingnl, lingze, rigger\}@comp.nus.edu.sg} \qquad
\fontsize{9}{9}{meihui\_zhang@bit.edu.cn}  \qquad 
\fontsize{9}{9}{\{cg, ooibc\}@zju.edu.cn}
\country{}
}

\renewcommand{\shortauthors}{Zhanhao Zhao, Haotian Gao, Naili Xing, Lingze Zeng,  Meihui Zhang, Gang Chen, Manuel Rigger, Beng Chin Ooi}


\sloppy




\begin{abstract}
Learned database components, which deeply integrate machine learning into their design, have been extensively studied in recent years.
Given the dynamism of databases, where data and workloads continuously drift, it is crucial for learned database components to remain effective and efficient in the face of data and workload drift.
Robustness, therefore, is
a key factor in assessing their practical applicability.
Although recent works examine learned database components under specific drift, they fail to enable systematic performance evaluations across a broad range of drift or under customized drift as needed.
This paper presents \dbname, a new benchmark suite that supports evaluating learned database components under measurable and controllable data and workload drift.
We 
quantify diverse types of drift by introducing a 
key
concept called the drift factor.
Building on this formulation, we propose a drift-aware data and workload generation framework that effectively simulates real-world drift while preserving inherent correlations.
\vldbrevision{Experimental results demonstrate the effectiveness of \dbname in generating realistic data and workload drift, while providing insights into the performance of representative learned database components under different drift scenarios.}
\end{abstract}
\maketitle

\section{Introduction}
\label{sec:intro}

Database systems are increasingly embracing machine learning (ML) techniques in their design.
Key database components, such as query optimizers, indexes, and concurrency control (CC), are being upgraded with ML models to boost performance.
For example, learned query optimizers~\cite{DBLP:journals/pvldb/ZhuCDCPWZ23,DBLP:conf/sigmod/MarcusNMTAK21,DBLP:journals/pvldb/YuC0L22,DBLP:conf/sigmod/YangC0MLS22,DBLP:journals/pvldb/MarcusNMZAKPT19} treat query optimization as an ML problem to obtain more efficient query plans.
Learned indexes~\cite{DBLP:conf/sigmod/DingMYWDLZCGKLK20,DBLP:conf/sigmod/KraskaBCDP18,DBLP:journals/pvldb/WuZCCWX21,DBLP:conf/ppopp/TangWDHWWC20} leverage ML models to predict the positions of query keys, enabling faster key lookups than traditional indexes.
Similarly, learned CC algorithms~\cite{DBLP:conf/osdi/WangDWCW0021} predict optimal CC actions (\emph{e.g.}, locking strategies and blocking timeouts) for concurrent reads and writes to improve transaction throughput.

One major challenge faced by learned database components is the dynamic nature of databases, commonly referred to as \emph{data and workload drift}~\cite{DBLP:journals/pacmmod/KurmanjiTT24,DBLP:journals/pacmmod/KurmanjiT23,DBLP:journals/pacmmod/WuI24}.
Specifically, the data stored in databases constantly changes over time due to inserts, updates, and deletions, while workloads also evolve, with drift in query patterns~\cite{DBLP:journals/pvldb/NegiWKTMMKA23,DBLP:journals/pacmmod/WuI24} and query/transaction arrival rates~\cite{DBLP:conf/sigmod/LinTLCW21,DBLP:journals/pvldb/RenenL23,DBLP:journals/pvldb/RenenHPVDNLSKK24}.
In contrast, ML models typically derive insights from static datasets.
This fundamental difference between the paradigms of databases and ML can cause learned database components to be ineffective and outdated in the face of data and workload drift.
For instance, as data is updated over time, learned query optimizers may lose precision because the knowledge they initially learned for generating query plans becomes obsolete~\cite{DBLP:journals/pvldb/LehmannSS24,DBLP:journals/pvldb/LiWZDZ023}.
Likewise, learned CC algorithms may choose actions that were once optimal but no longer suit the current workload with drifted transaction arrival rates~\cite{DBLP:conf/osdi/WangDWCW0021}.

Achieving robust learned database components that remain effective regarding data and workload drift is crucial.
In tandem with the advancement of ML techniques, more learned database components are being designed with robustness in mind.
For example, learned query optimizers~\cite{DBLP:conf/sigmod/MarcusNMTAK21,DBLP:journals/pvldb/ZhuCDCPWZ23} adopt techniques such as periodic model retraining and fine-tuning to handle workload drift.
Learned indexes~\cite{DBLP:conf/sigmod/DingMYWDLZCGKLK20,DBLP:journals/pvldb/WuZCCWX21} have also become updatable by supporting structural modifications with model retraining in response to new data.
These rapid developments naturally raise a pertinent question: 
How well do 
existing learned database components
perform under data and workload drift, and to what extent does drift affect their performance?
Answering this question can provide valuable insights to guide the design of more robust learned database components.

\begin{vldbrevisionenv}
Thus far, several benchmarks have been proposed to examine learned database components under specific types of drift.
For example, existing evaluations of learned query optimizers typically model workload drift by training models on some queries and testing on other queries~\cite{DBLP:journals/pvldb/NegiWKTMMKA23,DBLP:journals/pvldb/LehmannSS24}, and model data drift by training models on one data snapshot and testing on other snapshots~\cite{DBLP:conf/sigmod/KimJSHCC22}.
Similarly, existing works~\cite{DBLP:journals/pvldb/SunZL23,DBLP:journals/pvldb/WongkhamLLZLW22} initialize learned indexes on a filtered subset of the data, and reinsert the excluded records at test time to model data drift.
Although effective, these drift treatments are case-specific and therefore only partially reveal how drift affects the performance of learned database components.
Some important insights, such as how their key design choices impact performance under drift, or the threshold at which learned components cease to be effective, remain implicit.
Completely addressing this requires new benchmarks that support comprehensive empirical evaluations across a broad range of data and workload drift.
\end{vldbrevisionenv}
Systematic benchmarks typically use a metric to quantify the variable of interest and adjust it during evaluations to, in principle, cover all possible scenarios.
For example, previous works~\cite{DBLP:conf/cloud/CooperSTRS10} leverage the well-defined skew factor to measure data skewness and assess system performance under varying skew levels.
Following this principle, enabling systematic benchmarks under drift has to address two key requirements:
First, designing a metric to quantify drift.
Given the various types of drift in databases, such as data drift, query pattern drift, \emph{etc.}, defining a unified and effective drift metric is not trivial.
\vldbrevision{Second, supporting meaningful evaluation under controlled drift.
In practice, real-world data and workload drift follow certain patterns, such as data drift often occurring with specific inherent correlation changes~\cite{DBLP:journals/pvldb/NegiWKTMMKA23}.
Controlling drift injection while preserving realistic drift patterns is therefore challenging.}


In this paper, we present \dbname, a novel benchmark suite designed to evaluate learned database components under controllable data and workload drift. 
We first introduce a concept called the {\em drift factor} to quantify both data and workload drift.
In particular, we model data and workloads as underlying probabilistic distributions and measure drift as the distance between these distributions.
\vldbrevision{We formally define the drift factor based on the distribution distance, {\em i.e., Jensen-Shannon (JS) divergence}~\cite{DBLP:books/daglib/0001548}, ensuring
a bounded value range.
Given that the performance of ML models is generally sensitive to the distance between training and testing distributions~\cite{DBLP:conf/iclr/GulrajaniL21,DBLP:journals/ml/Ben-DavidBCKPV10, DBLP:conf/icml/ZeighamiS24}, the drift factor allows us to define and control evaluation scenarios with varying degrees of drift between training and testing data and workloads, thereby enabling the comprehensive robustness evaluation for learned database components.}

Leveraging this concept, we then design \dbname to supply data and workload drift according to a specified drift factor.
Ideally, it would be desirable to adopt real-world data and workload drift directly.
Nonetheless, purely relying on real-world drift constrains the scope of drift that can be evaluated.
Real drift only reflects what has happened, and it cannot cover unseen or more severe drift that may occur in the future. 
Moreover, real data drift may be incomplete or inaccessible.
For example, the public IMDB dataset switched to only providing partial data after 2017~\cite{imdbold}, making real drift not fully observable.
To address these problems, we opt to synthesize data and workloads that capture realistic drift patterns while allowing effective control over the drift factor.
However, real-world data and workloads typically exhibit complex correlations.
For example, movie ratings and the number of votes are highly correlated in the IMDB dataset~\cite{imdb}.
Thus, we must capture the inherent correlations in the drift generation process to approximate real data and workload drift.

We therefore propose a drift-aware data and workload generation framework that enables the synthesis of realistic drift.
Specifically, we first utilize a base generative model, called diffuser, to fit the underlying distribution of the real data and workload.
\vldbrevision{We use the Denoising Diffusion Probabilistic Model (DDPM)~\cite{DBLP:conf/nips/HoJA20} as the base generative model.}
We then train an independent \emph{drifter} module to capture the real-world data and workload drift patterns.
Given a drift factor, the drifter guides the base model, \emph{i.e.}, the denoising process of DDPM, to synthesize data and workloads that approximate real drift with the desired severity.
Further, by decoupling the drifter from the base model, our framework avoids retraining the base model when adapting to new drift patterns, enabling generalization across various types of drift.



In summary, we make the following contributions:

\begin{itemize}[leftmargin=*]

\item We introduce {\dbname}, a practical benchmark suite that facilitates the evaluation of learned database components under controllable data and workload drift.
The suite will help further advance the design of next generation learned components.

\item We formalize a unified concept called the drift factor to model data and workload drift, enabling
the systematic performance evaluation of 
learned 
database components under diverse drift scenarios.
\dbname allows users to configure the settings to match their own environment and applications.

\item We propose a drift-aware data and workload generation framework that synthesizes data and workload drift 
according to a specified drift factor while preserving their inherent correlations.


\item \vldbrevision{We conduct extensive experiments to evaluate \dbname, 
and the results demonstrate the effectiveness of \dbname in generating realistic data and workload drift.}

\item We use {\dbname} to evaluate start-of-the-art learned query optimizers, learned indexes, and learned CC.
The results reveal how their design choices trade off performance under varying degrees and kinds of drift, offering insights into their robustness and generalizability.

\end{itemize}

\dbname is a benchmark suite designed to evaluate the efficiency and robustness of learned database components in NeurDB~\cite{neurdb-scis-24,neurdb2025cidr}.

The remainder of the paper is structured as follows.
Section~\ref{sec:background} provides relevant background.
Section~\ref{sec:design} formally defines the drift factor and describes the overall design of \dbname. 
Section~\ref{sec:modeling} presents the drift-aware data and workload generation framework in detail.
Section~\ref{sec:preparation} presents the implementation of \dbname.
Section~\ref{sec:preparation} describes the experimental setup, and Section~\ref{sec:evaluation} shows the results.
Section~\ref{sec:related} discusses the related work, and Section~\ref{sec:conclusion} concludes.

\section{Background}
\label{sec:background}

In this section, we provide the necessary background on learned database components, data and workload drift in databases, and diffusion models.
\extended{\subsection{Learned Database Components}}
Since learned database components typically rely on ML models, 
we provide a general definition of learned database components by mapping them to learnable ML functions:



\begin{definition}[Learned Database Component]
\label{def:learned_component}
Consider a database that stores data $\vb{R}$ and handles workload $\vb{Q}$.
A learned database component can be formulated as an ML function
$f(\cdot; \theta):\vb{Z} \to \vb{Y}$, where $\theta$ denotes the trainable parameters.
The input space $\vb{Z}$ consists of features derived from the data $\vb{R}$ and workload $\vb{Q}$, such as data distributions and query patterns.
The output space $\vb{Y}$ represents the predictions or decisions, such as estimated query plans, index positions, or concurrency control actions.
\extended{The parameters $\theta$ are trained by a stochastic learning algorithm $\mathcal{A}$ (\emph{e.g.}, Stochastic Gradient Descent), which adjusts $\theta$ to minimize a loss function $\mathcal{L}(f(\vb{Z}; \theta), \vb{Y})$ representing the discrepancy between the predicted output and the ground truth. 
The learning process can be formally expressed as:
\[
\theta = \mathcal{A}(\vb{Z}, \theta^r, \mathcal{L}),
\]
where $\theta^r$ represents the initial random weights of the model, and $\mathcal{A}$ iteratively updates $\theta$ based on the gradients computed \emph{w.r.t.} the loss function $\mathcal{L}$. }
\end{definition}






\extended{
A learned database component may continuously refine its model parameters $\theta$ based on feedback from actual query execution and data/workload drift, ensuring its output $\vb{Y}$ meets the following performance objective:

\begin{definition}[Performance Objective]
\label{def:performance_objective}
An optimal learned database component is expected to maximize system performance (\emph{e.g.}, minimizing latency, execution cost, or resource consumption) over a time window $\vb{T} = (ts_1, ts_2, \dots, ts_N)^{\top}$.
Let $\psi(\vb{R}^{(i)}, \vb{Q}^{(i)})$ denote the performance metric at time $ts_i$, where $\vb{R}^{(i)}$ and $\vb{Q}^{(i)}$ are the data and workload at time $ts_i$, respectively,
the performance objective is then to find optimal model parameters $\theta^*$, such that
\begin{center}
$\theta^* = \argmax_\theta \mathbb{E}_{ts_i \in \vb{T}}[ \psi(\vb{R}^{(i)}, \vb{Q}^{(i)})]$,
\end{center}
\noindent where $\mathbb{E}[\cdot]$ denotes the 
expected performance.
\end{definition}
}

\subsection{Data and Workload Drift}
\label{sec:background:drift}

Data stored in databases, as well as the workloads they handle, are subject to continuous drift.
{\em Data drift} occurs due to operations including inserts, updates, deletes, and schema modifications.
{\em Workload drift} refers to variations in the queries and transactions processed by the database over time, which can be further categorized into, 1) Query pattern drift: changes in query structure, such as join patterns (\texttt{FROM} clauses) and predicates (\texttt{WHERE} clauses);
2) Arrival rate drift: fluctuations in the volume and concurrency of queries and transactions.



The effectiveness of ML models is generally related to the distance between training and testing distributions~\cite{DBLP:conf/iclr/GulrajaniL21,DBLP:journals/ml/Ben-DavidBCKPV10, DBLP:conf/icml/ZeighamiS24}, and therefore, we consider measuring the data and workload drift based on the distance of their underlying distributions.
Drift may occur in the data, the workload, or both, depending on whether the distribution of $\vb{R}$, $\vb{Q}$, or their joint distribution drifts over time.

\begin{figure}[t]
\centering
\includegraphics[width=0.93\linewidth]{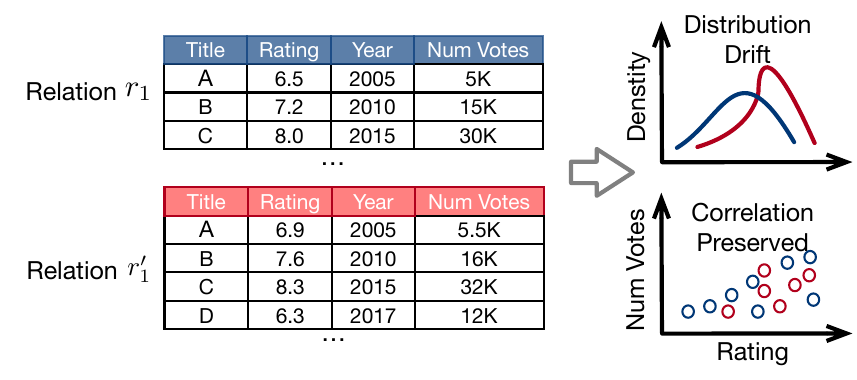}
\vspace{-4mm}
\caption{\vldbrevision{Example of Real-World Data Drift}}
\label{fig:bg:drift_example}
\vspace{-4mm}
\end{figure}

\extended{
\begin{figure}[t]
\centering
\subfigure[Distribution]{
\includegraphics[width=0.7\linewidth]{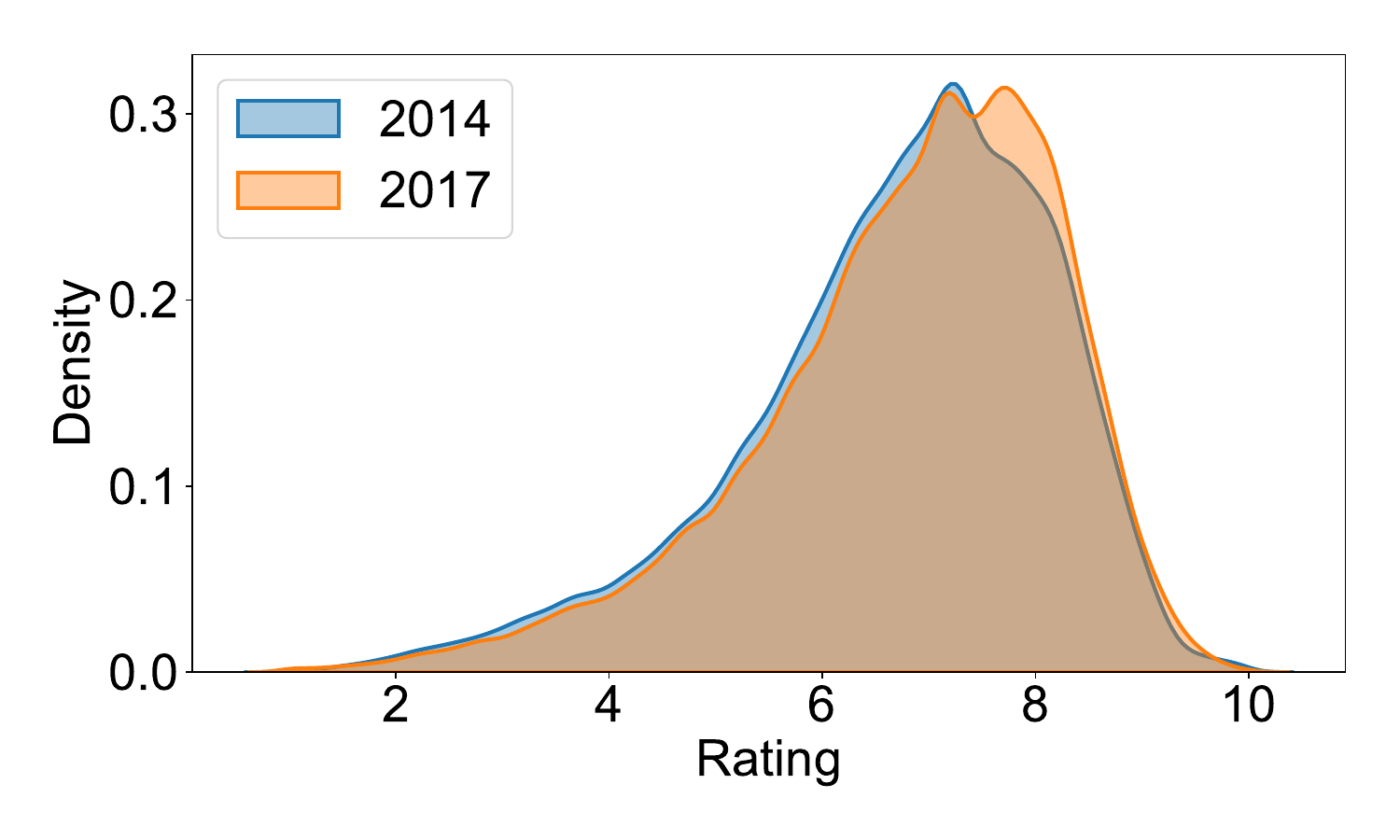}
}
\subfigure[Correlation]{
\includegraphics[width=0.7\linewidth]{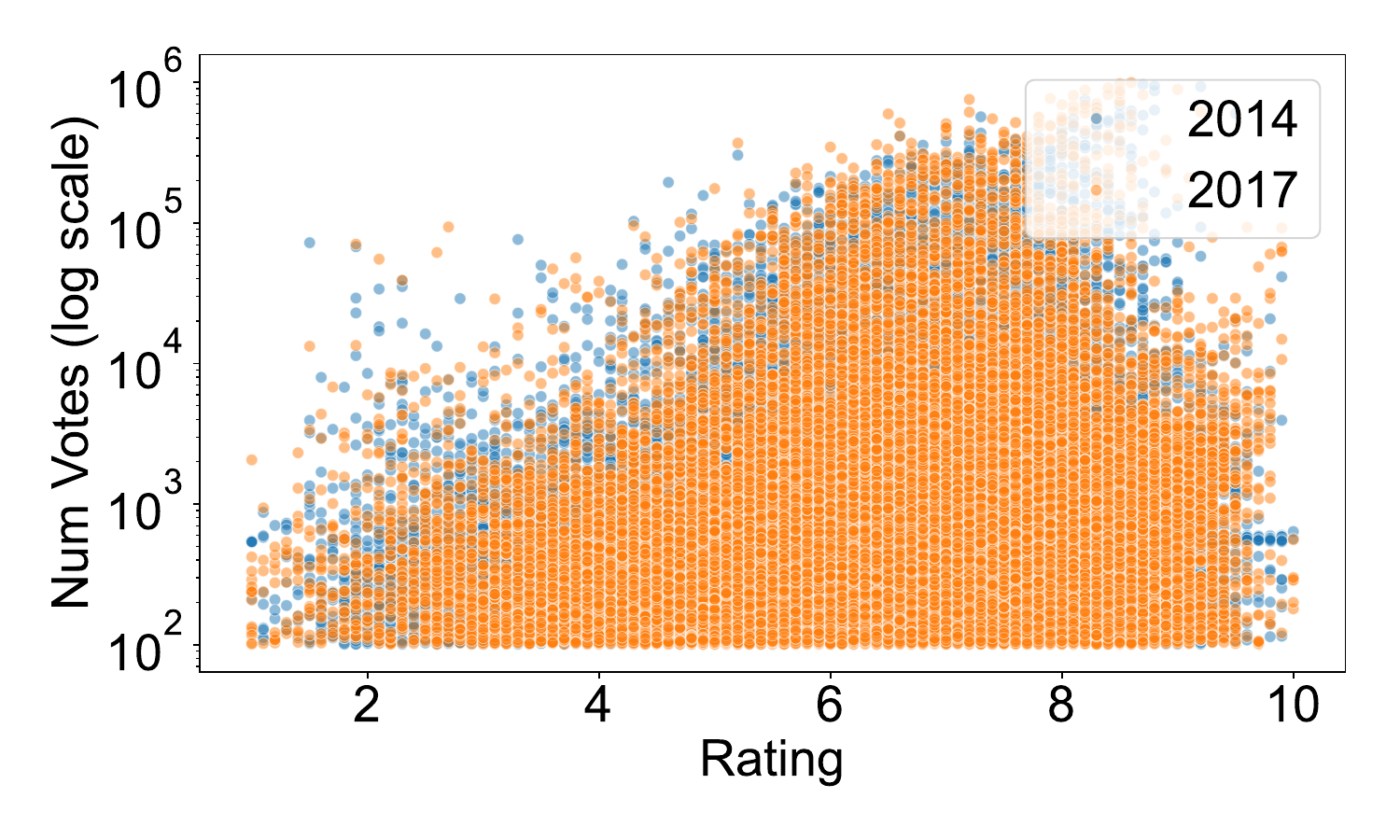}
}
\vspace{-4mm}
\caption{\vldbrevision{Statistics of Real-World Data Drift}}
\label{fig:bg:drift_example_statistics}
\vspace{-2mm}
\end{figure}
}

\begin{definition}[Data and Workload Drift]
\label{def:bg_drift}
Consider a database that operates over a time window $\vb{T}$ consisting of $N$ time slots, \emph{i.e.}, $\vb{T} = (ts_1, ts_2, \dots, ts_N)^{\top}$, where $N \geq 2$.
Let $\vb{R}^{(i)}$ and $\vb{Q}^{(i)}$ be random variables representing the state of the data and the workload, respectively, at time $ts_i$.
Data and workload drift occur when the distributions of data $\vb{R}$ and workload $\vb{Q}$ exhibit temporal changes across the time window $\vb{T}$. 
Let $\mathbb{P}_{\vb{R}^{(i)}, \vb{Q}^{(i)}} (\vb{r}, \vb{q})$ represent the joint probability density function of data $\vb{R}$ and workload $\vb{Q}$ at time $t_i$. 
Drift happens if, for any two timestamps $i, j \in N$, where $i \neq j$,
\[
\mathbb{P}_{\vb{R}^{(i)},\vb{Q}^{(i)}} (\vb{r}, \vb{q}) \neq \mathbb{P}_{\vb{R}^{(j)},\vb{Q}^{(j)}} (\vb{r}, \vb{q}).
\] 
\end{definition}
%
%




Constructing realistic data and workload drift that captures meaningful distribution drift is crucial.
\begin{vldbrevisionenv}
As shown in Figure~\ref{fig:bg:drift_example}, we illustrate using movie ratings and their associated number of ratings (num\_rates) as an example.
Let the relation $r_1$ denote a data snapshot from 2013 and $r_1’$ from 2017.
First, while the schema remains unchanged, the underlying data distribution can drift over time as new data arrives~\cite{DBLP:journals/pvldb/NegiWKTMMKA23,DBLP:journals/pvldb/RenenHPVDNLSKK24}.
Second, inherent data correlations are largely preserved. 
Real-world data drift typically maintains the original correlation structure and distributional shape, with shifts mainly in density.
Here, the movie rating is consistently positively correlated with the num\_rates.
For example, movie ratings and num\_rates consistently exhibit a positive correlation as movies get added. 
Other attribute pairs, such as country and language, or release\_year and genre, also tend to remain highly correlated across time, as observed in prior studies~\cite{DBLP:journals/pvldb/NegiWKTMMKA23, DBLP:journals/pvldb/LehmannSS24}.
Similarly, workload drift typically manifests as changes in the frequency or usage of query templates and predicates, while still preserving structural elements such as join patterns or access paths~\cite{DBLP:journals/pvldb/RenenL23,DBLP:journals/pvldb/RenenHPVDNLSKK24}.
These insights form the foundation for the drift modeling framework in \dbname.
\end{vldbrevisionenv}

In our design, we learn the underlying distribution drift from real data and workloads, and propose a generation framework that enables the controllable synthesis of data and workload drift.


\subsection{Diffusion Model}
\label{bg:diffusion}


Diffusion models are powerful generative models that have been recently adopted to synthesize high-quality data, such as images~\cite{DBLP:conf/nips/DhariwalN21}, audio~\cite{DBLP:conf/icml/LiuCYMLM0P23}, and tabular data~\cite{DBLP:journals/pacmmod/LiuFTLD24,DBLP:conf/icml/KotelnikovBRB23}.
Inspired by the physical process of diffusion, where information (\emph{e.g.}, particles, data) is gradually dispersed over time, 
these models learn to reverse this process so that a target distribution can be generated from a given prior distribution.
In particular, diffusion models operate in two key stages, \emph{i.e.}, forward process and reverse process, with slight variations in the details depending on the specific model.
For illustration purposes, we use the Denoising Diffusion Probabilistic Model (DDPM) as a representative to formally explain these stages:

\paragraphtitle{Diffusion (Forward) Process}
This process incrementally corrupts the data by adding Gaussian noise in discrete timesteps, ultimately transforming it into a simple noise distribution. 
Given a data sample \( \mathbf{x}_0 \sim q(\mathbf{x}) \) from the true data distribution \( q(\mathbf{x}) \), the forward process is formulated as a Markov chain:
\[
q(\mathbf{x}_t | \mathbf{x}_{t-1}) = \mathcal{N}(\mathbf{x}_t; \sqrt{1-\beta_t} \mathbf{x}_{t-1}, \beta_t \mathbf{I}),
\]
where \( \beta_t \) represents the variance schedule controlling the noise added at each step \( t \). 
The data distribution after \( T \) timesteps, \( q(\mathbf{x}_T) \), approximates an isotropic Gaussian distribution \( \mathcal{N}(\mathbf{0}, \mathbf{I}) \).

\paragraphtitle{Denoising (Reverse) Process}
This process reconstructs the target data $\mathbf{x}_0$ by iteratively denoising  $\mathbf{x}_t$ in reverse timesteps, starting from $\mathbf{x}_T \sim \mathcal{N}(\mathbf{0}, \mathbf{I})$.
The reverse process $p(\mathbf{x}_{t-1} | \mathbf{x}_{t})$ can be approximated as a Gaussian distribution:
\[
p(\mathbf{x}_{t-1} | \mathbf{x}_t) = \mathcal{N}(\mathbf{x}_{t-1}; \boldsymbol{\mu}(\mathbf{x}_t, t; \theta), \sigma_t^2 \mathbf{I}),
\]
where \( \boldsymbol{\mu} \) is learned parameters that estimate the mean of the reverse process, while $\sigma_t^2$ is derived from  $\beta_t$.

\paragraphtitle{Guided Diffusion}
The reverse process can incorporate auxiliary information $c$, yielding a conditional reverse process $p(\mathbf{x}_{t-1} | \mathbf{x}_t, c)$,
where the conditioning variable $c$ may represent factors such as correlation-preserving constraints, or contextual information~\cite{DBLP:conf/nips/DhariwalN21,DBLP:journals/pacmmod/LiuFTLD24}.
By applying the Bayes' theorem, $p(\mathbf{x}_{t-1} | \mathbf{x}_t, c)$ can be decomposed as:
\[
p(\mathbf{x}_{t-1} | \mathbf{x}_t, c) \propto p(\mathbf{x}_{t-1} | \mathbf{x}_t) \textrm{Pr}(c |\mathbf{x}_{t-1}, \mathbf{x}_{t}).
\]
In particular, $p(\mathbf{x}_{t-1} | \mathbf{x}_t)$ should perform unconditional generation, while
$\textrm{Pr}(c |\mathbf{x}_{t-1}, \mathbf{x}_{t})$ guides the reverse process to introduce the specified condition $c$.

Inspired by such guided diffusion models, we treat the introduction of drift to data and workloads as a specific condition, and propose a novel drift-aware data and workload generation framework based on DDPM.
Our generation process differs from prior work~\cite{DBLP:journals/sigmod/RGPW24} in that we directly inject drift signals into the denoising steps.
This allows \dbname to generate drifted data controlled by a drift factor with correlation retained.

\section{Overview of \dbname} \label{sec:design}

In this section, we define the drift factor and then present the system overview of \dbname.

\subsection{Drift Factor Modeling}

We define a drift factor that can uniformly and effectively describe the data and workload drift.
In our problem setting, the drift factor must satisfy two key properties:
1) It should effectively capture the differences between data or workloads while providing a bounded value range for users to specify the desired drift. 
2) Its value should provide a consistent representation for all types of drift.

According to Definition~\ref{def:bg_drift}, where drift is mapped to the underlying distribution of data and workloads, we define the drift factor based on the distance of distributions.
To satisfy the properties above, we utilize the {\em Jensen-Shannon (JS) divergence}~\cite{DBLP:books/daglib/0001548}.
Specifically, given a distribution $\mathbb{P}$ and the drifted one $\mathbb{P}'$, their JS divergence $\mathcal{D}_{JS}(\mathbb{P} \parallel \mathbb{P}')$ is defined as:
%
\begin{align*}
\mathcal{D}_{JS}(\mathbb{P} \parallel \mathbb{P}')
= \frac{1}{2} \mathcal{D}_{KL}(\mathbb{P} \parallel M) + \frac{1}{2} \mathcal{D}_{KL}(\mathbb{P}' \parallel \mathbb{M}),
\end{align*}
\noindent where $\mathbb{M} = \frac{1}{2} (\mathbb{P} + \mathbb{P}')$, and $D_{KL}$ is the Kullback–Leibler (KL) divergence, which is defined by:
\[
\mathcal{D}_{KL}(\mathbb{P} \parallel \mathbb{P}') = \sum_x \mathbb{P}(x) \log \frac{\mathbb{P}(x)}{\mathbb{P}'(x)}.
\]
$\mathcal{D}_{JS}$ satisfies $\mathcal{D}_{JS}(\mathbb{P} \parallel \mathbb{P}') = \mathcal{D}_{JS}(\mathbb{P}' \parallel \mathbb{P})$, and is bounded within $[0, \log 2]$.
Normalizing it by $\log 2$, the drift factor $d$ is defined as:
\[
d(\mathbb{P}, \mathbb{P}') = \frac{\mathcal{D}_{JS}(\mathbb{P} \parallel \mathbb{P}')}{\log 2}, \quad d \in [0, 1].
\]

We allow users to specify the desired drift level using a {\em drift factor} $d$, where $d\in[0,1]$.
By default, users can independently indicate drift factors for either data or workload to achieve effective control over each.
A drift factor of $d$=0 indicates no drift, while $d$=1 denotes that the entire data or workload is fully drifted.



\subsection{System Overview}

\begin{vldbrevisionenv}
As shown in Figure~\ref{fig:overview}, we design \dbname with two key modules, a \emph{drift-aware data and workload generator}, and a \emph{performance evaluator}.
In particular, the {\em generator}, based on guided diffusion, consists of:
1) the \emph{diffuser}, a DDPM-based generative model that learns the underlying distribution of the real data and workloads;
2) the \emph{drifter}, a classification model that guides the diffuser to generate data and workloads with specified drift.
The evaluator is responsible for assessing the performance of learned database components under drift.
We organize the overall workflow of \dbname into three stages: training, drift-aware generation, and evaluation.
In the following, we focus on how to use \dbname to conduct experiments under data drift for illustration purposes.
\end{vldbrevisionenv}

\begin{figure}[t]
\centering
\includegraphics[width=0.86\linewidth]{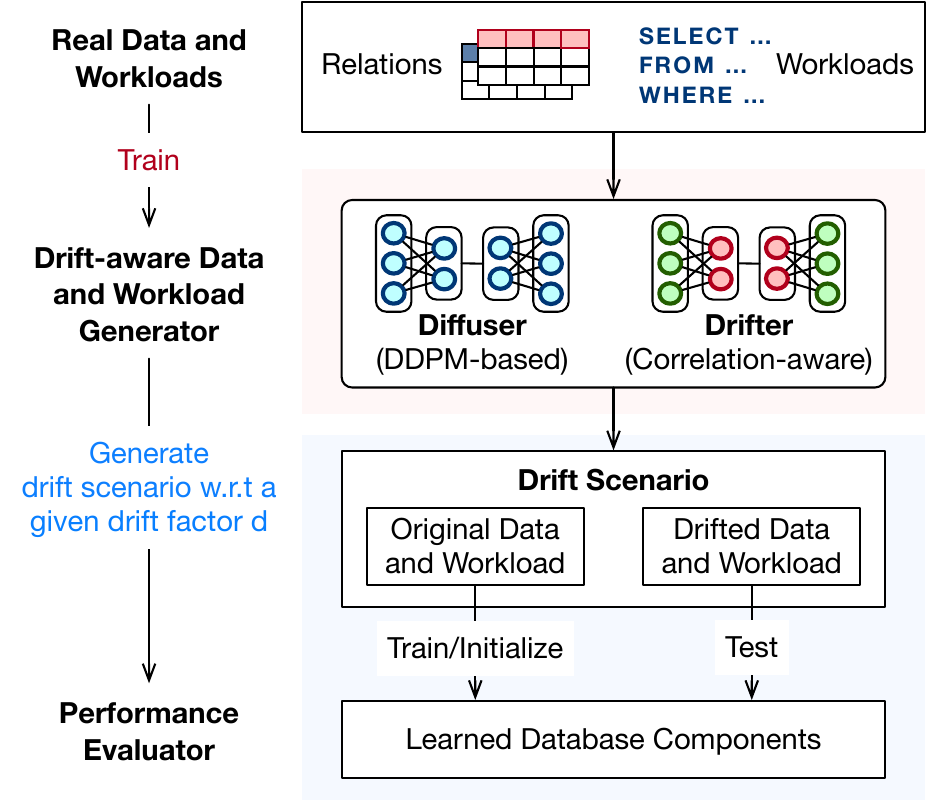}
\vspace{-2mm}
\caption{\vldbrevision{System Overview of \dbname}}
\label{fig:overview}
\vspace{-4mm}
\end{figure}

\begin{vldbrevisionenv}
\paragraphtitle{Diffuser and Drifter Training} 
We input real data and workloads to train the drift-aware data and workload generator, with the diffuser and drifter trained separately.
The diffuser is first trained on a fixed data snapshot (\emph{e.g.}, 2013 data) to understand the underlying distribution $\mathbb{P}_{\text{base}}$, following the DDPM framework.
The drifter is then trained to model real drift patterns using a sequence of real data snapshots (\emph{e.g.}, 2013–2016).
In particular, the drifter maps real drift patterns to modulation vectors that condition the diffuser’s denoising process, enabling the diffuser to generate data that reflects realistic drift while preserving key correlations.
By decoupling the training of the diffuser and drifter, our framework allows the diffuser to remain fixed while adapting to new types of drift by simply learning new drifter mappings, thereby improving generality and reusability.
We will detail the diffuser and drifter training in Section~\ref{sec:gen_construction}.
\end{vldbrevisionenv}

\paragraphtitle{Drift-aware Data and Workload Generation} 
Given a specified drift factor $d$ and original data as input, we use the generator to drift the given data to a target data that aligns with $d$.
In particular, starting from a standard normal distribution, the generator iteratively invokes the diffuser to denoise the distribution, while leveraging the drifter to introduce controlled drift at each denoising step, gradually transforming the distribution into the target distribution.
Section~\ref{sec:ctrl-drift-gen} details how to use the trained diffuser and drifter to generate required data and workload drift.





\begin{vldbrevisionenv}
\paragraphtitle{Performance Evaluation} 
We assess the robustness of learned database components under data and workload drift using the evaluator. 
For a given drift factor $d$, we first use the generator to produce a pair of data: the base and its drifted counterpart.
A target component is then trained on the original data and workload and tested on the drifted data and workload.
By repeating this procedure across a spectrum of drift factors, we can systematically evaluate performance degradation to varying degrees of drift.
\end{vldbrevisionenv}
\section{Controllable Drift Synthesis Framework}
\label{sec:modeling}
In this section, we detail the design of the drift-aware data and workload generation framework.

\subsection{Mapping Data and Workload Drift to Distribution Drift}
\label{subsec:drift_factor}

Here we describe how we map both data and workload drift to the corresponding distribution drift in detail.
For data, we consider a relational database $\vb{R}$ containing multiple tables, each comprising a set of attributes.
The distribution of $\vb{R}$, $\mathbb{P}_{\vb{R}}$, can be approximated as the product 
of its marginal attribute distributions, \emph{i.e.},
$\mathbb{P}_{\vb{R}}(A_1, A_2, \dots, A_m) \approx \prod_{i=1}^m \mathbb{P}(A_i)$,
where $\mathbb{P}(A_i)$ represents the distribution of the $i$-th attribute.
Note that this approximation is used solely for distribution measurement. 
During data generation, we jointly synthesize multiple columns per tuple to preserve attribute correlations, as discussed in Section~\ref{sec:ctrl-drift-gen}.
For each attribute, we employ specific strategies based on its data type to approximate its distribution.
Categorical attributes (\emph{e.g.}, \texttt{CHAR}, \texttt{VARCHAR}) and discrete numerical values are handled by computing distributions over all observed values in the table.
Continuous numerical attributes
are approximated by discretizing values into a fixed number of bins. 
The corresponding distribution is then estimated by computing the probability density of each bin.
This binning strategy is commonly employed in database histograms to estimate attribute distributions and query selectivity.

For a workload consisting of $N$ queries $\vb{Q} = \{\vb{q}_1, \vb{q}_2, \dots, \vb{q}_N\}$, we define the workload distribution 
$\mathbb{P}_\vb{Q}$ as the distribution over {\em query templates}.
Specific features, such as query patterns, query arrival rates, \emph{etc.}, characterize each query template.
To model query patterns and arrival rates in a unified manner, we define a workload table with 3 attributes: Join Pattern, Predicate, and Query Interval.
The Join Pattern and Predicate columns capture the query pattern, representing table join patterns (\emph{e.g.}, $r_1$ $\bowtie$ $r_2$), and predicates (\emph{e.g.}, Bal = `10K'). 
The Query Interval column quantifies the arrival rate by recording the time gap between consecutive queries.
We control the distributions of the workload table to synthesize workloads with query pattern drift, arrival rate drift, or both, while preserving the inherent correlations.
Note that the workload table used here is for illustration purposes; it can be extended with additional attributes for more complex queries.

\subsection{Theoretical Analysis}
\label{sec:theoretical}

To generate data and workload drift that satisfy a specified $d$, we provide a formal problem definition and then theoretically explain how we enable the drift-aware data and workload generation.

\subsubsection{Problem Definition}
\label{sec:problem_definition}

We define a database running in a time window $\vb{T}$ with $N (N \geq 2)$ time slots, \emph{i.e.}, $\vb{T} = (ts_1, ts_2, ..., ts_N)^{\top}$.
Recall that we denote $\vb{R}$ and $\vb{Q}$ as random variables of the data and workload in the database, respectively, and we map them to their underlying distributions $\mathbb{P}_\vb{R}^{(i)}$ and $\mathbb{P}_\vb{Q}^{(i)}$ at $ts_i$.
Formally, we define the drift-aware data and workload generation below.

\begin{definition}[Drift-aware Data and Workload Generation] \label{def:problem_definition}
Given a series of drift factors $\vb{D} = (d_1, d_2, ..., d_N)^{\top}$, where each drift factor $d_i$ quantifies the degree of change in the distribution between adjacent time slots $(ts_j, ts_i)$.
The distribution of either data or workload at time $ts_j$ is obtained by transforming the distribution at time $ts_i$ with a function $g$, according to the drift factor $d_i$:
\begin{gather*}
\mathbb{P}_\vb{R}^{(j)} = g_{\vb{R}}(\mathbb{P}_\vb{R}^{(i)}, d_i, \mathcal{C}), \quad \mathbb{P}_\vb{Q}^{(j)} = g_{\vb{Q}}(\mathbb{P}_\vb{Q}^{(i)}, d_i, \mathcal{C}),
\end{gather*}
where $\mathcal{C}$ is the corresponding inherent correlations, and $g_{\vb{Z}}: \Omega_{\vb{Z}} \times [0, 1] \times \mathcal{C} \rightarrow \Omega_\vb{Z}$ is the drift function applied on input space $\vb{Z}$, transforming its distribution on the same domain.
\end{definition}

Given a distribution $\mathbb{P}^{(i)}$ and the drift factor $d_i$, 
the generated new distribution $\mathbb{P}^{(j)}$ should
1) exhibit drift quantified by $d$ from $\mathbb{P}^{(i)}$; and 2) preserve inherent correlations $\mathcal{C}$.
The quality of $\mathbb{P}^{(j)}$ degrades as either requirement is violated.
Mathematically, this is equivalent to finding a hypersurface in $\Omega_{\vb{Z}}$, the domain of input space $\vb{Z}$ (\emph{i.e.}, either data and workload), which satisfies the equation:
\begin{equation}
\label{eq:find-d}
d(\mathbb{P}_{i}, \mathbb{P}_{j}) \approx d^* \quad \text{subject to} \quad \mathcal{C}(\mathbb{P}_{j}) \approx \mathcal{C}(\mathbb{P}_i).
\end{equation}
\noindent However, there are infinite solutions to Equation~\ref{eq:find-d}, which makes it difficult to interpret and control.
To solve this, we limit the expressiveness of $\mathbb{P}_{j}$ through a DDPM-based generator.

\subsubsection{Theoretical Proof}
Given a desired drift factor $d$, our objective as outlined in Section~\ref{sec:problem_definition} is to generate the drifted data $\mathbf{x}_{\text{drift}}$ by sampling from $p(\mathbf{x}_{t-1} | \mathbf{x}_t, d)$, \emph{i.e.}, a conditional DDPM reverse process that incorporates the desired drift $d$.
As outlined in Section~\ref{bg:diffusion},
$p(\mathbf{x}_{t-1} | \mathbf{x}_t, d)$ can be decomposed as:
\begin{equation}
\label{eq:drift-bayes}
p(\mathbf{x}_{t-1} | \mathbf{x}_t, d) \propto p(\mathbf{x}_{t-1} | \mathbf{x}_t) \textrm{Pr}(d | \mathbf{x}_{t-1}, \mathbf{x}_t),
\end{equation}
where $p(\mathbf{x}_{t-1} | \mathbf{x}_t)$ can be modeled by the denoising process of DDPM, \emph{i.e.}, $p(\mathbf{x}_{t-1} | \mathbf{x}_t) = \mathcal{N}(\mathbf{x}_{t-1}; \boldsymbol{\mu}(\mathbf{x}_t, t; \theta), \sigma_t^2 \mathbf{I})$,
and $\textrm{Pr}(d | \mathbf{x}_{t-1}, \mathbf{x}_t)$ guides the process to introduce the specified drift $d$.
Since the drifted data $\mathbf{x}_{drift}$ is the target output, Equation~\eqref{eq:drift-bayes} can be rewritten as:
\begin{equation}
\label{eq:drift-x}
p(\mathbf{x}_{t-1} | \mathbf{x}_t, \mathbf{x}_{\text{drift}}) 
    \propto 
    p(\mathbf{x}_{t-1} | \mathbf{x}_t) 
    \textrm{Pr}(\mathbf{x}_{\text{drift}} | \mathbf{x}_{t-1}, \mathbf{x}_t).
\end{equation}
\noindent Note that the mean of the Gaussian distribution, $\boldsymbol{\mu}(\mathbf{x}_t, t; \theta)$, determines the trajectory of the denoising process~\cite{DBLP:journals/pacmmod/LiuFTLD24,DBLP:conf/nips/DhariwalN21}.
We therefore leverage $\textrm{Pr}(\mathbf{x}_{\text{drift}} | \mathbf{x}_{t-1}, \mathbf{x}_t)$ 
to introduce a controlled drift
to the mean $\boldsymbol{\mu}$.
%
\begin{theorem}
\label{th:drifter}
The controlled reverse distribution $p(\mathbf{x}_{t-1} | \mathbf{x}_t, \mathbf{x}_{\text{drift}})$ is able to be approximated using $\mathrm{Pr}(\mathbf{x}_{\text{drift}} | \mathbf{x}_t)$ and its gradient at $\mathbf{x}_t$, or $\nabla_{\mathbf{x}_t} \mathrm{Pr}(\mathbf{x}_{\text{drift}} | \mathbf{x}_t)$.
\end{theorem}
\begin{proof}
For Equation~\eqref{eq:drift-x}, the first part is approximated by a Gaussian model, while the second part can be simplified using conditional independence, which transforms the equation to:
\begin{equation*}
p(\mathbf{x}_{t-1} | \mathbf{x}_t, \mathbf{x}_{\text{drift}}) 
    \approx 
    z 
    \mathcal{N}(\mathbf{x}_t; \boldsymbol{\mu}, \mathbf{\Sigma}) 
    \textrm{Pr}(\mathbf{x}_{\text{drift}} | \mathbf{x}_t),
\end{equation*}
\noindent where $z$ is the normalizing factor.
%
Next, as proven by Dickstein et al.~\cite{DBLP:conf/icml/Sohl-DicksteinW15}, we can perform Taylor expansion around $\boldsymbol{\mu}$ to further approximate $p(\mathbf{x}_{t-1} | \mathbf{x}_t, \mathbf{x}_{\text{drift}})$ so that a perturbed Gaussian distribution can be used to model it:
\begin{align*}
\log p(\mathbf{x}_{t-1} | \mathbf{x}_t, \mathbf{x}_{\text{drift}}) 
    \approx 
    \log \mathcal{N}(\mathbf{x}_t; \boldsymbol{\mu}, \mathbf{\Sigma}) 
    + \log \textrm{Pr}(\mathbf{x}_{\text{drift}} | \mathbf{x}_t) \\
\approx -\frac{1}{2} 
    (\mathbf{x}_t - \boldsymbol{\mu})^{\top} 
    \mathbf{\Sigma}^{-1} 
    (\mathbf{x}_t - \boldsymbol{\mu})
    +
    (\mathbf{x}_t - \boldsymbol{\mu}) \cdot g
    +
    C\\
= -\frac{1}{2} 
    (\mathbf{x}_t - \boldsymbol{\mu} - \mathbf{\Sigma} \cdot g)^{\top}
    \mathbf{\Sigma}^{-1}
    (\mathbf{x}_t - \boldsymbol{\mu} - \mathbf{\Sigma} \cdot g)
    +
    C\\
= \log \mathcal{N}(\mathbf{x}_t; \boldsymbol{\mu} + \mathbf{\Sigma} \cdot g, \mathbf{\Sigma}),
\end{align*}
\noindent where $g = \nabla_{\mathbf{x}_t} \mathrm{Pr}(\mathbf{x}_{\text{drift}} | \mathbf{x}_t)$, and the constant term $C$ can be ignored.
In other words, $p(\mathbf{x}_{t-1} | \mathbf{x}_t, \mathbf{x}_{\text{drift}})$ can be approximated by $\mathcal{N}(\mathbf{x}_t; \boldsymbol{\mu} + \mathbf{\Sigma} \cdot \nabla_{\mathbf{x}_t} \mathrm{Pr}(\mathbf{x}_{\text{drift}} | \mathbf{x}_t), \mathbf{\Sigma})$.
\end{proof}
Conceptually, Theorem~\ref{th:drifter} ensures the generated data moves toward the target drifted distribution.
Unlike existing data generators~\cite{DBLP:journals/sigmod/RGPW24}, our method directly injects drift signals into the denoising steps, and with the ability of the DDPM to preserve correlations in the latent space, we enable \dbname to generate drifted data controlled by a drift factor with correlation retained.


\subsection{Diffuser and Drifter Training}
\label{sec:gen_construction}
Based on Theorem~\ref{th:drifter}, we design a generation framework comprising two key components: the diffuser and the drifter.
Concretely, the Diffuser learns the reverse denoising distribution to map $\mathcal{N}(0,I)$ to data, while the Drifter conditions each reverse step to match a specified drift factor.
We illustrate the training process of the diffuser and drifter in Figure~\ref{fig:training}.


\begin{figure}
\centering
\includegraphics[width=0.99\linewidth]{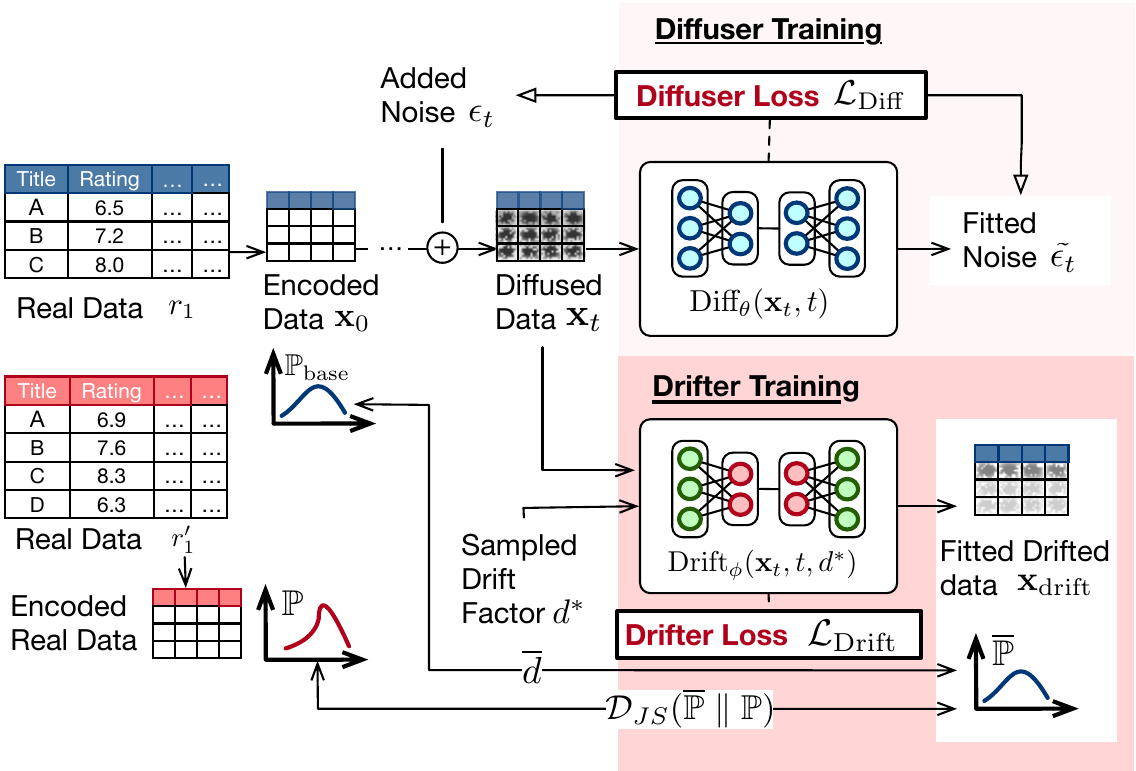}
\vspace{-2mm}
\caption{\vldbrevision{Diffuser and Drifter Training}}
\label{fig:training}
\vspace{-4mm}
\end{figure}

\paragraphtitle{Diffuser Training}
The diffuser \texttt{Diff} serves as the core generative model, learning to approximate the unconditional distribution $p(\mathbf{x}_{t-1} | \mathbf{x}_t)$.
To achieve this, we train \texttt{Diff} to predict the mean  $\boldsymbol{\mu}(\mathbf{x}_t, t; \theta)$ of the reverse denoising process, which is directly related to the noise added during the forward process.
In particular, the reverse process can be rewritten in terms of noise $\epsilon_t$, where:
$\boldsymbol{\mu}(\mathbf{x}_t, t; \theta) = \frac{1}{\sqrt{\alpha_t}} \left( \mathbf{x}_t - \sqrt{1 - \alpha_t} \cdot \epsilon_t \right),
$ and $\alpha_t$ is a parameter derived from the variance schedule $\sigma_t$.
This formulation makes noise reconstruction a central task in training, and thus, 
the noise prediction output of the diffuser at timestep $t$, or \texttt{Diff}$(\mathbf{x}_t, t)$, is trained by minimizing the noise reconstruction error:
\[
\mathcal{L}_{\texttt{Diff}} = \mathbb{E}_{t} \left[ \lVert \epsilon_t - \texttt{Diff}(\mathbf{x}_t, t) \rVert^2 \right].
\]

As shown in Figure~\ref{fig:training}, the diffuser is trained to reconstruct samples from the base data with the distribution $\mathbb{P}_{\text{base}}$ by learning to predict the added noise during the forward diffusion process. 
In each training iteration, a clean encoded sample $\mathbf{x}_0$ is drawn from real data $r_1 \sim \mathbb{P}_{\text{base}}$.
A timestep $t$ and Gaussian noise $\epsilon_t$ are sampled, and noise is added to $\mathbf{x}_0$ to produce the diffused sample $\mathbf{x}_t$. The diffuser \texttt{Diff} takes $\mathbf{x}_t$ and $t$ as inputs, and predicts the noise $\hat{\epsilon}_t$ used in denoising.
We train \texttt{Diff} by minimizing the mean squared error between the predicted noise $\hat{\epsilon}t$ and the true noise $\epsilon_t$, denoted as $\mathcal{L}_{{\texttt{Diff}}}$. 
This enables the diffuser to capture the underlying distribution and latent correlations in $\mathbb{P}_{\text{base}}$.

\paragraphtitle{Drifter training}
The drifter \texttt{Drifter} is used to fit $\textrm{Pr}(\vb{x}_{\text{drift}}|\vb{x}_t)$ to generate $\vb{x}_{\text{drift}}$ with the same size as the current noisy data $\vb{x}_t$, and satisfies any drift factor $d$, \emph{i.e.}, $\vb{x}_{\text{drift}} = \texttt{Drifter}(\mathbf{x}_t, t, d)$.
After that, the expected drift factor $d^*$ is input into \texttt{Drifter} to compute $\nabla_{\mathbf{x}_t} \mathrm{Pr}(\mathbf{x}_{\text{drift}} | \mathbf{x}_t)$ and guide \texttt{Diff} toward the expected distribution.
The drifter is trained to guide the denoising process toward drifted target distributions while preserving key statistical properties of the base data. 
To achieve this, we adopt a composite loss function so that it balances among the accuracy of drift-introduction, preservation of original correlations, and coherence with the real data.


The detailed algorithm to train the drifter is illustrated in Algorithm~\ref{alg:drifter_training}.
It first initializes drifter parameters $\phi$, pre-computes $\alpha_t$ and $\bar{\alpha_t}$ based on $\{\beta_t\}_{t=1}^{T}$, where both the definition of $\{\beta_t\}$ and the computation of $\alpha_t$ and $\bar{\alpha_t}$ (Lines 1-2). 
After that, it starts the training by looping through the training steps.
At each training step, it first samples a clean base sample $x_0$ from $\mathbb{P}_{\text{base}}$ and selects a target real data $\mathbb{P}$ from $\{\mathbb{P}_{\text{target}}\}$.
After that, a random drift factor $d^*$, a timestep $t$, and a noise $\epsilon \sim \mathcal{N}(0, I)$ are sampled (Lines 4-5).
The drifter then constructs a noisy sample $\vb{x}_t$ via the forward diffusion process and applies \text{Drift} to denoise, retrieving the drifted sample $\vb{x}_{\text{drift}}$, and computes its distribution $\overline{\mathbb{P}}$ (Lines 6-8).
The drift factor between distributions of the base and synthesized data $\bar{d}$ is computed afterward (Line 9).
The loss value $\mathcal{L}_{\text{Drift}}$ is constructed by combining three factors (Line 10).
First, a drift loss between $\bar{d}$ and the real drift factor $d^*$ is measured.
Second, a correlation loss is also computed and adjusted by the hyperparameter $\lambda_1$ to preserve statistical dependencies in the base data, such as the strong primary-foreign key correlation between two tables.
Third, a divergence loss between the synthesized and real drifted sample is included and adjusted by $\lambda_2$ to shift the distribution toward real drift.
Finally, $\mathcal{L}_{\text{Drift}}$ is used to update \texttt{Drifter} via gradient descent.

\begin{algorithm}[t]
\small
\caption{Drifter Training Procedure}
\label{alg:drifter_training}
\KwIn{Base distribution $\mathbb{P}_{\text{base}}$, drifted target distribution set $\{\mathbb{P}_{\text{target}}\}$, schedule variables $\{\beta_t\}_{t=1}^{T}$, loss balancer $\lambda_1$, $\lambda_2$}
\KwOut{Trained drifter \texttt{Drifter}$_\phi$}
\def\arraystretch{0.95}
\LinesNumberedHidden

\Numberline Initialize drifter parameters $\phi$\;

\tcc{Pre-compute $\alpha_t$ and $\bar{\alpha_t}$ for all timesteps}
\Numberline $\{\alpha_t\}_{t=1}^{T} \gets \{1 - \beta_t\}_{t=1}^{T}$, $\{\bar{\alpha_t}\}_{t=1}^{T} \gets \{\prod_{i=1}^{t} \alpha_i\}_{t=1}^{T}$\;

\Numberline \While{not converged}{
    \Numberline $\vb{x}_0 \sim \mathbb{P}_{\text{base}}$, $\mathbb{P} \sim \{\mathbb{P}_{\text{target}}\}$\;
    \Numberline $d^* \sim U(0,1)$, $t \sim \{1,\ldots,T\}$, $\boldsymbol{\epsilon}_t \sim \mathcal{N}(0, \vb{I})$\;

    \Numberline $\vb{x}_t \gets \sqrt{\bar{\alpha}_t} \vb{x}_0 + \sqrt{1 - \bar{\alpha}_t} \boldsymbol{\epsilon}_t$\;

    \Numberline $\vb{x}_{\text{drift}} \gets \texttt{Drifter}_\phi(\vb{x}_t, t, d^*)$\;

    \Numberline $\overline{\mathbb{P}} \gets p(\vb{x}_{\text{drift}})$\tcp*{get distribution of $\vb{x}_{\text{drift}}$}

    \Numberline $\overline{d} \gets \frac{1}{\log 2} \mathcal{D}_{JS}(\mathbb{P}_{\text{base}} \parallel \overline{\mathbb{P}})$\;
    
    \Numberline $\mathcal{L}_{{\text{Drift}}} = \left\| \overline{d} - d^* \right\|^2_2 + \lambda_1\left\| \mathcal{C}(\vb{x}_{\text{drift}}) - \mathcal{C}(\vb{x}_t) \right\|^2_2 $\\\quad\quad\quad\,\,\,$+ \lambda_2 \mathcal{D}_{JS}(\overline{\mathbb{P}} \parallel \mathbb{P})$\;
    
    \Numberline Update $\phi$ using gradient descent\;
}
\end{algorithm}
\vspace{-2mm}

\begin{figure}[t]
\centering
\includegraphics[width=0.98\linewidth]{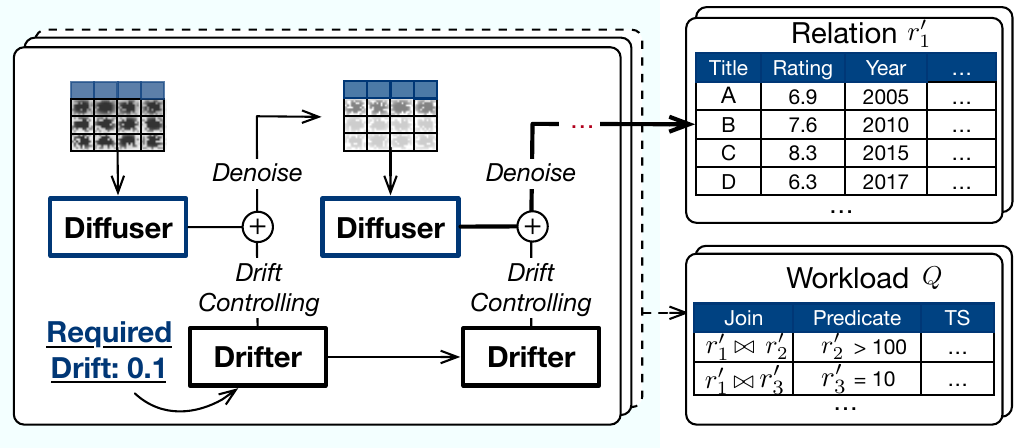}
\vspace{-2mm}
\caption{Data and Workload Generation Process}
\label{fig:generation}
\vspace{-4mm}
\end{figure}

\subsection{Drift-aware Data and Workload Generation}
\label{sec:ctrl-drift-gen}

We synthesize data and workload drift based on the diffuser and the drifter, as shown in Figure~\ref{fig:generation}.
After sampling an initial noise $\mathbf{x}^T \sim \mathcal{N}(0, \mathbf{I})$,
the diffuser transforms $\mathbf{x}^T$ into drifted data $\tilde{\mathbf{x}_0}$ over $T$ discrete timesteps.
At each step $t$, the drifter computes a gradient 
$g_t = \nabla_{\mathbf{x}_t} \mathrm{Pr}(\mathbf{x}_{\text{drift}} | \mathbf{x}_t) = \nabla_{\mathbf{x}_t} \texttt{Drifter}(\mathbf{x}_t, t, d)$,
which adjusts the predicted mean $\tilde{\boldsymbol{\mu}}_t = \boldsymbol{\mu} + \mathbf{\Sigma} \cdot g_t$, and use the adjusted mean $\tilde{\boldsymbol{\mu}}_t$ to sample the next state $\mathbf{x}_{t-1}$.
We repeat this process for timesteps $t = T, T-1, \dots, 1$, gradually refining $\mathbf{x}_t$ toward $\mathbf{x}_{\text{drift}}$.
After decoding $\mathbf{x}_{\text{drift}}$, we obtain the data or workload that reflects the desired drift while maintaining the inherent correlations.

\paragraphtitle{Correlation Preservation}
\dbname generates data as full tuples with multiple columns, thereby preserving inter-column correlations throughout the process.
First, the diffuser learns the joint distribution of columns and captures their correlations in the latent space.
Second, the drifter is trained using a composite loss (Algorithm~\ref{alg:drifter_training}, Line 10) that includes both a drift loss and a correlation preservation loss.
As a result, during generation, the drifter injects drift signals that drift data distributions, while the diffuser ensures that column correlations are retained.
\paragraphtitle{Snapshot-based Generation}
We would like to note that using a single generator trained across a broad temporal span (\emph{e.g.}, 2014–2017) may miss fine-grained drift patterns, especially for a small drift factor $d$.
For example, if a user sets 2013 as the baseline and specifies $d=0.05$ to simulate drift toward 2014, a generator trained on 2014–2017 data may capture overall trends but overlook subtle shifts between 2013 and 2014, due to smoothing effects from large-scale drift exposure.
To alleviate this, we introduce an \emph{agent-based} mechanism that selects the most appropriate generator for a target drift $d^*$.
Specifically, we train multiple snapshot-specific generators $\mathcal{G}_i\in\{\mathcal{G}_{14\text{-}15}, \mathcal{G}_{15\text{-}17}, \dots\}$, where each $\mathcal{G}_i$ consists of a diffuser trained on data up to time $i$, and a drifter guided by subsequent data. 
For example, $\mathcal{G}{14\text{-}15}$ uses data up to 2014 for diffusion and captures drift patterns up to 2015.
At runtime, the agent selects the best-matching generator by computing a distance metric $\mathrm{dist}(d, \alpha_i)$ between the requested drift $d$ and the representative drift $\alpha_i$ of each generator $S_i$.
The agent then selects the generator $\mathcal{G}_{i^*}$ with the minimal distance:
$\mathcal{G}_{i^*}=\argmin_{\mathcal{G}_{i}} \mathrm{dist}(d,\,\alpha_i)$.
\section{Implementation}
\label{sec:impl}
We have made the source code of \dbname available~\cite{neuralbenchimpl}, which includes a drift-aware generator and evaluation suites for learned database components.
We implement the generator based on DDPM~\cite{DBLP:journals/pacmmod/LiuFTLD24,DBLP:conf/icml/KotelnikovBRB23}.
In particular, the diffuser is based on a Multilayer Perceptron (MLP) with hidden layer dimensions of $(512, 1024, 1024, 512)$ and the drifter is a similar MLP with hidden layer dimensions of $(512, 512)$.
We use tables to represent both data and workloads, as specified in Section~\ref{subsec:drift_factor}, and encode the tables using the analog bit encoding~\cite{DBLP:journals/pacmmod/LiuFTLD24}.
For each table, we dynamically adjust the learning rate from $1e^{-4}$ to $2e^{-3}$ to ensure the convergence of the training process.

\dbname provides user commands such as \texttt{gd} for drift-aware data generation and \texttt{tqo} for learned query optimizer evaluations.
Users can easily reconfigure these settings to match their own environments and workloads for customized benchmarking.
Since no existing database integrates all tested learned database components, we adopt and extend existing evaluation suites~\cite{DBLP:journals/pvldb/WongkhamLLZLW22,DBLP:conf/osdi/WangDWCW0021,DBLP:journals/pvldb/LehmannSS24} to support benchmarking under data and workload drift.
\dbname is designed to be extensible to support additional evaluators.

\dbname is functionally adequate for our use case by enabling the generation of both numerical and categorical attributes.
However, as a diffusion-based method, \dbname introduces training and inference overhead.
Its generation efficiency can be constrained by available resources, and the long generation time required for TB-scale data is a potential limitation.

\section{Experimental Setup} 
\label{sec:preparation}
\label{sec:exp-setup}

In this section, we specify the experimental setup.



\subsection{Datasets and Workloads}
\label{subsec:benchmarks}
\label{subsec:workloaddataset}

We use two real-world datasets and generate synthetic datasets based on them.
1) \noindent\textbf{IMDB~\cite{imdb}} is a real-world dataset containing 21 tables with complex correlations.
By default, we use the IMDB version collected by Leis et al.~\cite{DBLP:journals/pvldb/LeisGMBK015}, containing real data up to 2013.
In addition, we collected snapshots from the IMDB archives~\cite{imdbarchive}, which contain IMDB data up to 2017.
\textbf{STACK~\cite{DBLP:conf/sigmod/MarcusNMTAK21}} consists of questions and answers from StackExchange websites~\cite{stackoverflow}.
We split the table by year to construct multiple yearly snapshots.

We execute JOB queries~\cite{DBLP:journals/pvldb/LeisGMBK015} on IMDB, which consists of 113 queries derived from 33 real-world query templates, and execute another 113 queries sampled by~\cite{DBLP:journals/pvldb/LehmannSS24} on STACK.



\begin{vldbrevisionenv}
\end{vldbrevisionenv}

\subsection{Evaluated Data Generators}
\label{subsec:eval_data_genertor}

We compare the drift-aware data and workload generation framework of \dbname with two data generation techniques that can be extended to support drift data generation:
1) \textbf{RelDDPM~\cite{DBLP:journals/pacmmod/LiuFTLD24}} is a state-of-the-art data synthesis framework that generates correlation-preserving tabular data using DDPM.
RelDDPM originally supported only non-drift data generation.
We adjust the strength of the guidance in the denoising process of RelDDPM to generate drifted data, and use it as our baseline.
2) \textbf{CTGAN~\cite{DBLP:conf/nips/XuSCV19}} is a widely used framework for modeling tabular data distributions via conditional GANs. As CTGAN lacks explicit drift control, we vary its hyperparameters to generate drifted datasets.





\subsection{Evaluated Learned Components}
\label{sec:revisit}

We select representative learned database components, as listed in Table~\ref{tab:baselines}, to conduct our evaluations.
We ensure the evaluated components cover as many key design choices as possible. 
Due to space constraints, we opt not to conduct exhaustive evaluations of all existing learned components.

\extended{
\begin{table}
\renewcommand\arraystretch{1.1}
\caption{Training Time (hours) of End-to-end Learned Query Optimizers}
\label{tab:train}
\vspace{-2mm}
\resizebox{0.95\linewidth}{!}{%
\begin{tabular}{llllll}
\toprule
 &
\textbf{Bao~\cite{DBLP:conf/sigmod/MarcusNMTAK21}} &
  \textbf{HybridQO~\cite{DBLP:journals/pvldb/YuC0L22}} &
  \textbf{Lero~\cite{DBLP:journals/pvldb/MarcusNMZAKPT19}} &
  \textbf{Balsa~\cite{DBLP:conf/sigmod/YangC0MLS22}}
  \\ \midrule
Train on JOB & 2.5 & 18.3 & 10.1 & 10.5 \\ 
Train on JOB-light & 5.3 & 40.2 & 15.2 & 19.3 \\ 
\bottomrule
\end{tabular}%
}
\vspace{-3mm}
\end{table}

\begin{figure}
    \centering
    \includegraphics[width=0.44\linewidth]{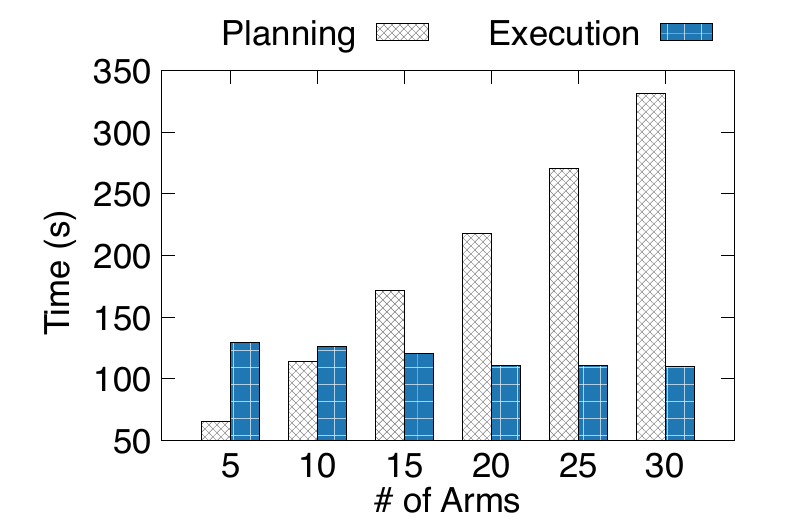}
    \vspace{-4mm}
    \caption{Performance of Bao with Varying Number of Arms}
    \label{fig:revision.baoarm}
    \vspace{-4mm}
\end{figure}

\begin{figure}
    \centering
    \subfigure[Bao]{
        \includegraphics[width=0.44\linewidth]{figures/learning curve/loss_curve_bao.pdf}
    }
    \subfigure[Balsa]{
        \includegraphics[width=0.44\linewidth]{figures/learning curve/loss_curve_balsa.pdf}
    }
    \vskip -4mm
    \subfigure[HybridQO]{
        \includegraphics[width=0.44\linewidth]{figures/learning curve/loss_curve_hybridqo.pdf}
    }
    \subfigure[Lero]{
        \includegraphics[width=0.44\linewidth]{figures/learning curve/loss_curve_lero.pdf}
    }
    \vspace{-4mm}
    \caption{Learning Curves of Evaluated End-to-end Learned Query Optimizers}
    \label{fig:revision.learningcurve}
    \vspace{-4mm}
\end{figure}
}

\subsubsection{End-to-end Learned Query Optimizer} \label{sec:overview:optimzer}




Given a query $\vb{q}$, the optimizer first performs \textit{plan search} to obtain a set of candidate query plans $\vb{P}(\vb{q})=\{\vb{p}_0,\vb{p}_1,...,\vb{p}_N\}$, and then select the best one from $\vb{P}(\vb{q})$ for real execution.
During the plan search and selection, it interacts with an \textit{ML model} to assess the plan quality.
Existing approaches~\cite{DBLP:journals/pvldb/YuC0L22,DBLP:journals/pvldb/MarcusNMZAKPT19,DBLP:conf/sigmod/YangC0MLS22,DBLP:journals/pvldb/ZhuCDCPWZ23,DBLP:conf/sigmod/MarcusNMTAK21} typically repeat the above steps to construct an experience set, which is then used to guide learning.

\paragraphtitle{Model}
There are two main types of models for assessing plan quality:
1) \textit{\underline{Regression model}}
predicts the execution cost or latency $\hat{L}(\vb{p})$ of a given plan $\vb{p}$, optionally providing a confidence score $C(\vb{p})$ that reflects the model's certainty, \emph{i.e.},
$f_{qo}(\vb{p}) = (\hat{L}(\vb{p}), C(\vb{p}))$~\cite{DBLP:conf/sigmod/MarcusNMTAK21,DBLP:journals/pvldb/MarcusNMZAKPT19,DBLP:journals/pvldb/YuC0L22}.
2) \textit{\underline{Ranking model}}
predicts a binary ranking value $r$ for two given plans $\vb{p}_1$ and $\vb{p}_2$, where $r=1$ if $\vb{p}_1$ is expected to outperform $\vb{p}_2$, and $r=0$ otherwise, \emph{i.e.},
$f_{qo}(\vb{p}_1,\vb{p}_2) = r$,
$r \in \{0,1\}$~\cite{DBLP:journals/pvldb/ZhuCDCPWZ23}.
\extended{We note that since a query plan can be represented as a tree, these models are typically implemented with one of two structures: Tree-CNN~\cite{DBLP:conf/aaai/MouLZWJ16} and Tree-LSTM~\cite{DBLP:conf/acl/TaiSM15}.
Both structures are deep neural networks designed for tree structures to capture bottom-up sequential information.
}

\paragraphtitle{Plan Search}
There are two paradigms in plan search, which differ in whether they rely on traditional query optimizers:
1) \textit{\underline{White-box search}} provides auxiliary information to traditional query optimizers for generating candidate plans.
Based on the type of information provided, it can be further categorized into:
(a) \textit{\underline{Hint-based.}} It supplies hints (parameters that control physical operators, such as join or scan types) to the traditional optimizer,
guiding it to generate corresponding query plans~\cite{DBLP:conf/sigmod/MarcusNMTAK21,DBLP:journals/pvldb/YuC0L22}.
\extended{Selecting proper hints is therefore important for the quality of the candidate plan set.
Early approaches~\cite{DBLP:conf/sigmod/MarcusNMTAK21} determine static hint sets in advance by database experts with in-depth knowledge of query optimization.
Some recent methods are proposed to dynamically select hints based on search algorithms~\cite{DBLP:journals/pvldb/YuC0L22,DBLP:journals/pvldb/AnneserTCXPLM23}.
For example, 
HybridQO~\cite{DBLP:journals/pvldb/YuC0L22} uses a Monte Carlo tree search combined with a join order prediction model to generate a set of potential join orders as hints.
AutoSteer~\cite{DBLP:journals/pvldb/AnneserTCXPLM23} automatically discovers hints that can impact a given query through a greedy search of the hints supported by the database.}
(b) \textit{\underline{Cardinality-based.}}
It adjusts (by magnifying or reducing) estimated cardinalities multiple times to generate a variety of candidate plans~\cite{DBLP:journals/pvldb/ZhuCDCPWZ23}.
2) \textit{\underline{Black-box search}} embeds the model directly into the plan search process without relying on the traditional optimizers~\cite{DBLP:journals/pvldb/MarcusNMZAKPT19}.
It starts with an empty or partial plan and iteratively constructs a complete query plan until a predefined number of candidate plans is reached.
At each step, multiple subplans are generated.
After using an ML model to evaluate their quality, ineffective subplans are pruned while the remaining ones are expanded in subsequent iterations.

\begin{table}[t]
\renewcommand\arraystretch{1.1}
\caption{A Taxonomy of Evaluated Learned Components}
\label{tab:baselines}
\vspace{-2mm}
\resizebox{0.85\linewidth}{!}{%
\begin{tabular}{llll}
\toprule
\multirow{4}{*}{\begin{tabular}[c]{@{}l@{}}End-to-end \\ Learned Query \\ Optimizer\end{tabular}}     & \textbf{Method}                                                      & \textbf{Model} & \textbf{Plan Search}\\ \cline{2-4}
 &
  Bao~\cite{DBLP:conf/sigmod/MarcusNMTAK21} &
  Regression &
  Hint-based \\
& Balsa~\cite{DBLP:conf/sigmod/YangC0MLS22} & Regression     & Black-box         \\
& Lero~\cite{DBLP:journals/pvldb/ZhuCDCPWZ23}    & Ranking        & Cardinality-based      \\
\toprule
\multirow{4}{*}{\begin{tabular}[c]{@{}l@{}}Updatable\\ Learned Index\end{tabular}}     & \textbf{Method}                                                      & \textbf{Structure} & \textbf{Index Insertion}  \\ \cline{2-4}
& ALEX~\cite{DBLP:conf/sigmod/DingMYWDLZCGKLK20} &
  Heterogeneous &
  In-place \\
& LIPP~\cite{DBLP:journals/pvldb/WuZCCWX21} &
  Homogeneous &
  In-place \\
& XIndex~\cite{DBLP:conf/ppopp/TangWDHWWC20} & 
  Heterogeneous & 
  Delta-buffer \\
\hline
Learned CC    & PolyJuice~\cite{DBLP:conf/osdi/WangDWCW0021}                                                            &                &                                       \\
\bottomrule
\end{tabular}
}
\vspace{-4mm}
\end{table}

\subsubsection{Updatable Learned Index} \label{sec:overview:index}
Given a search key $\vb{k}$, the learned index $f_{idx}$ predicts its position $\hat{PS}(\vb{k})$ in the data table, \emph{i.e.}, $f_{idx}(\vb{k}) = \hat{PS}(\vb{k})$.
Learned indexes organize a set of ML models in a layered structure, where each node contains a model for position prediction.

\paragraphtitle{Structure}
Existing learned indexes can be classified into two representative structures, based on whether they have different node types:
1) \textit{\underline{Heterogeneous structure.}}
It consists of two node types: inner nodes and leaf nodes.
Inner nodes do not store data but predict the next step in the index.
Leaf nodes store the actual data, and a local search is performed to locate the exact key position.
2) \textit{\underline{Homogeneous structure.}}
Each node is responsible for storing part of the data and also predicting whether the key falls within its range or should be passed to a corresponding child node.

\paragraphtitle{Index Insertion}
A learned index is regarded as updatable if it supports key insertions.
There are two main key insertion mechanisms:
1) \textit{\underline{In-place.}}
Index nodes reserve gaps (empty slots), and newly inserted keys are directly inserted into an appropriate node.
\extended{
In such cases, conflict resolution is required, typically following one of these two methods:
(a) \textit{\underline{Shift method}}, sequentially moving the existing key to the nearest available gap to create space at the target location for successful insertion~\cite{DBLP:conf/sigmod/DingMYWDLZCGKLK20};
(b) \textit{\underline{Chain method}}, creating a new node for the conflicting target position and inserting the key into this new node~\cite{DBLP:journals/pvldb/WuZCCWX21}.
}
2) \textit{\underline{Delta-buffer.}}
Dedicated buffers are maintained, and new keys are first inserted into the buffer and later merged into the main structure periodically.
\extended{The buffer can be organized at two different levels of granularity:
(a) \textit{\underline{Tree-level buffer.}} A single buffer is maintained for the entire index~\cite{DBLP:journals/pvldb/FerraginaV20};
(b) \textit{\underline{Node-level buffer.}} Each node has its own dedicated buffer~\cite{DBLP:conf/ppopp/TangWDHWWC20}.}


\subsubsection{Learned Concurrency Control (CC)} \label{sec:overview:cc}
Learned CC algorithm dynamically selects the most suitable CC actions, such as locking, optimistic validation, \emph{etc.}, to gain higher transaction throughput.
It takes a transaction operation $op$ (\emph{e.g.}, read/write) and the current system condition $C(op)$ (\emph{e.g.}, contention levels) as inputs, and leverages an ML model $f_{cc}$ to predict the most appropriate CC action for the operation $op$.
Formally, let $\vb{A}$ be the set of available CC actions, $f_{cc}$ can be defined as
$f_{cc}(op,C(op)) = a$,
where $a \in \vb{A}$.






\subsection{Evaluation Metrics} 
\label{subsec:metrics}


We assess the quality of synthesized drift by comparing its \textbf{drift factor} and \textbf{correlation difference (corr. diff.)} to those observed in real data. 
A drift factor closer to the real drift indicates more realistic distributional changes. 
Corr. diff., measured via Pearson correlation coefficient~\cite{pearson}, quantifies how much attribute-level correlation deviates from the original data. 
Lower correlation difference, i.e., similar to that under real drift, suggests better correlation preservation.
We define a good drift as one that achieves the target drift factor while preserving correlations as much as possible.


We use \textbf{execution time} to assess learned query optimizers, which is the actual time to execute a query.
We do not include query planning time in the execution time.
We use \textbf{throughput}, the number of operations/transactions processed per second, to evaluate learned indexes and learned CC.

\begin{figure*}[t]
\centering
\subfigure[Drift Factor (IMDB)]{
\label{subfig.drift_imdb_17}
\includegraphics[height=0.12\textwidth]{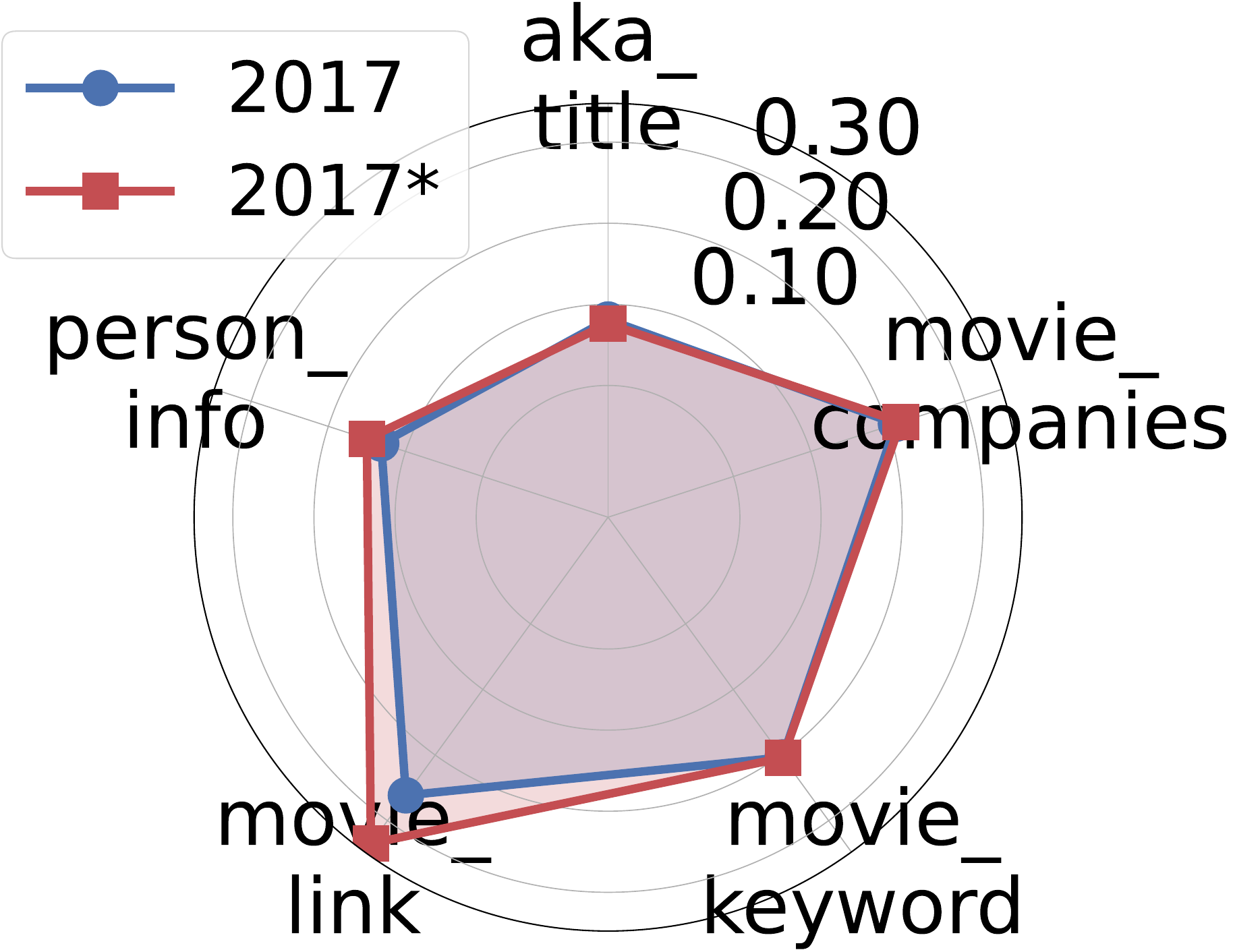}
}
\subfigure[Corr. Diff. (IMDB)]{
\label{subfig.corr_imdb_17}
\includegraphics[height=0.12\textwidth]{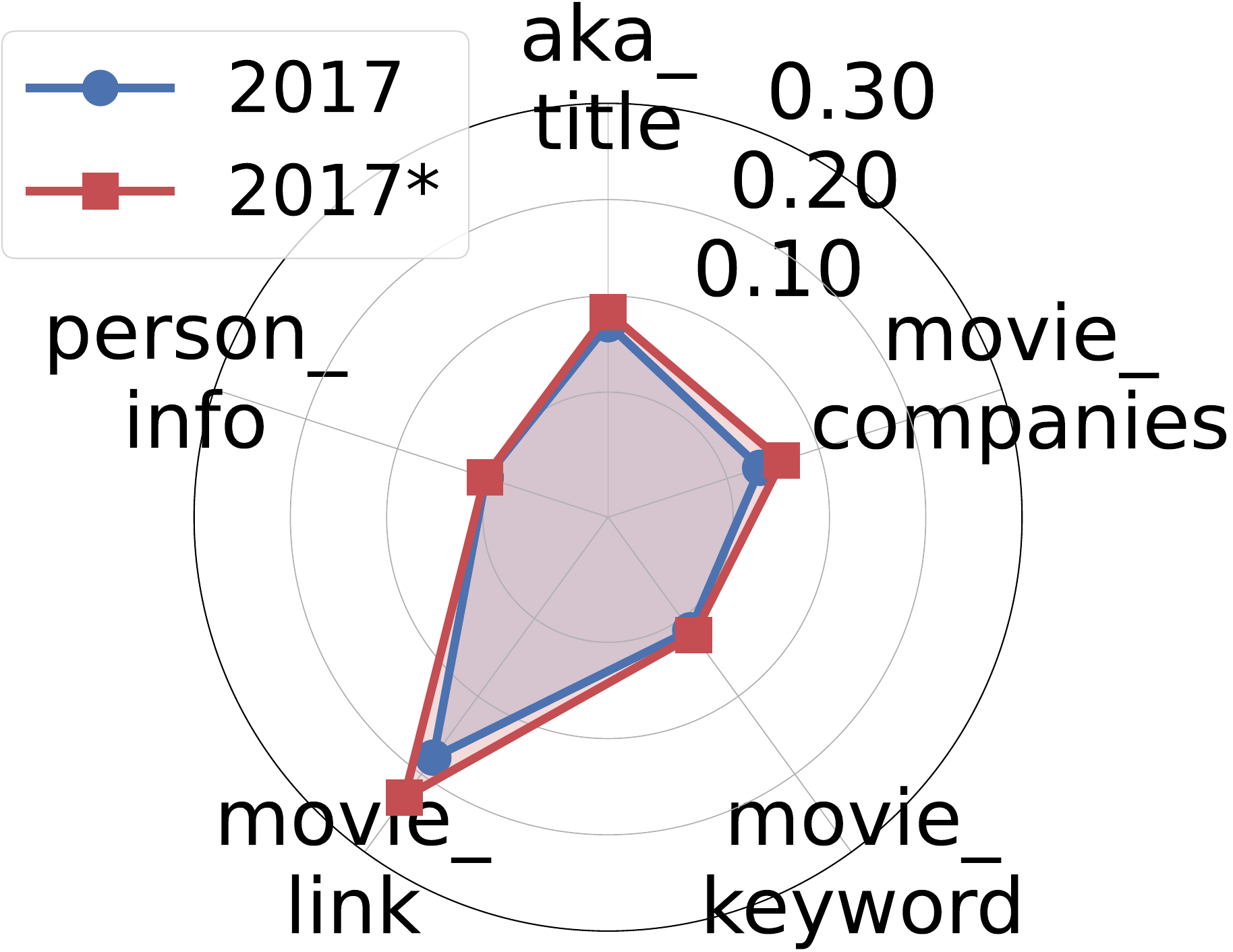}
}
\subfigure[Exec. Time (IMDB)]{
\label{subfig.exec_imdb_17}
\includegraphics[height=0.12\textwidth]{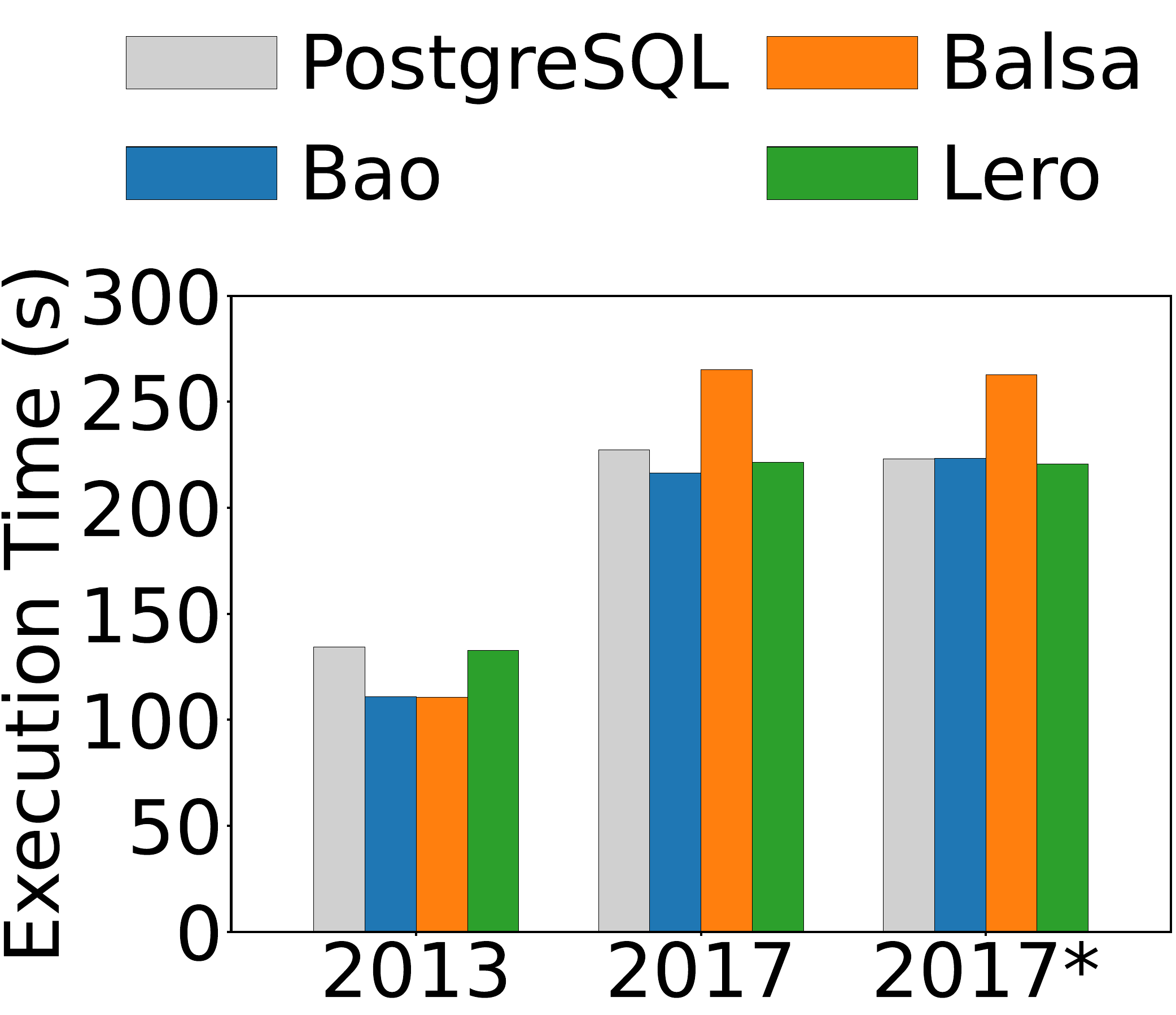}
}
\subfigure[Drift Factor (STACK)]{
\label{subfig.drift_stack_10}
\includegraphics[height=0.12\textwidth]{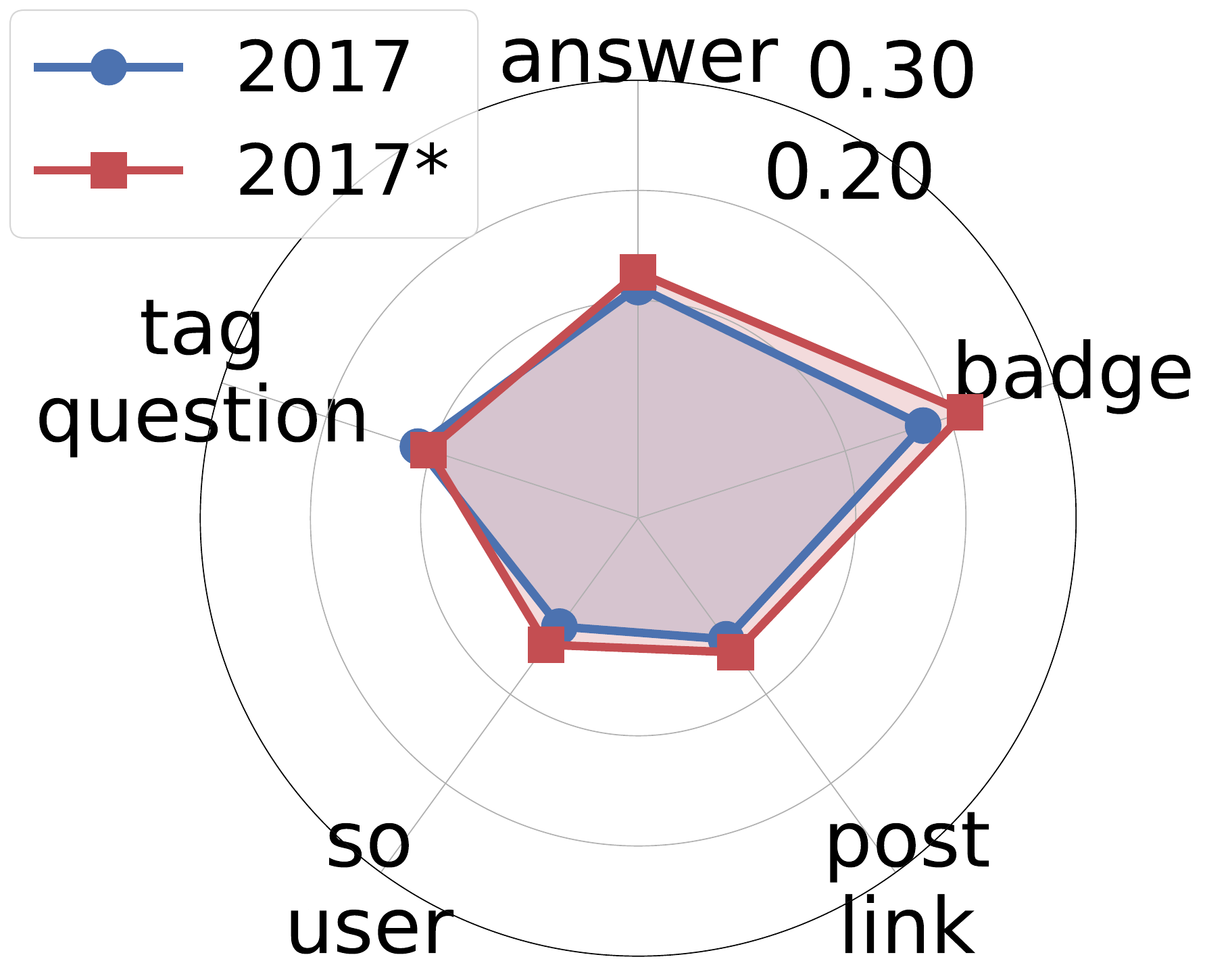}
}
\subfigure[Corr. Diff. (STACK)]{
\label{subfig.corr_stack_10}
\includegraphics[height=0.12\textwidth]{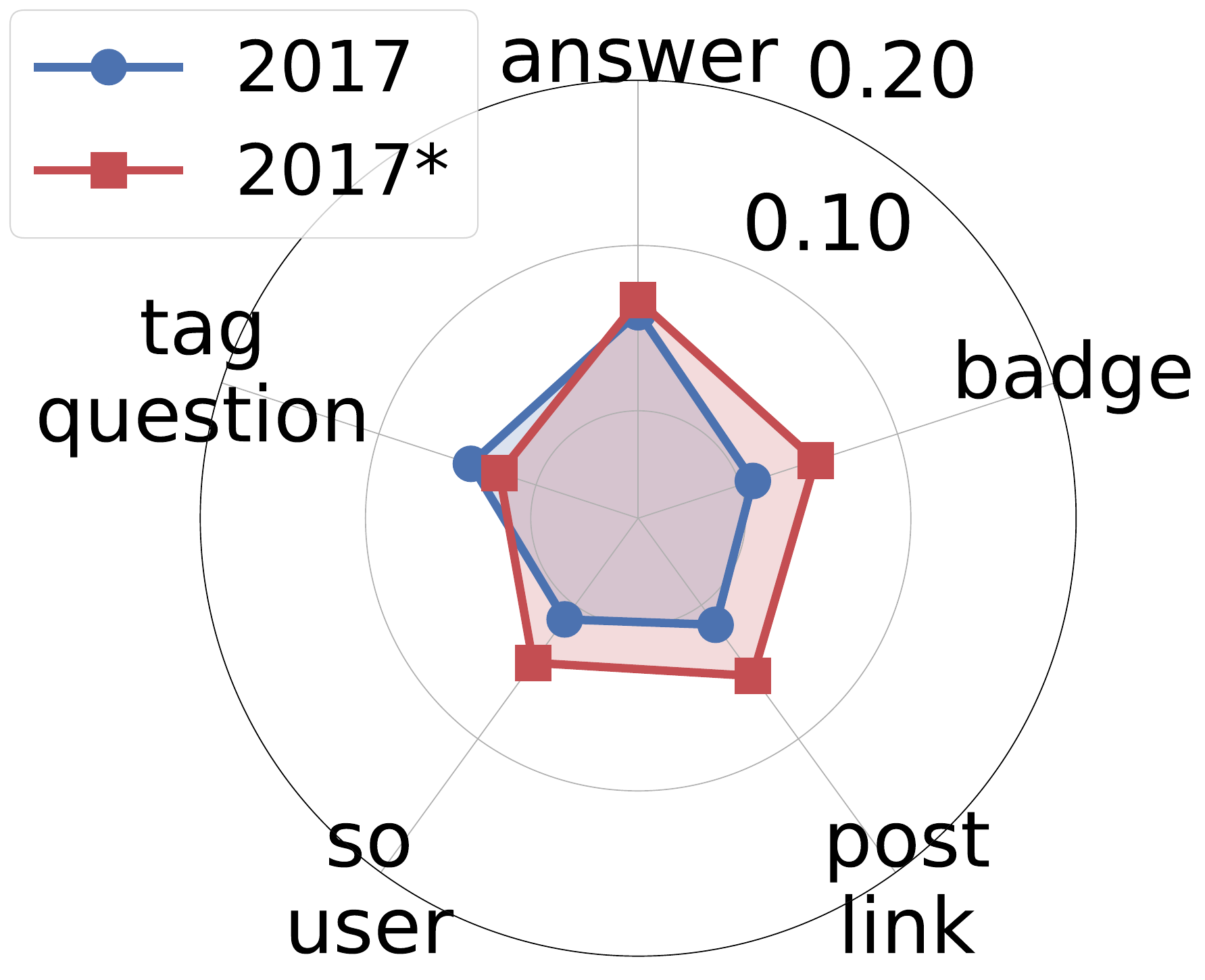}
}
\subfigure[Exec. Time (STACK)]{
\label{subfig.exec_stack_10}
\includegraphics[height=0.12\textwidth]{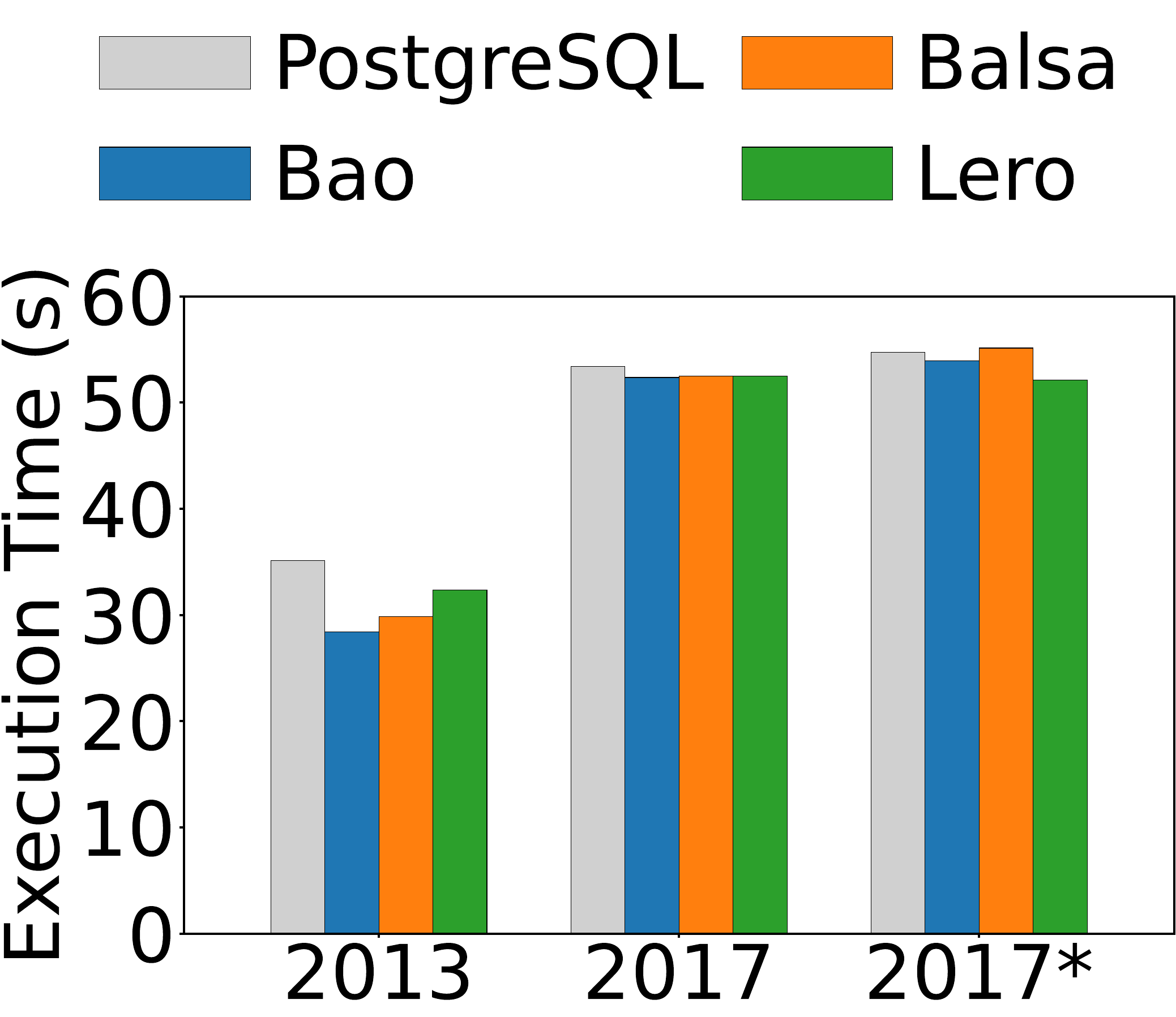}
}
\vspace{-4mm}
\captionsetup{justification=raggedright}
\caption{Extrapolation Performance of Drift-aware Generator}
\label{fig.generation_extrapolating}
\vspace{-4mm}
\end{figure*}

\begin{figure*}[t]
\centering
\subfigure[Drift Factor (IMDB)]{
\label{subfig.drift_imdb_15}
\includegraphics[height=0.12\textwidth]{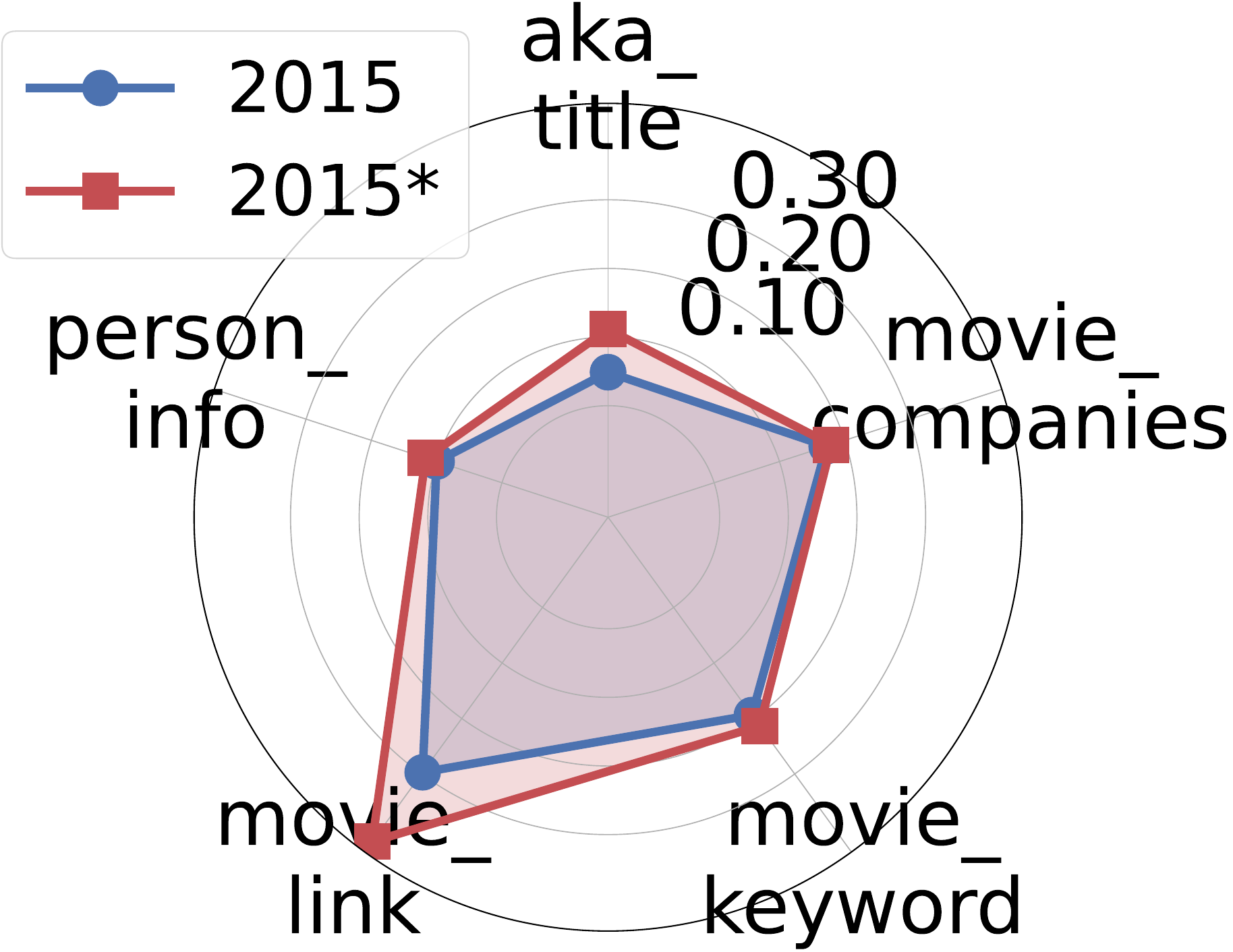}
}
\subfigure[Corr. Diff. (IMDB)]{
\label{subfig.corr_imdb_15}
\includegraphics[height=0.12\textwidth]{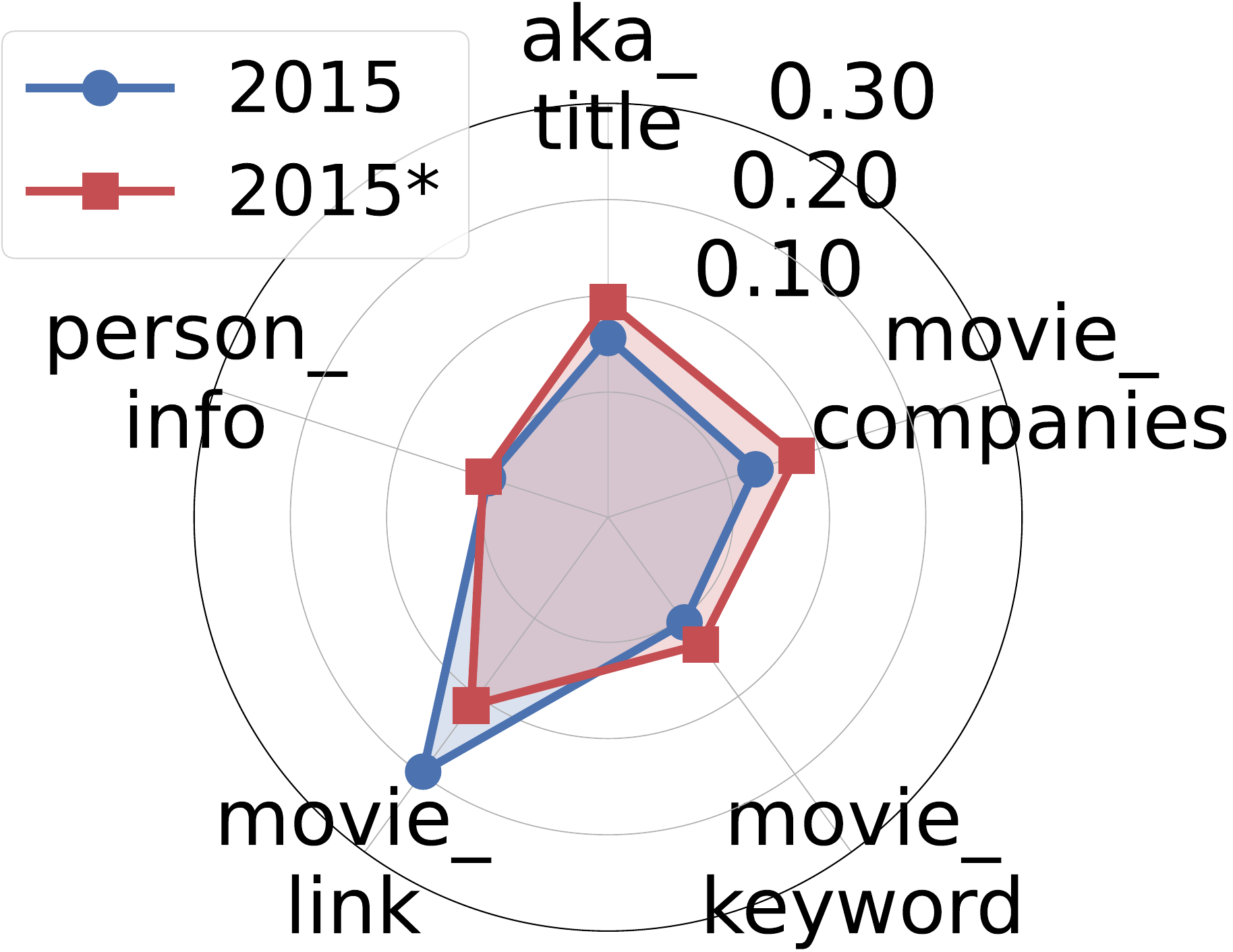}
}
\subfigure[Exec. Time (IMDB)]{
\label{subfig.exec_imdb_15}
\includegraphics[height=0.12\textwidth]{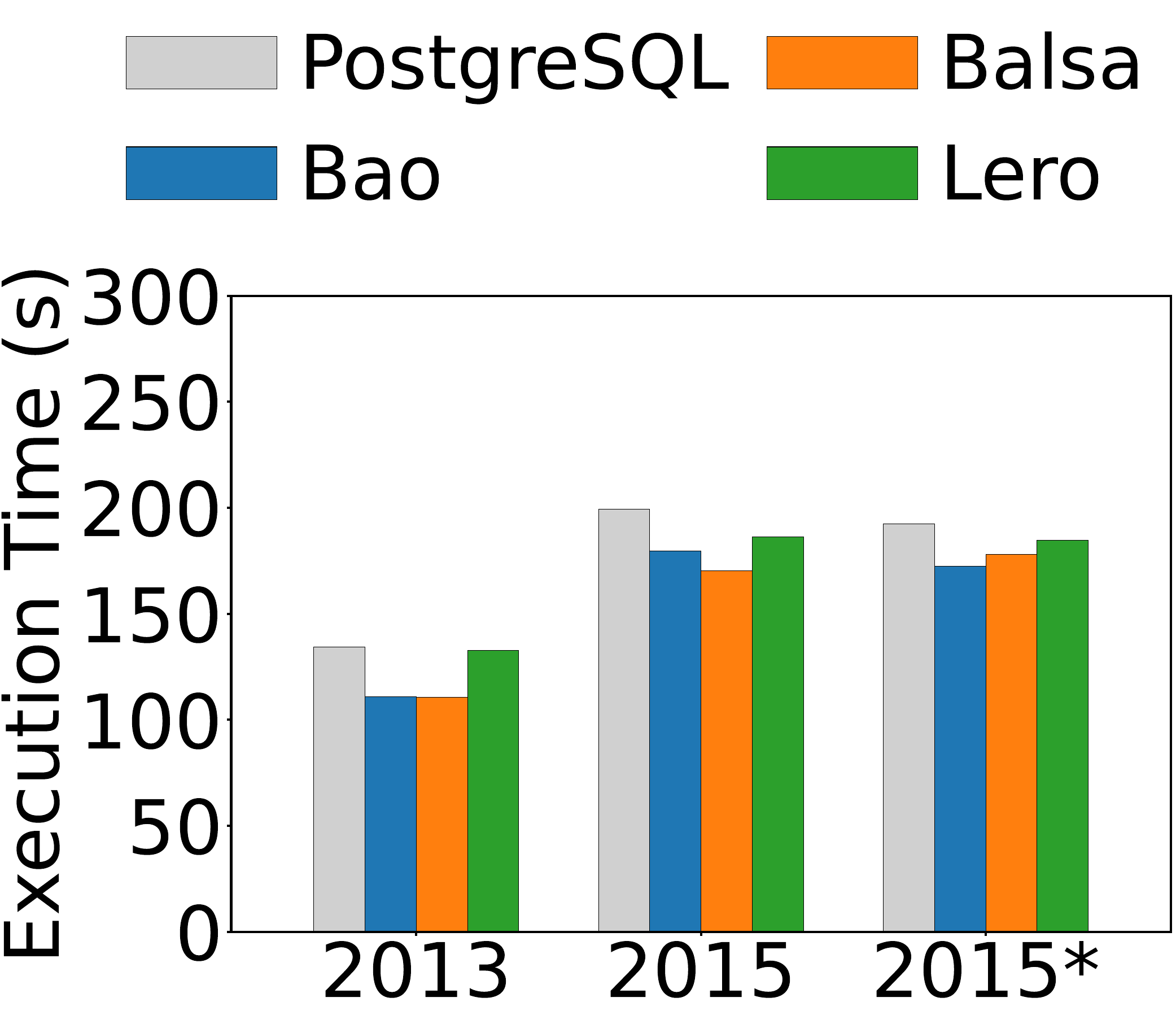}
}
\subfigure[Drift Factor (STACK)]{
\label{subfig.drift_stack_15}
\includegraphics[height=0.12\textwidth]{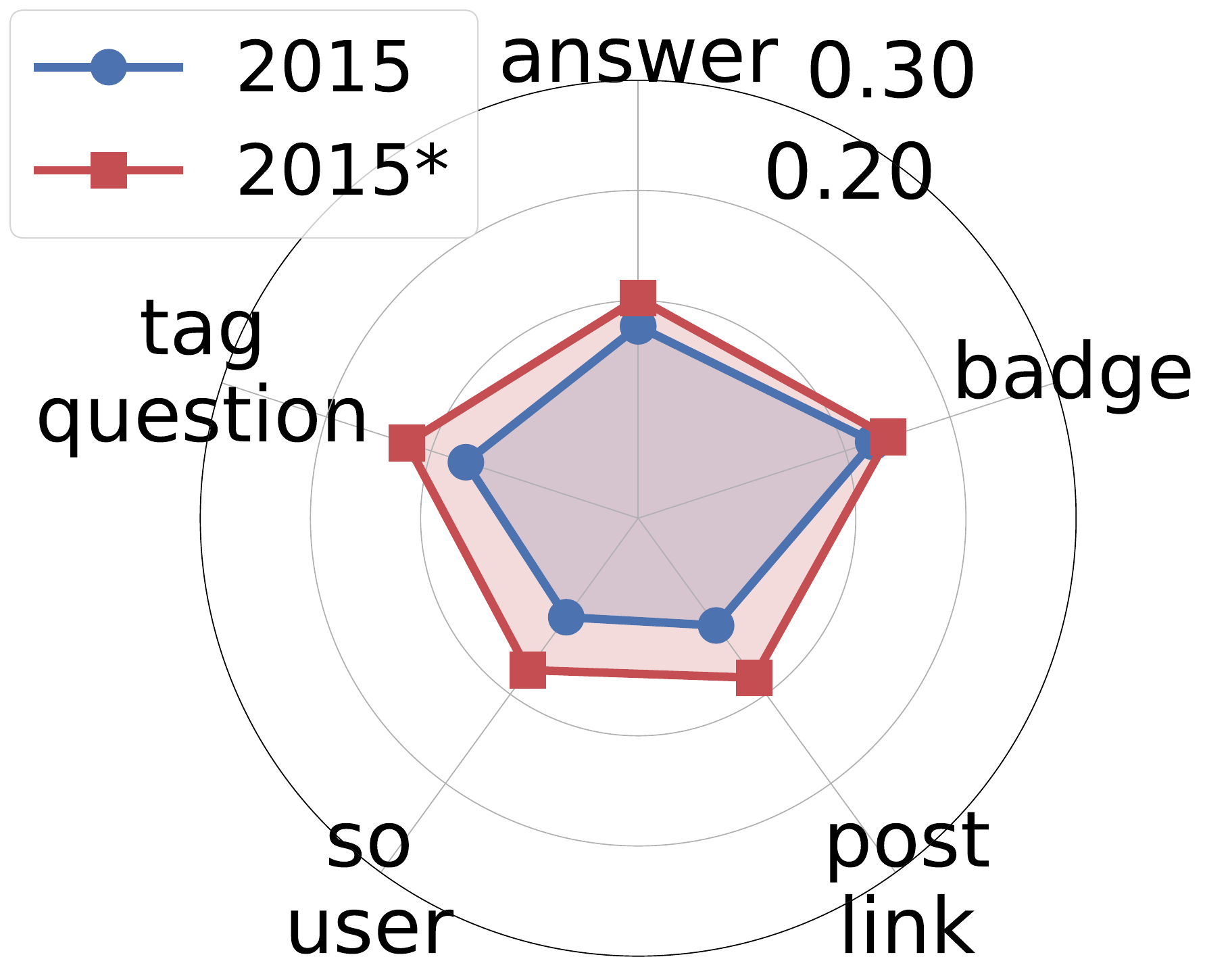}
}
\subfigure[Corr. Diff. (STACK)]{
\label{subfig.corr_stack_15}
\includegraphics[height=0.12\textwidth]{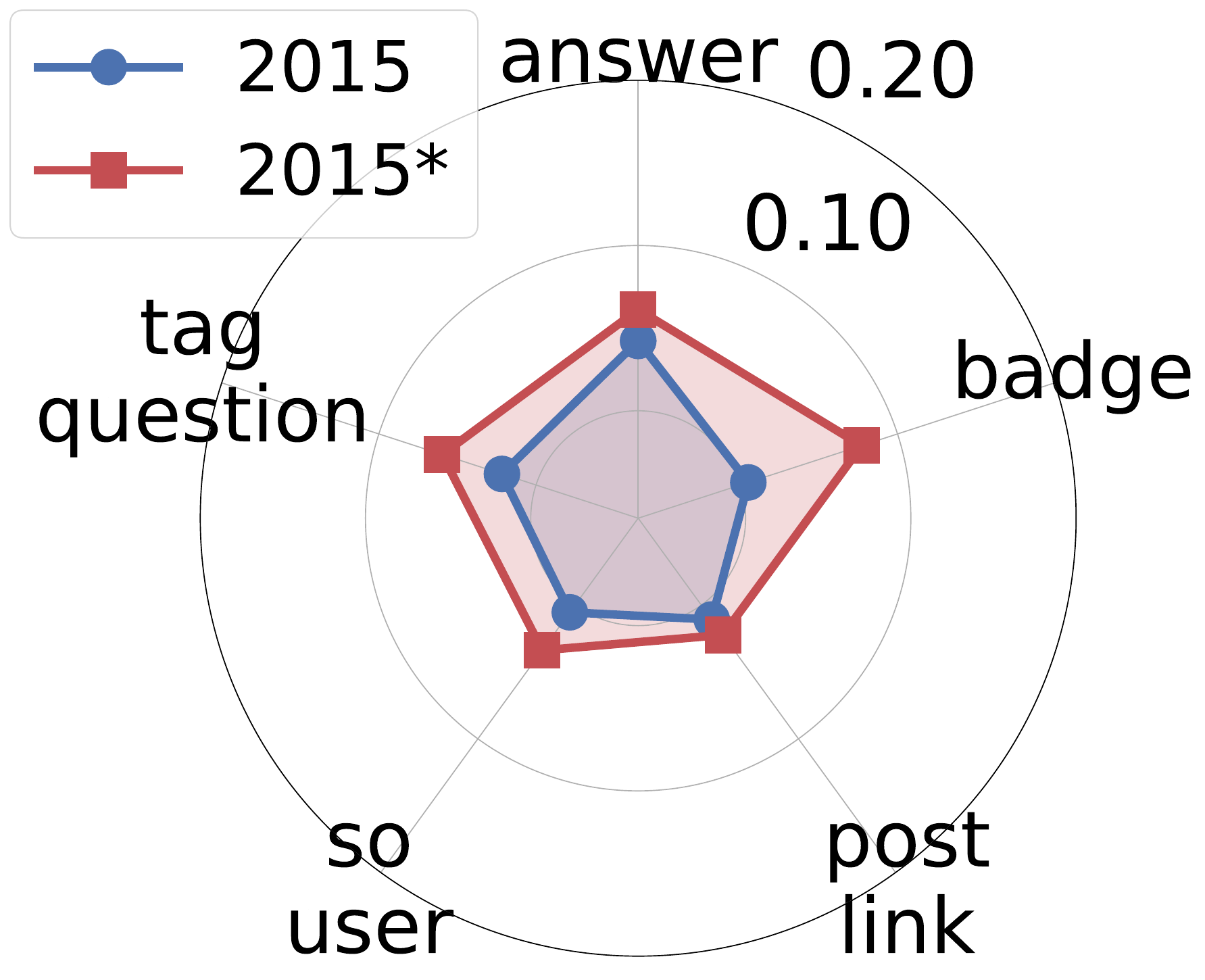}
}
\subfigure[Exec. Time (STACK)]{
\label{subfig.exec_stack_15}
\includegraphics[height=0.12\textwidth]{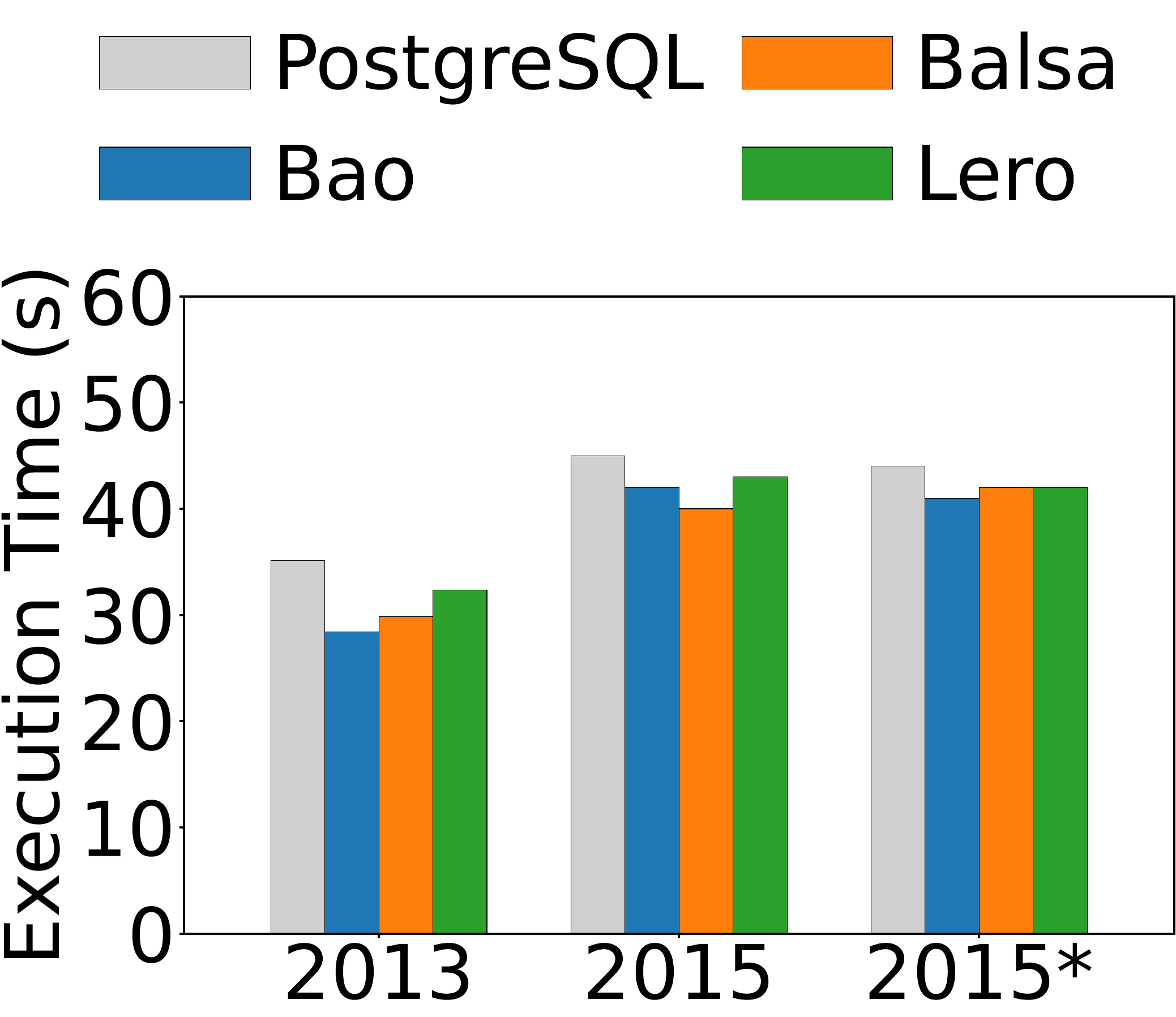}
}
\vspace{-4mm}
\captionsetup{justification=raggedright}
\caption{Interpolation Performance of Drift-aware Generator}
\label{fig.generation_interpolating}
\vspace{-4mm}
\end{figure*}

\begin{figure*}[t]
\centering
\subfigure[aka\_title]{
\label{subfig.imdb_kind_id}
\includegraphics[width=0.18\linewidth]{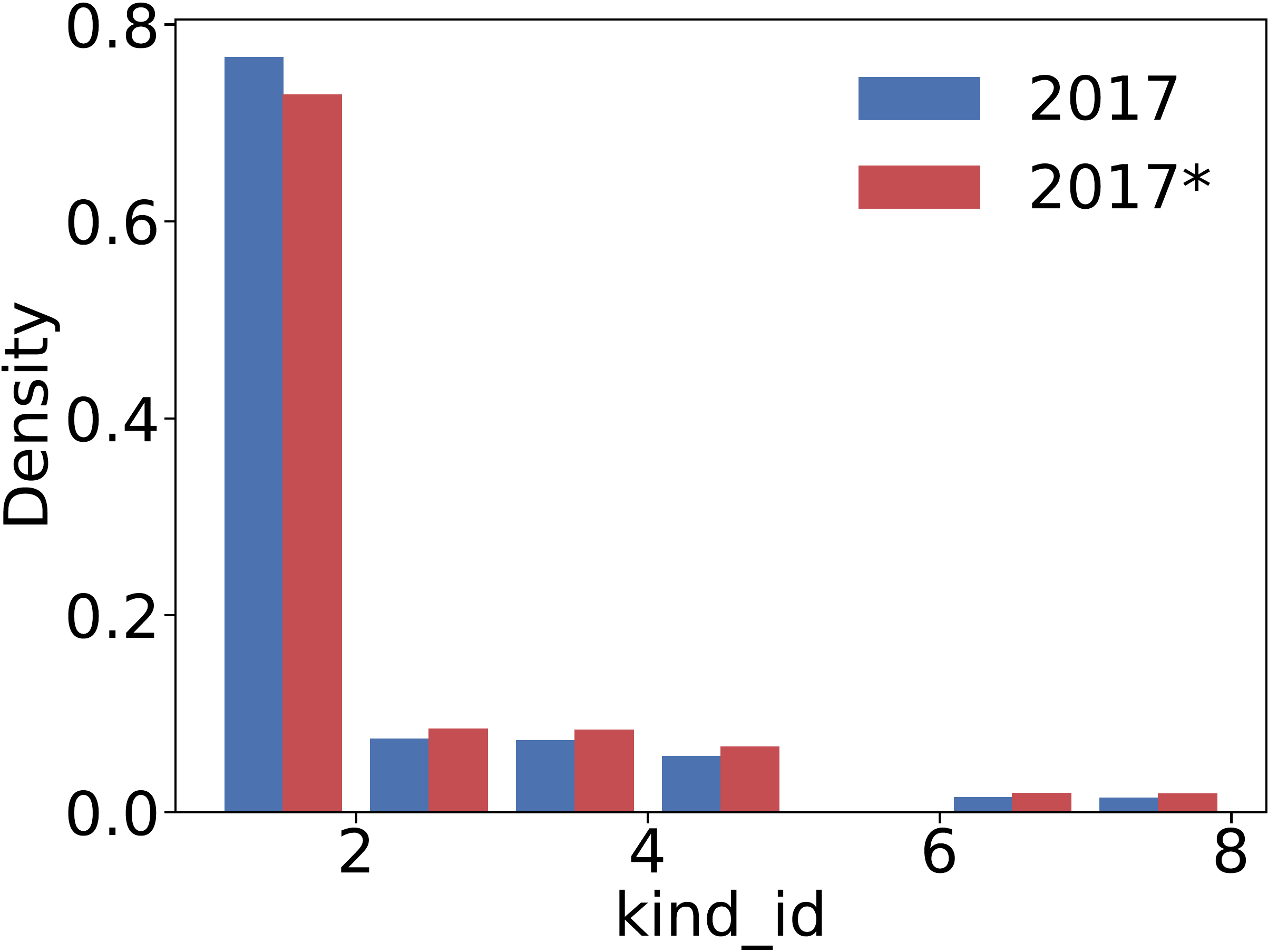}
}
\subfigure[movie\_companies]{
\label{subfig.imdb_movie_companies}
\includegraphics[width=0.18\linewidth]{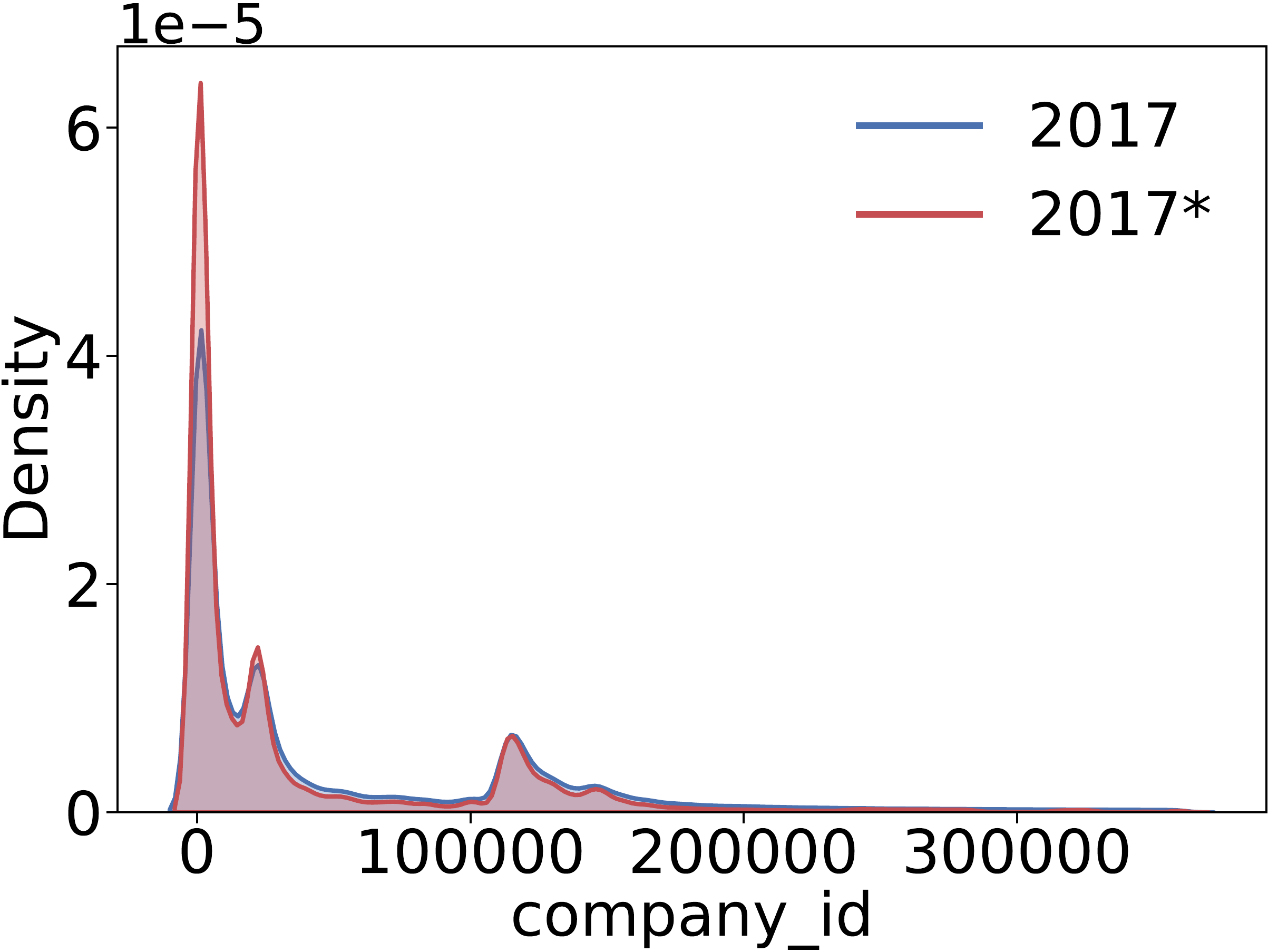}
}
\subfigure[movie\_keyword]{
\label{subfig.imdb_movie_keyword}
\includegraphics[width=0.18\linewidth]{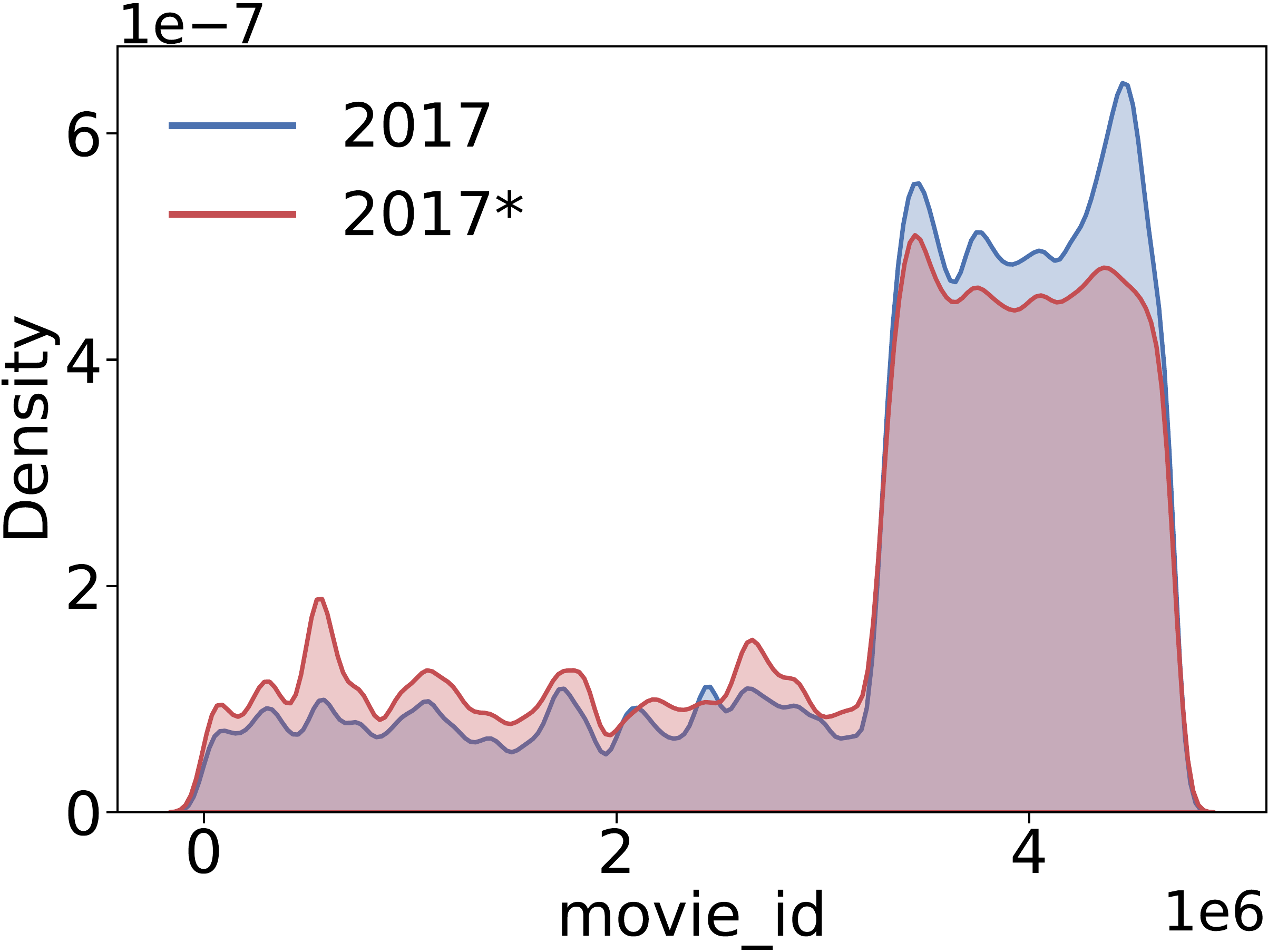}
}
\subfigure[movie\_link]{
\label{subfig.imdb_movie_link}
\includegraphics[width=0.18\linewidth]{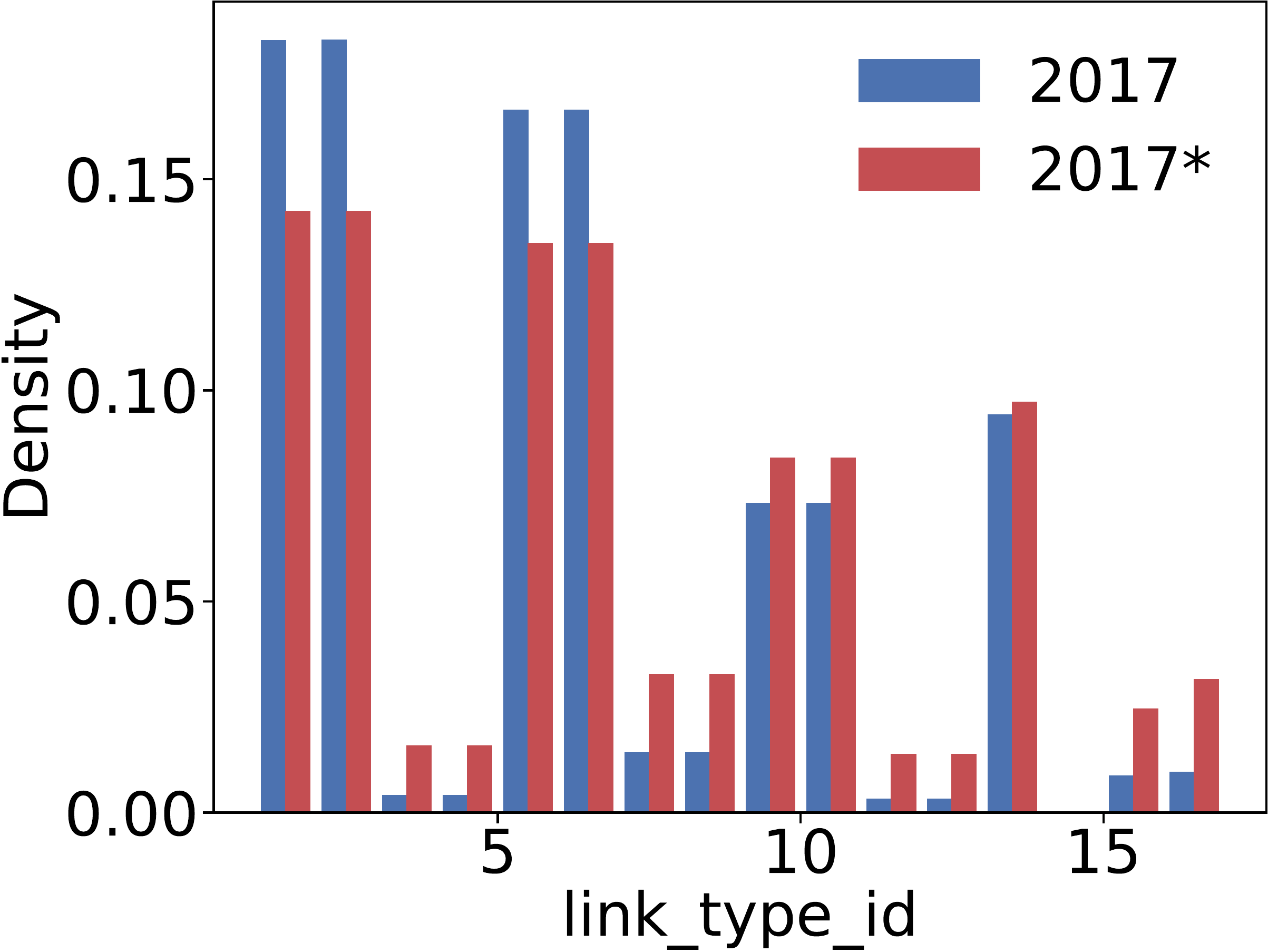}
}
\subfigure[person\_info]{
\label{subfig.imdb_person_info}
\includegraphics[width=0.18\linewidth]{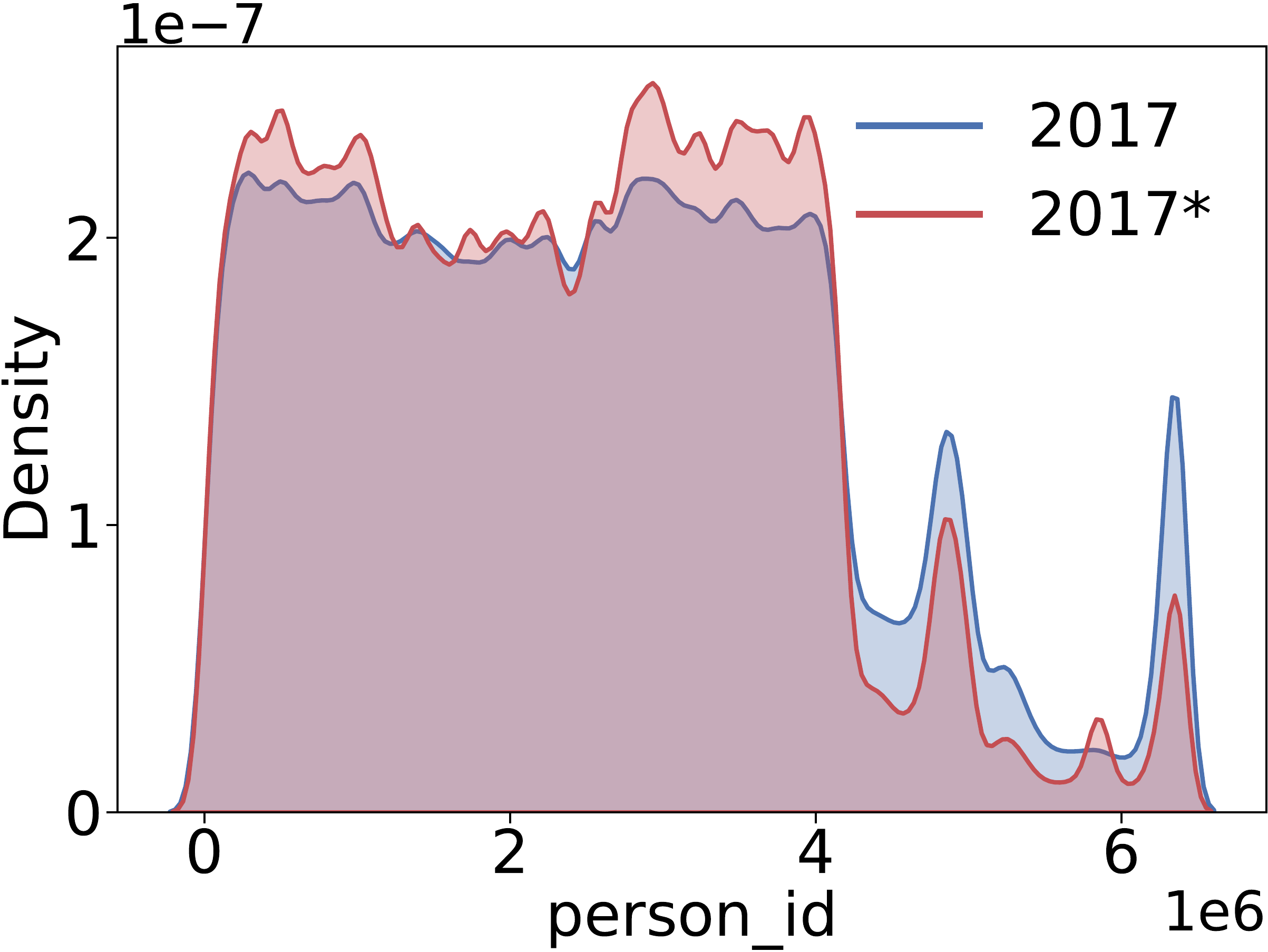}
}
\vspace{-4mm}
\captionsetup{justification=raggedright}
\caption{Column-level Distribution of Generated Data (IMDB)}
\label{fig.generation_distribution}
\vspace{-4mm}
\end{figure*}

\subsection{Default Configuration}
\label{subsec:default_setting}
We conduct all experiments on a server with 2 Intel Xeon Silver 4214R @ 2.40GHz CPUs, a total of 24 cores and 48 threads, and 128 GB of memory. 
The server is equipped with 8 NVIDIA RTX 3090 GPUs.
All experiments are run within Docker containers (Ubuntu 22.04 with CUDA 11.8.0).



Unless otherwise specified, we train the evaluated components on one dataset and workload, and test on other datasets and workloads to rigorously evaluate drift~\cite{DBLP:journals/pvldb/NegiWKTMMKA23}.
We randomly sample and repeatedly execute the queries provided in Section~\ref{subsec:workloaddataset} to train the evaluated learned query optimizers, following current practices~\cite{DBLP:journals/pvldb/LehmannSS24,DBLP:conf/sigmod/YangC0MLS22}.
By default, we execute 1500/33900/1500 queries to train Bao/Balsa/Lero, respectively.
Note that Lero may execute additional queries during its exploratory process.
We only report execution time in our experiments, while excluding planning time.
We do not implement parallel planning for Bao, and use the planning code provided in the original Bao repository~\cite{DBLP:conf/sigmod/MarcusNMTAK21}.
For data drift experiments, we compare the delta between the drifted data and the original data, and apply updates or inserts accordingly.
The statistics are updated in PostgreSQL.
We standardize the PostgreSQL version to v12.20 and apply consistent configurations to ensure fair comparisons.
Following~\cite{DBLP:journals/pvldb/LehmannSS24}, we set the shard buffer size of PostgreSQL to 32GB.
Consequently, the absolute performance levels of the evaluated learned query optimizers may differ from those reported in their original papers.

\section{Experimental Results} \label{sec:evaluation}
In this section, we first evaluate the effectiveness of \dbname in generating realistic drift, and then conduct benchmarks on selected learned database components using \dbname.



\extended{
\begin{figure}
    \centering
    \includegraphics[width=0.42\linewidth]{figures/learning curve/loss_curve_drifter.pdf}
    \vspace{-4mm}
    \caption{Learning Curve of the Drifter}
    \label{fig:revision.drifter_learning}
    \vspace{-4mm}
\end{figure}
}

\subsection{Experiments on Drift Generator}
\label{subsec:drift_generator_eval}

\subsubsection{Experiments on the quality of synthesized data drift}
\label{subsec:eval:abl:data_quality}


We evaluate the effectiveness of \dbname in drift data generation by comparing synthesized drift with real drift across two settings: extrapolation and interpolation. 
In the extrapolation setting, we train \dbname using data from 2013–2016 and generate synthetic data for 2017 (denoted as 2017*). 
In the interpolation setting, we train using 2013–2017 excluding 2015, and generate 2015* to fill the gap. 
For illustrative purposes, we first analyze the results on the IMDB dataset, followed by those on the STACK dataset.

\paragraphtitle{Drift Analysis}
Figure~\ref{subfig.drift_imdb_17} and Figure~\ref{subfig.drift_imdb_15} show that the synthesized snapshots exhibit drift factors highly similar to real snapshots, confirming that the generated data induces realistic distributional drift.
We further report the per-column distribution of 2017* in Figure~\ref{fig.generation_distribution}, which shows that the generated data exhibits comparable distribution patterns and similar density to the real data.

\paragraphtitle{Correlation Preservation}
Figure~\ref{subfig.corr_imdb_17} and Figure~\ref{subfig.corr_imdb_15} report the correlation difference between generated data and real data.
We observe that the generated snapshots maintain low correlation differences, comparable to those of real data, indicating that \dbname effectively preserves attribute dependencies during generation.



\paragraphtitle{Downstream Performance}
To evaluate the realism of the generated drift from a workload perspective, we train learned query optimizers (Bao, Balsa, Lero) on IMDB and evaluate them on both real (e.g., 2017) and generated (e.g., 2017*) datasets. 
We report the accumulated query execution time in Figure~\ref{subfig.exec_imdb_17} and Figure~\ref{subfig.exec_imdb_15}.
We can observe that all evaluated systems exhibit performance degradation on both real and synthetic drifted data. 
This is expected, as learned query optimizers may suffer from suboptimal plan choices due to outdated knowledge about data distributions, which we will further evaluate in Section~\ref{subsec:eval:lqo_drift}.
Further, the degradation patterns are highly similar, indicating that the generated data successfully simulates the impact of real-world drift on query performance.




We also perform experiments in generating drifted STACK data.
As shown in Figure~\ref{subfig.drift_stack_10} and Figure~\ref{subfig.corr_stack_10}, \dbname is able to synthesize drifted data with a relatively low correlation difference.
In addition, Figure~\ref{subfig.exec_stack_10} demonstrates similar performance degradation for learned query optimizers under both synthesized and real drift.

\noindent\expmessage{\dbname can synthesize data with realistic distribution drift and strong correlation fidelity, which in turn enables faithful downstream performance evaluation across multiple datasets.}

\subsubsection{Experiments on generating drifted datasets under varying drift factors}
\label{subsec:drift_severity}
We compare the data generation quality of \dbname against RelDDPM and CTGAN.
Following RelDDPM~\cite{DBLP:journals/pacmmod/LiuFTLD24}, we evaluate on joined IMDB tables and measure the correlation difference between the generated and real tables to quantify drift quality.
Specifically, we construct two joined tables:
1) \texttt{Movie} $\bowtie$ \texttt{Director}, where movies are linked to directors via director ID; and 2) \texttt{Movie} $\bowtie$ \texttt{Actor}, where movies are linked to actors via actor ID.
We plot the results of these two joint tables in Figure~\ref{subfig.generation1} and Figure~\ref{subfig.generation2}, respectively.
We can observe that \dbname consistently achieves lower average correlation differences across all drift factors.
RelDDPM fails to generate data for drift factors greater than 0.4, and the correlation difference fluctuates unpredictably, showing no clear relationship with the drift factor.
CTGAN produces data with less consistent correlation control. Its conditional generation mechanism is not tailored for modeling data drift, making it difficult to produce a controlled spectrum of drifted distributions.

\noindent{\expmessage{ \dbname effectively synthesizes data across a wide range of drift factors with relatively low correlation differences.}
\vspace{-1mm}

\subsubsection{Generation Quality Analysis}
\label{subsec:snapshot_quality}
We test the factors that can affect the data generation quality.

We first evaluate the effects of drift severity.
We generate a series of synthetic datasets from IMDB with drift factors $d \in {0.1, 0.3, 0.5}$, representing mild, medium, and severe drift.
We keep the data size the same as the original IMDB.
As shown in Figure~\ref{subfig.drift_severity}, the correlation difference increases with drift factor, indicating that preserving correlations becomes more challenging under more severe drift.
However, as shown in Figure~\ref{subfig.generation1} and Figure~\ref{subfig.generation2}, the average correlation difference rises nearly linearly with increased drift.
The results confirm that
\dbname can effectively synthesize drifted data while preserving correlations in a controlled manner.

\begin{figure}[t]
\centering
\subfigure[\texttt{Movie} $\bowtie$ \texttt{Director}]{
\includegraphics[height=0.12\textwidth]{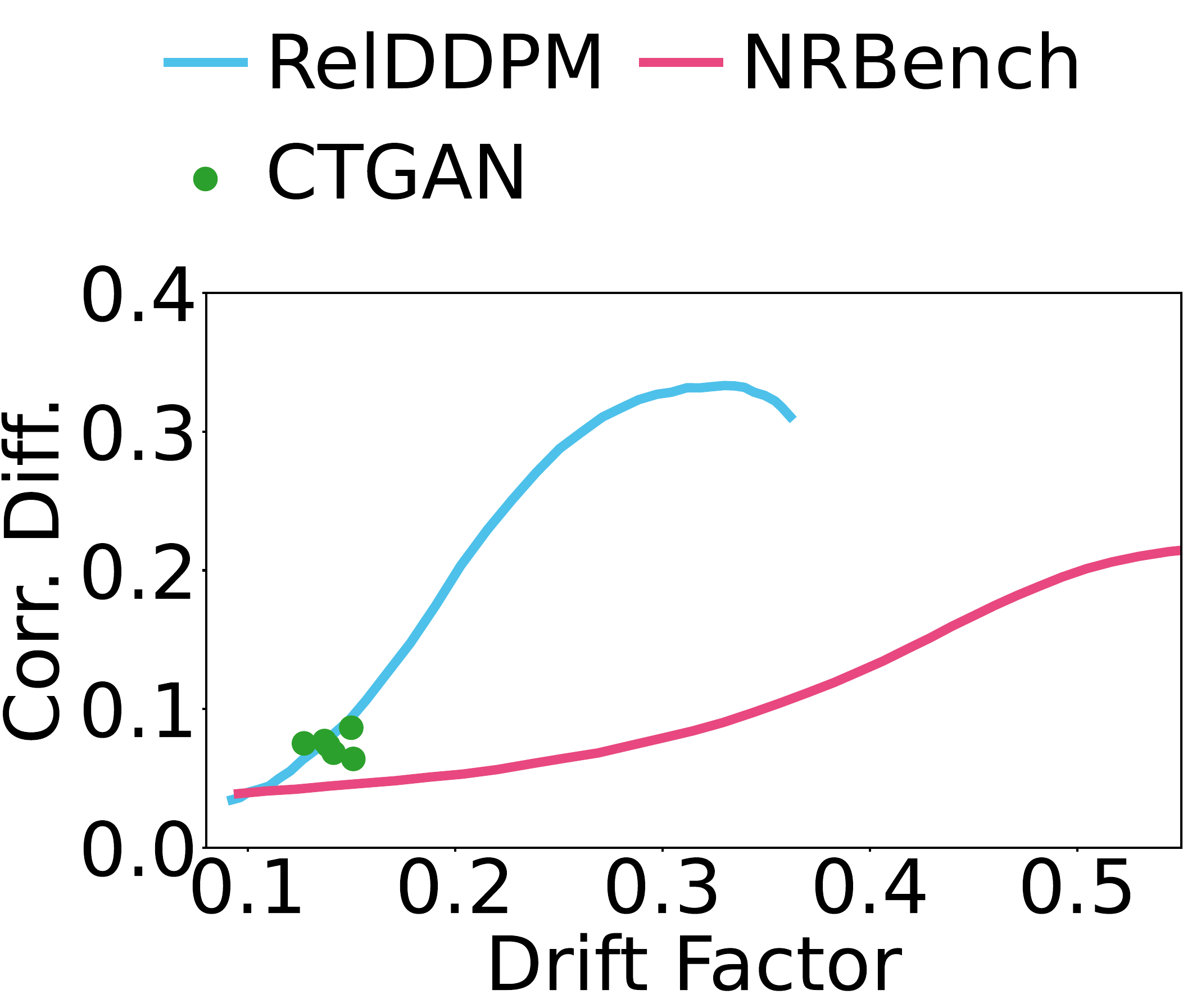}
\label{subfig.generation1}
}
\subfigure[\texttt{Movie} $\bowtie$ \texttt{Actor}]{
\includegraphics[height=0.12\textwidth]{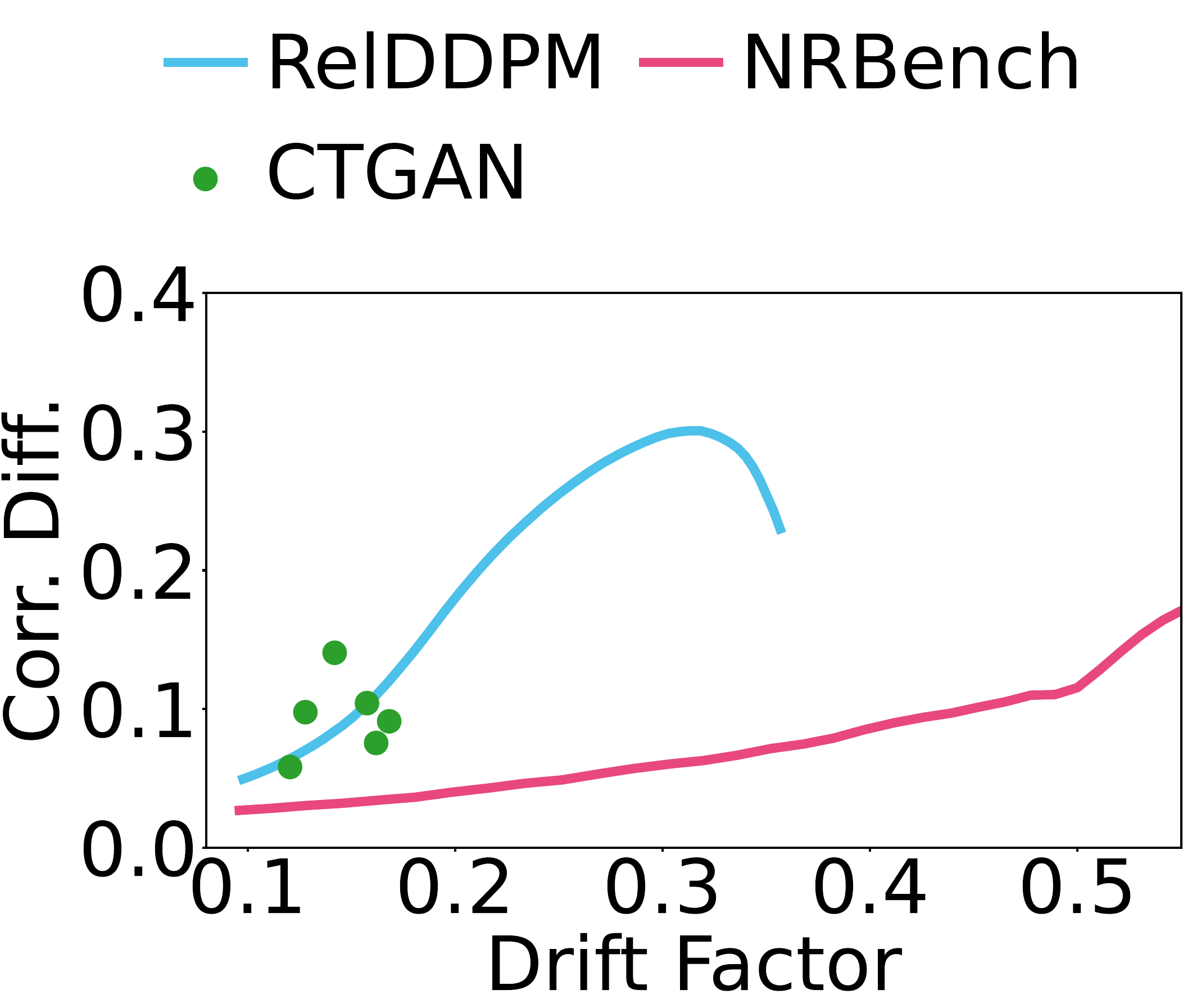}
\label{subfig.generation2}
}
\subfigure[Drift Severity]{
\includegraphics[height=0.12\textwidth]{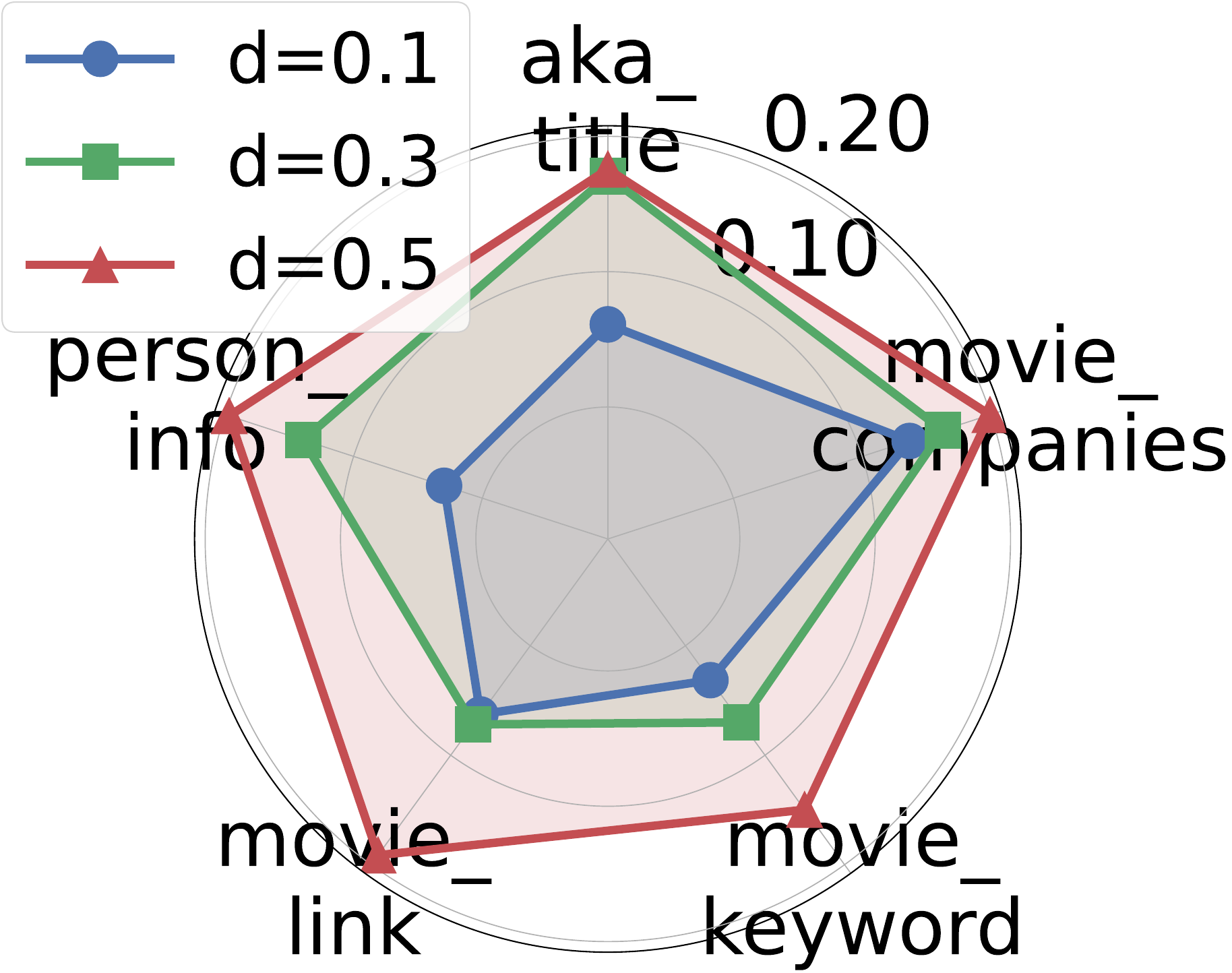}
\label{subfig.drift_severity}
}
\subfigure[Snapshot Quality]{
\label{subfig.snapshot_quality}
\includegraphics[height=0.12\textwidth]{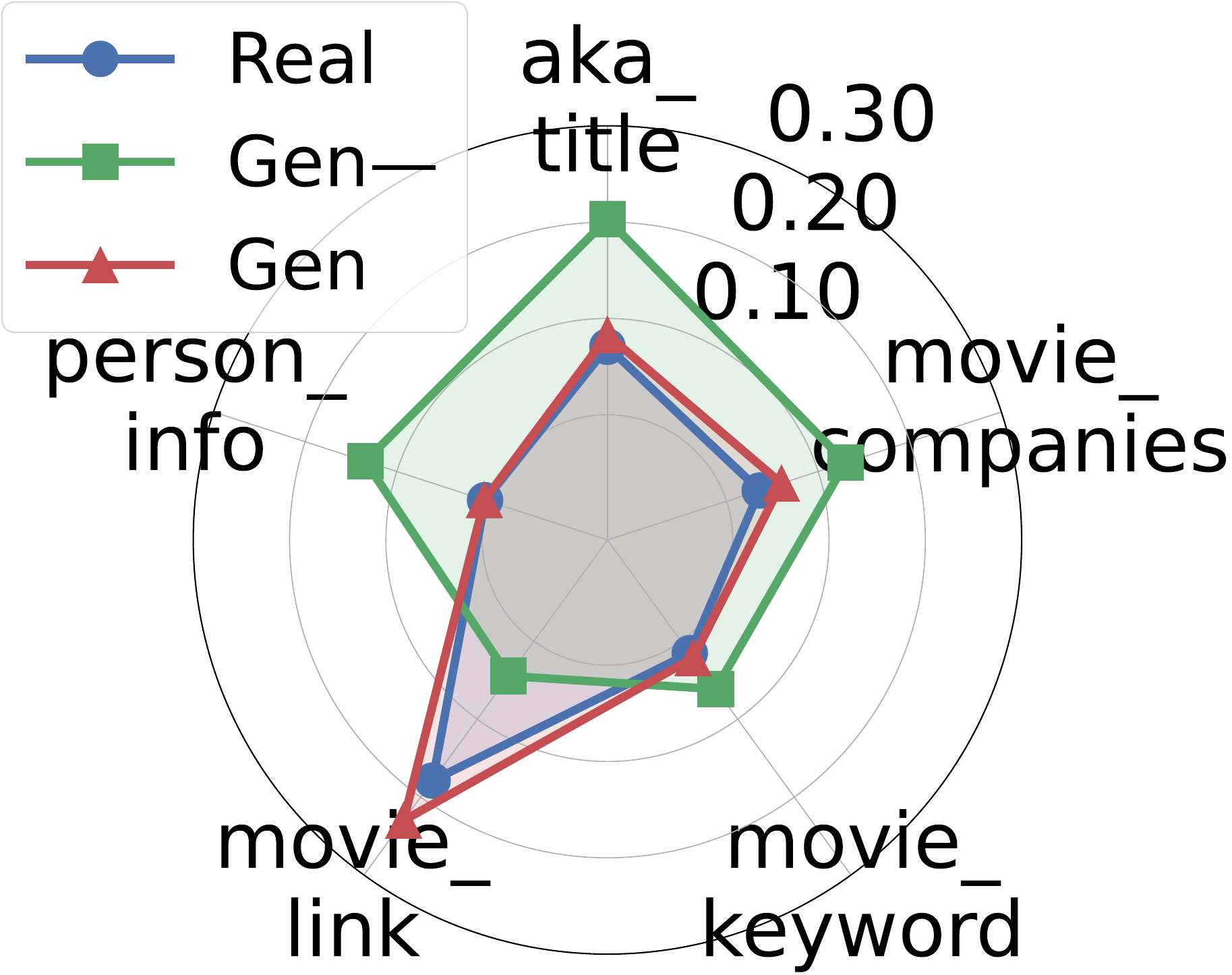}
}
\subfigure[Overhead]{
\label{subfig.overhead_data_size}
\includegraphics[height=0.12\textwidth]{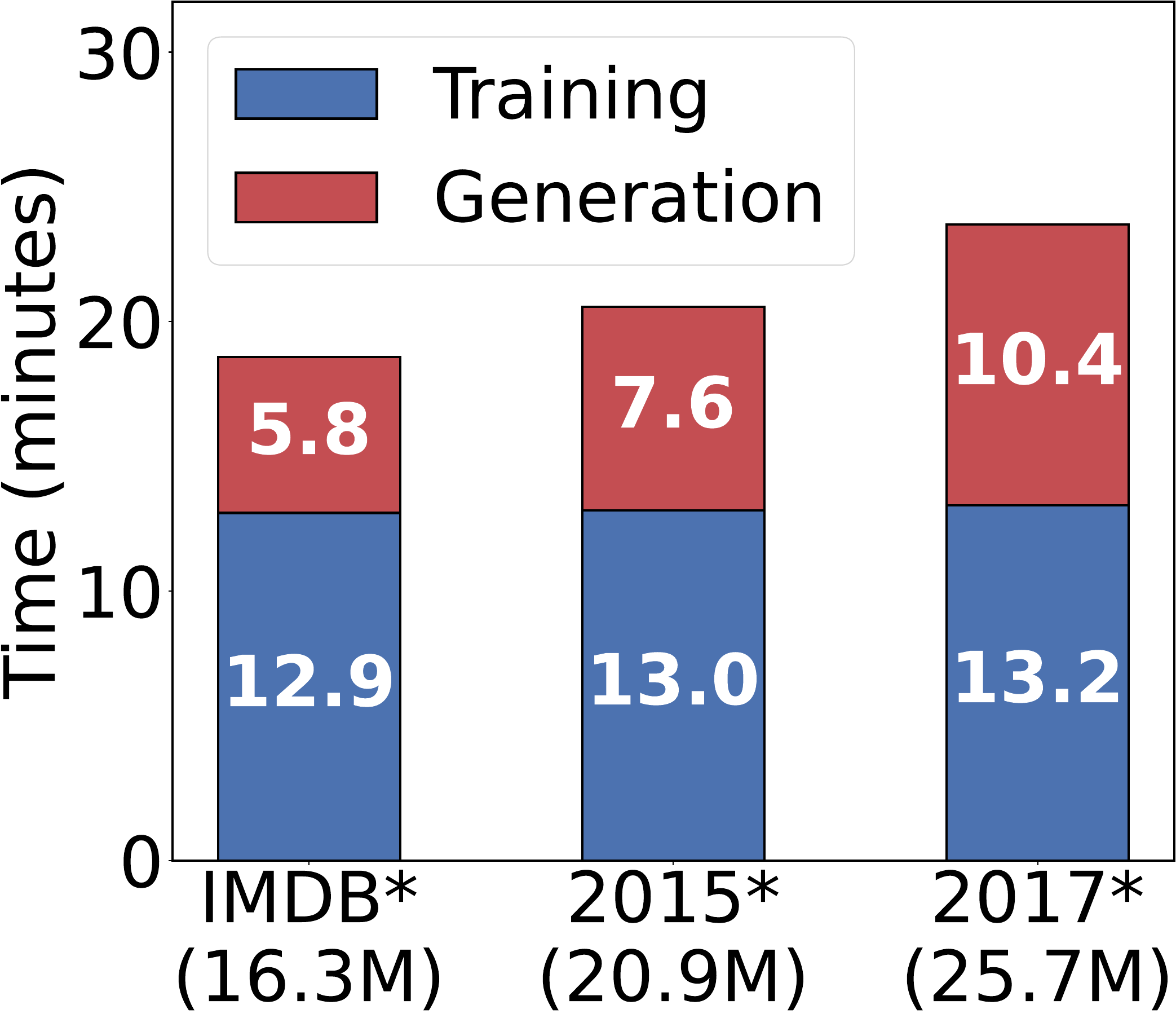}
}
\subfigure[Overhead]{
\label{subfig.overhead_gpu}
\includegraphics[height=0.12\textwidth]{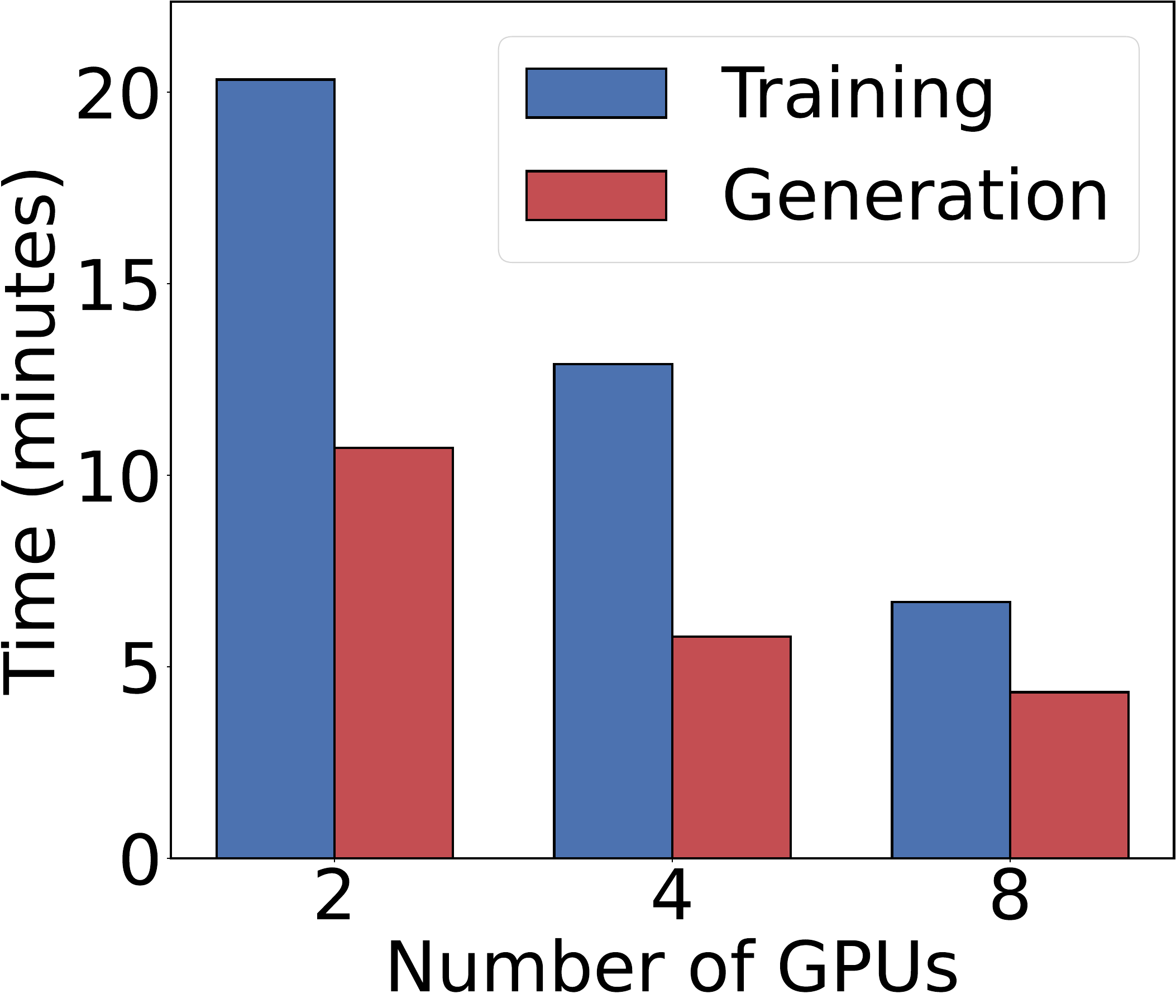}
}
\vspace{-4mm}
\captionsetup{justification=raggedright}
\caption{Overhead and In-depth Analysis}
\label{fig.generation_analysis}
\vspace{-4mm}
\end{figure}

We then evaluate the effects of snapshot quality.
With the same setup as in Figure~\ref{subfig.corr_imdb_17}, we additionally include a variant (Gen--) that omits the 2016 snapshot, representing a case without a high-quality snapshot.
As shown in Figure~\ref{subfig.snapshot_quality}, Gen-- exhibits higher correlation differences compared to Gen, confirming that high-quality snapshots improve the correlation preservation ability.

\noindent\expmessage{The drift data generation quality can be affected by the drift severity and snapshot quality.}

\subsubsection{Overhead Analysis}
\label{subsec:overhead_analysis}
We evaluate the time overhead incurred by drift generator training and drift data generation.
We first study the overhead under varying data sizes.
IMDB* refers to the IMDB dataset generated with different drift factors.
As shown in Figure~\ref{subfig.overhead_data_size}, generation time increases from 5.8 minutes on IMDB* (16.3M tuples) to 10.4 minutes on IMDB 2017* (25.7M tuples).
This is due to the fact that the generator produces one batch of tuples at a time, and as the dataset grows, more batches are required, leading to a near-linear increase in generation time.
In contrast, the training time remains stable ($\sim$13 minutes), as it depends solely on the number of training steps, which we fix to 500 across all experiments.
We then evaluate the overhead under varying GPU resources, and plot the results in Figure~\ref{subfig.overhead_gpu}.
As observed, both training and generation times decrease as the number of GPUs increases from 2 to 8.
This is expected, as batch-based generation can be effectively parallelized across multiple GPUs, demonstrating good scalability of our approach.

\noindent\expmessage{The factors that determine the generation overhead are data sizes and GPU resources.}


\subsubsection{Data Drift Generator Extensibility}
\label{subsec:gen_extended_drift}
We study whether a drifter trained on one dataset can be transferred to another dataset. 
Specifically, we apply a drifter trained on the STACK dataset to an IMDB diffuser to generate drifted IMDB data, denoted as IMDB(S)*.
As shown in Figure~\ref{subfig.drift_genertor_entended}, this transfer is feasible because the drifter predicts noise at each denoising step of the diffuser. 
Since the diffusion process starts from a Gaussian distribution, injecting such noise still induces meaningful distributional changes on the target dataset.
We further examine whether IMDB data can be extrapolated to approximate a snapshot with incomplete information. 
We collect partial IMDB data for 2019 from~\cite{imdb} and use it as a snapshot to train the data generator, synthesizing 2019*. 
As shown in Figure~\ref{subfig.drift_genertor_entended}, 2019* leads to more severe performance degradation than 2017*, indicating that \dbname can synthesize stronger drift over a larger temporal span.


\noindent\expmessage{A drifter trained on one dataset can be transferred to another dataset to enable drift generation.}

\begin{figure}[t]
\centering
\subfigure[{Data Drift Generator}]{
\includegraphics[width=0.45\linewidth]{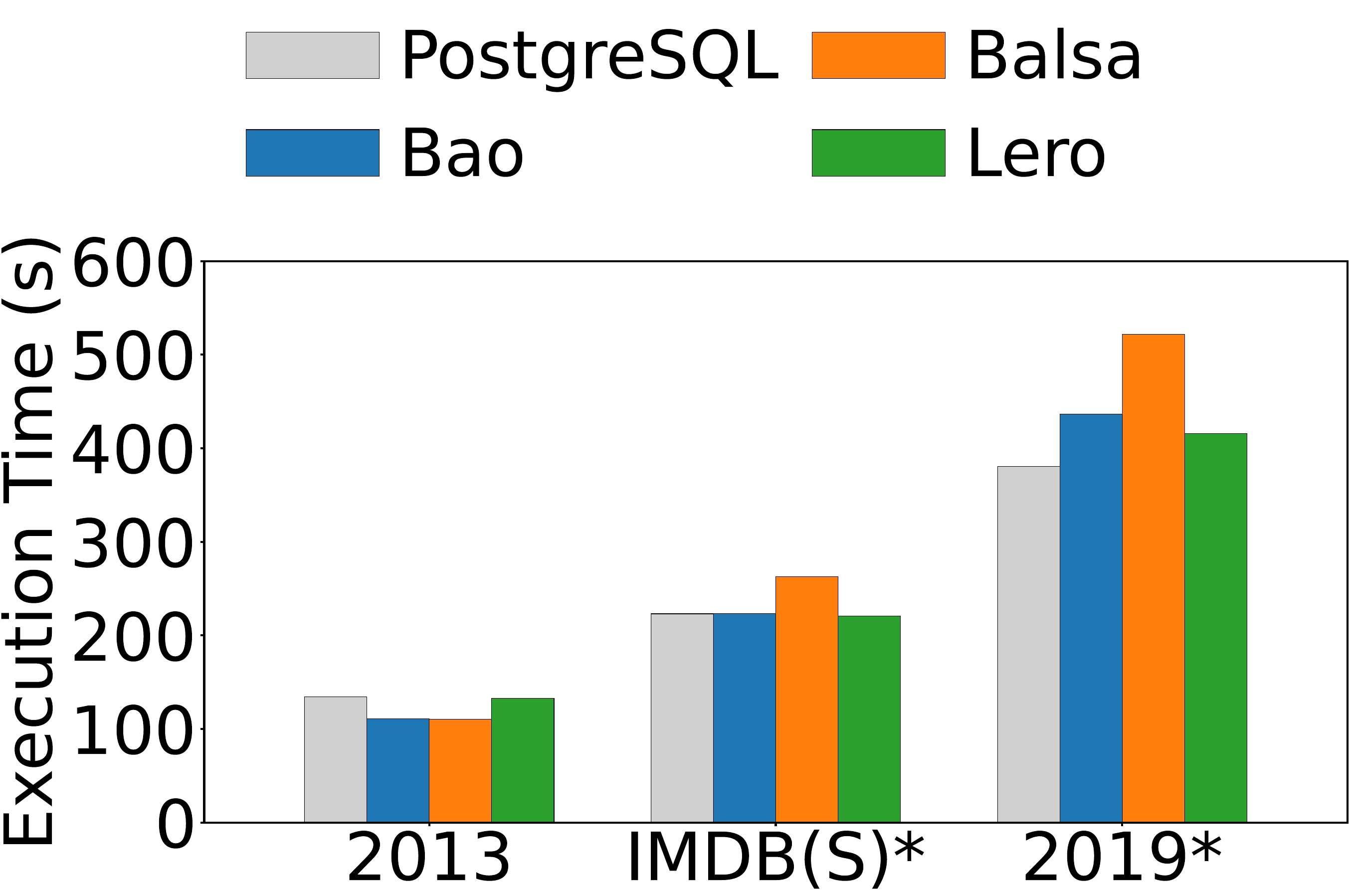}
\label{subfig.drift_genertor_entended}
}\hspace{3mm}
\subfigure[Workload Drift]{
\includegraphics[width=0.45\linewidth]{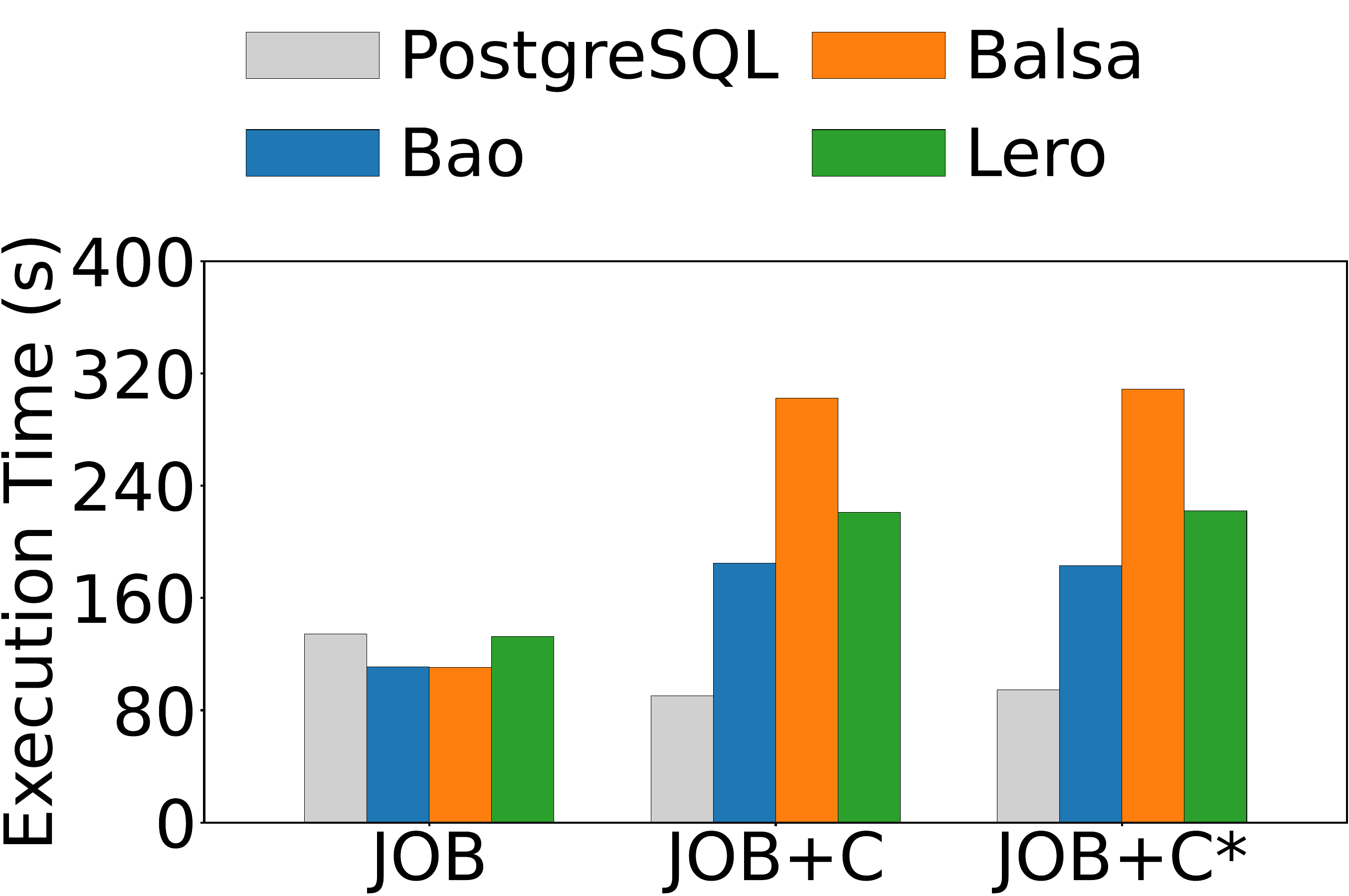}
\label{subfig.real_lqo_wl_ceb}
}
\vspace{-4mm}
\captionsetup{justification=raggedright}
\caption{Extended Data Drift and Workload Drift Analysis}
\label{fig.extended_drift_gen}
\vspace{-4mm}
\end{figure}

\subsubsection{Experiments on the quality of synthesized workload drift}

We validate the effectiveness of \dbname in generating realistic workload drift.
Following Krid et al.\cite{redbenchimpl}, we adopt JOB+C, a hybrid workload composed of 46 queries from JOB and 66 from CEB~\cite{DBLP:journals/pvldb/NegiMKMTKA21}, constructed based on workload pattern mappings derived from real Redshift traces~\cite{DBLP:journals/pvldb/RenenL23}.
We use \dbname to synthesize JOB+C*, aiming to mimic the drift from JOB to JOB+C.
The workload generator of \dbname is trained with a query set collected from~\cite{DBLP:journals/pvldb/NegiMKMTKA21,joblightimpl,DBLP:journals/pvldb/YuC0L22}, excluding those used in JOB+C.
We fix the dataset to the default IMDB, train each system on JOB, and test on both JOB+C and JOB+C*. 
As shown in Figure~\ref{subfig.real_lqo_wl_ceb}, the performance on both workloads is highly comparable, indicating that \dbname effectively captures realistic workload drift.
Moreover, the effectiveness of learned query optimizers reduces under drifted workloads, which we will further investigate in Section~\ref{subsec:eval:ablation_workload_drift}.
 

\begin{figure*}[t]
\centering
\subfigure[Drift Factor]{
\label{subfig.lqo_drift_factor}
\includegraphics[height=0.163\textwidth]{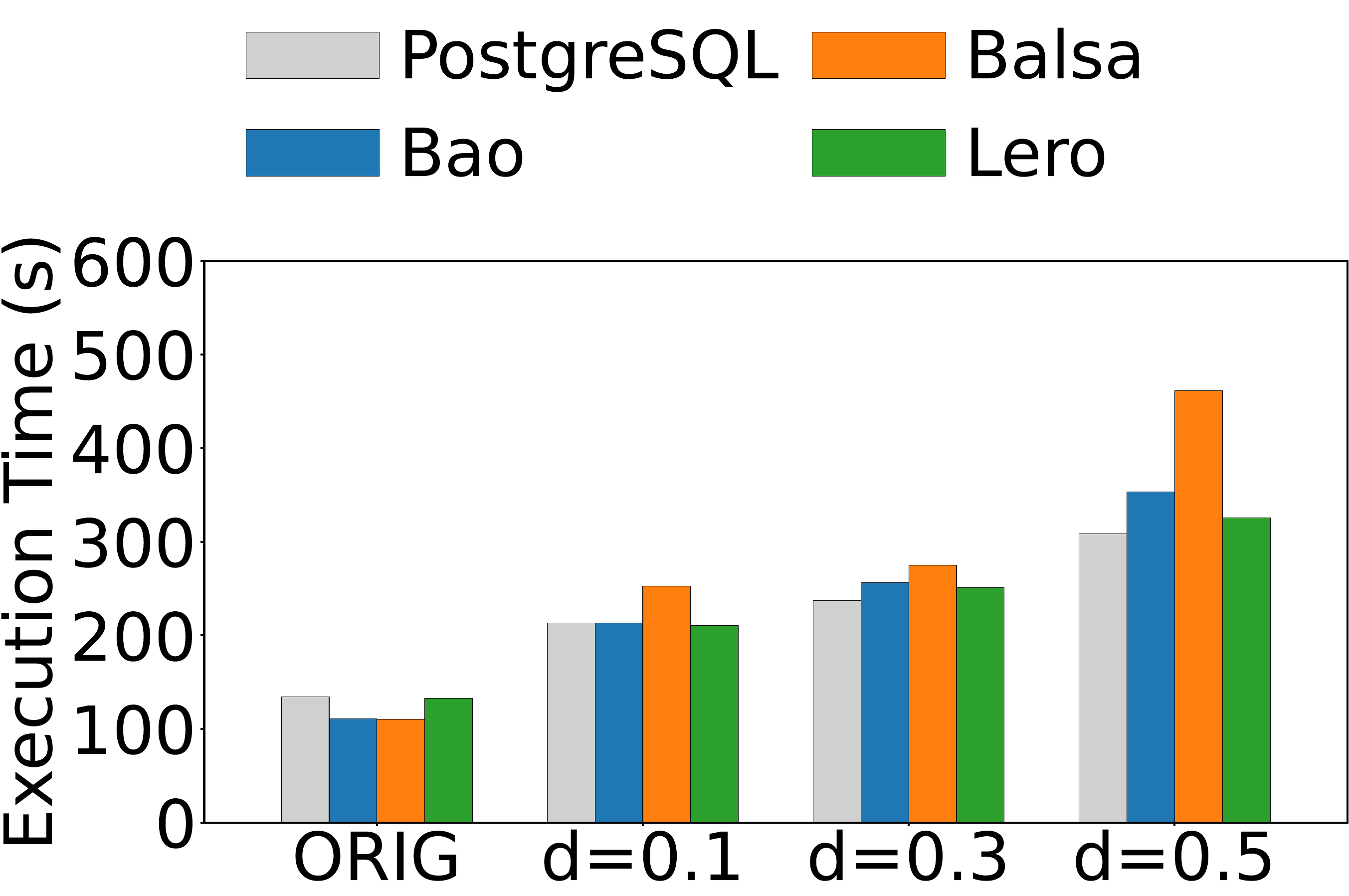}
}
\begin{minipage}[b]{0.7\textwidth}
    \centering
    \includegraphics[height=0.63cm]{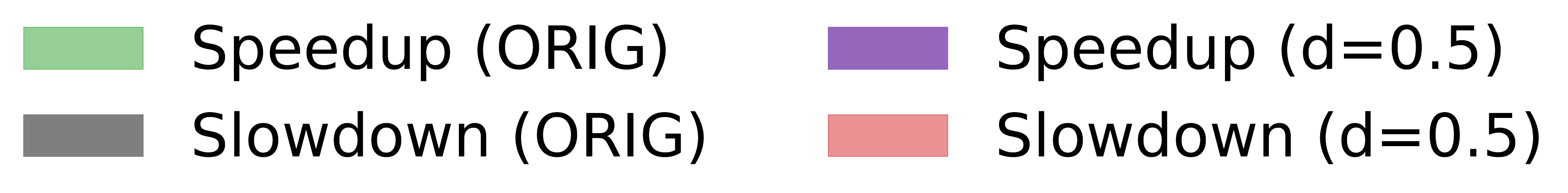}\\
    \subfigure[Bao]{
        \includegraphics[height=0.17\textwidth]{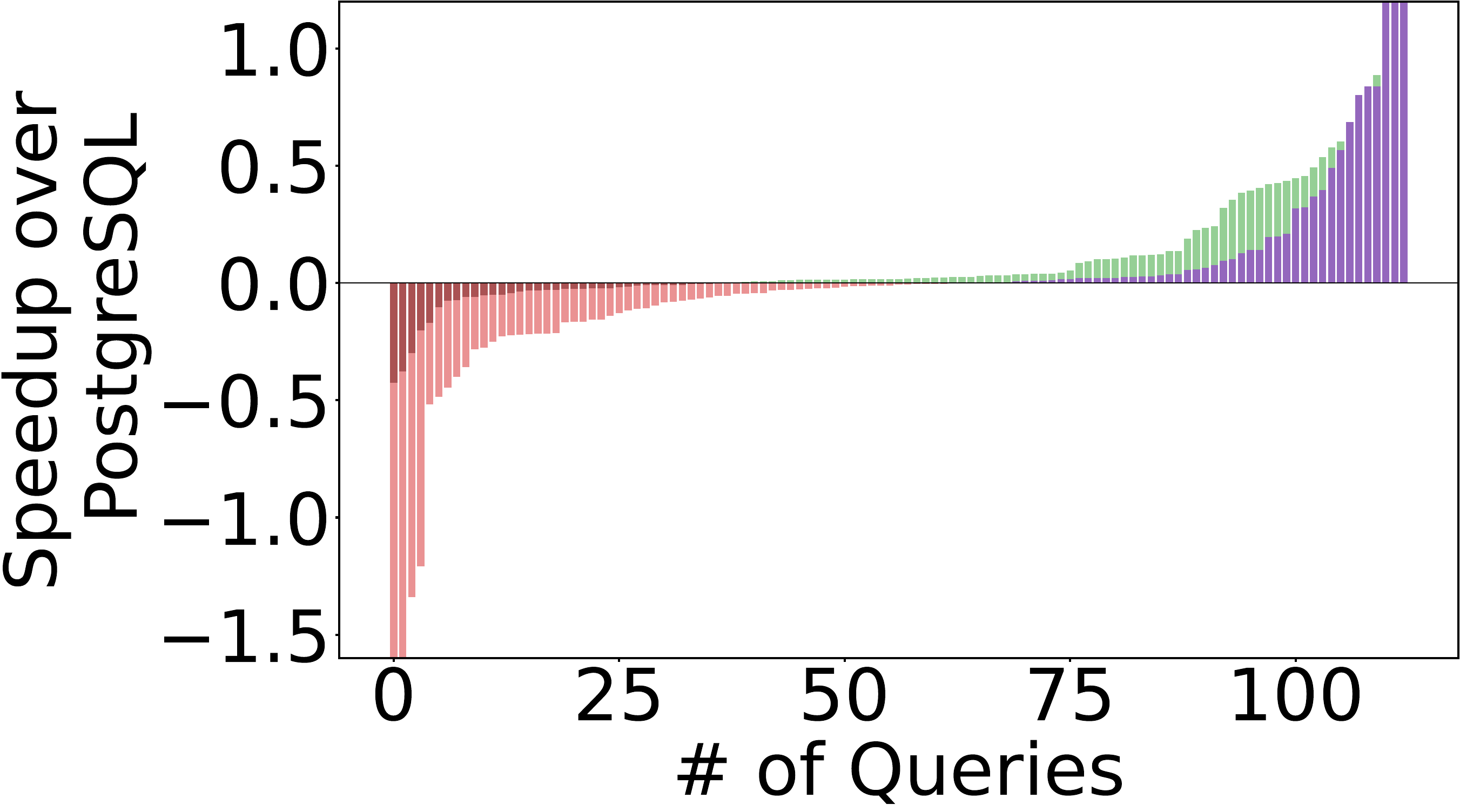}
        \label{subfig.ablation_bao}
    }
    \subfigure[Balsa]{
        \includegraphics[height=0.17\textwidth]{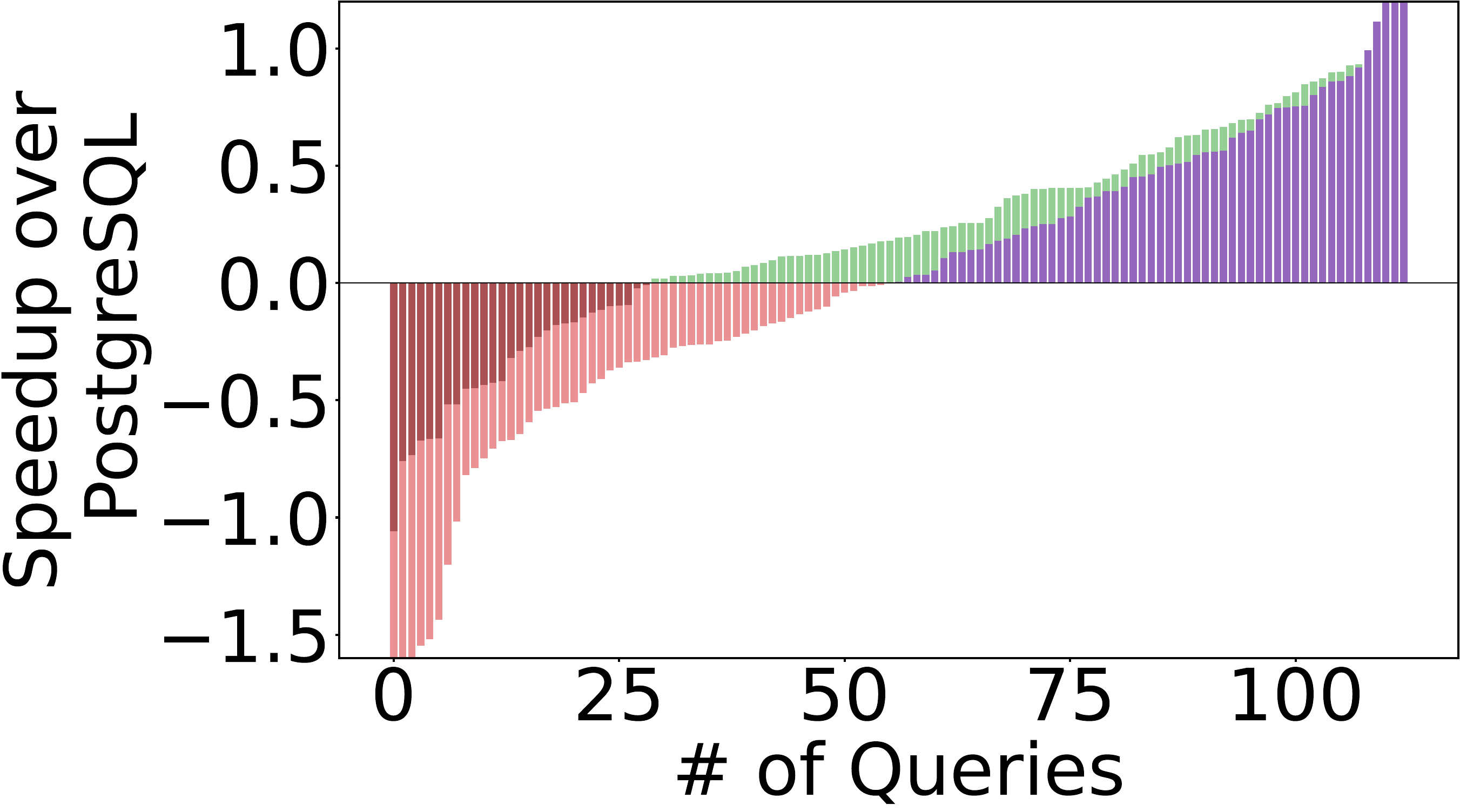}
        \label{subfig.ablation_balsa}
    }
    \subfigure[Lero]{
        \includegraphics[height=0.17\textwidth]{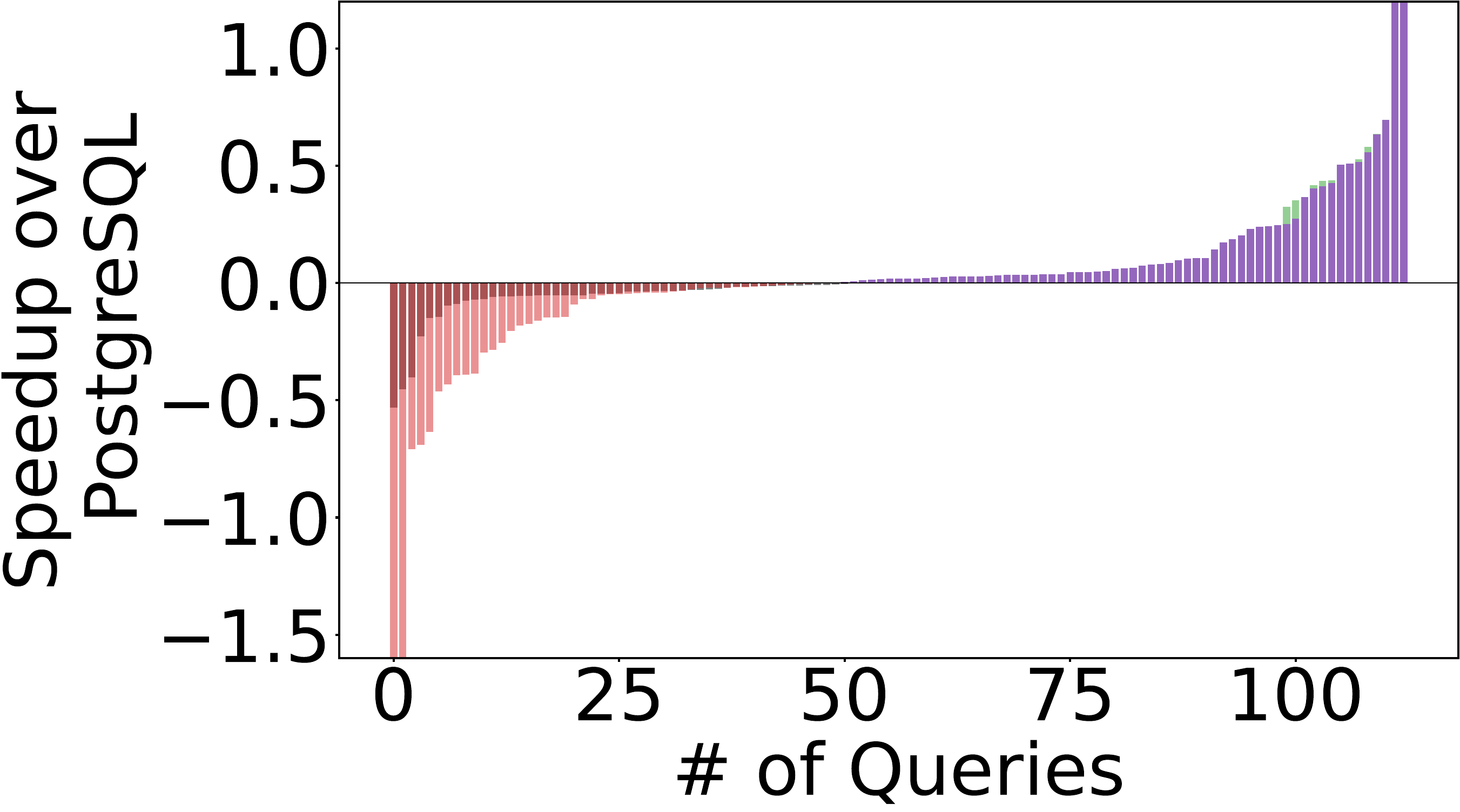}
        \label{subfig.ablation_lero}
    }
\end{minipage}
\begin{minipage}[b]{\linewidth} 
    \centering
    \includegraphics[height=0.021\textwidth]{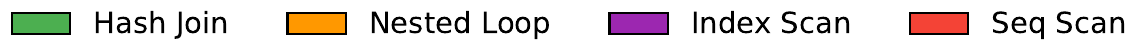} \\
    \subfigure[PostgreSQL (Pre-Drift)]{
        \label{subfig.pg_plan_pre_drift}
        \includegraphics[height=0.16\textwidth]{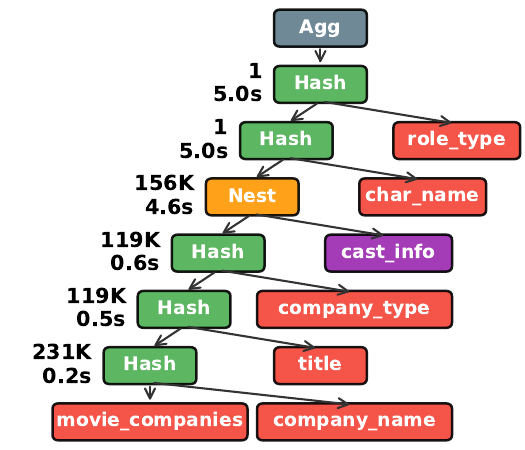}
    }
    \subfigure[Bao (Pre-Drift)]{
        \label{subfig.bao_plan_pre_drift}
        \includegraphics[height=0.16\textwidth]{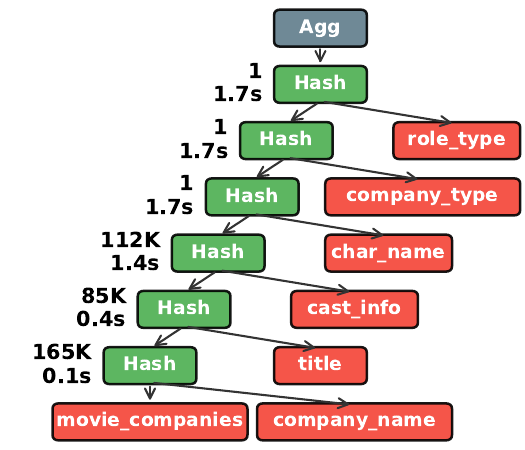}
    }
    \subfigure[PostgreSQL (Post-Drift)]{
        \label{subfig.pg_plan_post_drift}
        \includegraphics[height=0.16\textwidth]{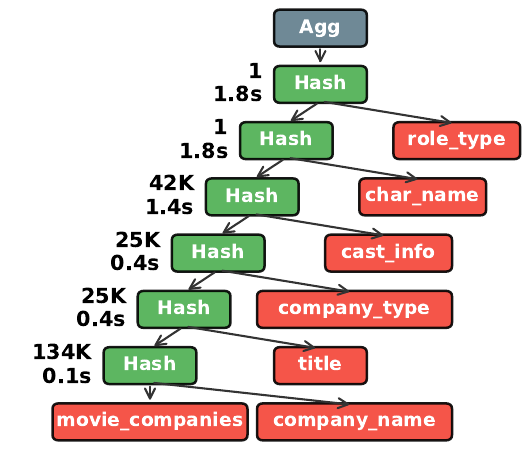}
    }
    \subfigure[Bao (Post-Drift)]{
        \label{subfig.bao_plan_post_drift}
        \includegraphics[height=0.16\textwidth]{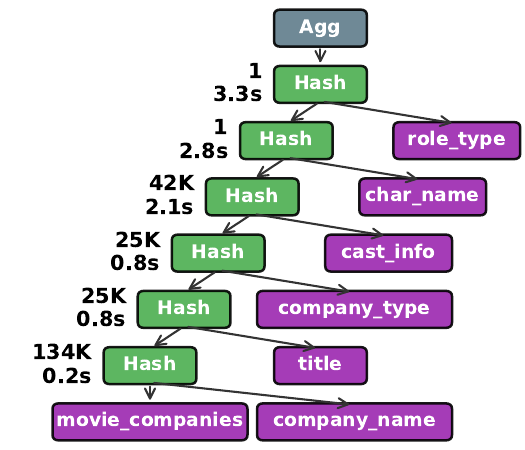}
    }
    \subfigure[Balsa (Pre- and Post-Drift)]{
        \label{subfig.balsa_plan}
        \includegraphics[height=0.16\textwidth]{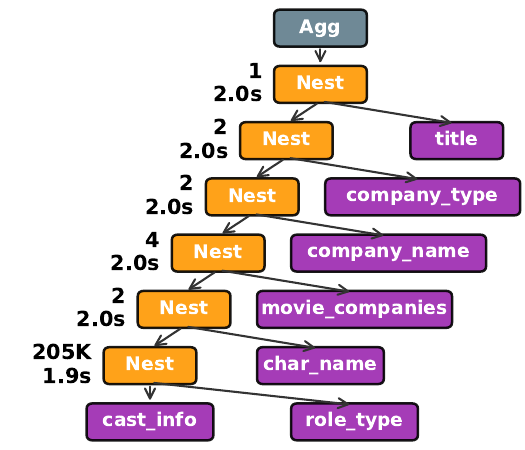}
    }
\end{minipage}
\vspace{-4mm}
\captionsetup{justification=raggedright}
\caption{{Evaluation of Learned Query Optimizers under Synthetic Data Drift with Various Drift Factors}}
\vspace{-4mm}
\label{fig.eval.lqo_data_drift}
\end{figure*}

\subsection{Experiments on Learned Query Optimizers under Synthetic Data Drift}
\label{subsec:eval:data_drift}
\label{subsec:eval:lqo_drift}

In this set of experiments, we evaluate the selected learned query optimizers under the drift synthesized by \dbname.


\subsubsection{Experiments under synthesized data drift with various drift factors}
\label{subsec:eval:lqo_drift_factor}
We construct synthesized datasets with gradually increasing drift factors to evaluate the impact of data drift. 
By default, we use the drifted datasets generated in Section~\ref{subsec:snapshot_quality}.
We train each system on the original IMDB (denoted as ORIG) and evaluate on datasets with increasing $d$.
As shown in Figure~\ref{subfig.lqo_drift_factor}, the execution time increases with larger drift values.
This aligns with our expectation that models trained on less similar data distributions experience more severe staleness.

\label{subsec:eval:ablation_drift}

\label{subsubsec:eval:acc}
To better understand the performance degradation caused by data drift (Figure~\ref{subfig.lqo_drift_factor}), we compute model accuracy metrics and report the results in Table~\ref{tab:lqo_eval:acc_synthesized}.
We use the median Q-error to measure the ratio between the estimated and true execution time~\cite{DBLP:conf/sigmod/KimJSHCC22,DBLP:journals/pvldb/HanWWZYTZCQPQZL21}.
Top-3 AUC measures how often the best plan appears in the top 3 results ranked by the model.
We observe that learned query optimizers generally suffer from reduced model accuracy under data drift, which aligns with their overall performance degradation. 
Among them, Balsa shows the most significant decline, as it follows a black-box design that fully depends on the model's predictions to generate query plans.
In contrast, white-box approaches such as Bao embed learned hints into the optimizer, and therefore, they still partially leverage the DBMS's native cardinality estimates.

\begin{table}
\renewcommand\arraystretch{1.1}
\caption{Model Accuracy of Learned Query Optimizers}
\label{tab:lqo_eval:acc_synthesized}
\vspace{-2mm}
\resizebox{0.86\linewidth}{!}{
\begin{tabular}{llllll}
\toprule
Competitors & Metrics    & ORIG & $d$=0.1 & $d$=0.3 & $d$=0.5 \\
\midrule
Bao        & \multirow{2}{*}{Median Q-error} & 1.184 & 1.489 & 2.057 & 2.671 \\
Balsa      &                      & 1.723 & 1.974 & 2.156 & 3.512 \\
\hline
Lero       & Top-3 AUC  & 0.646 & 0.478 & 0.407 & 0.389 \\
\bottomrule
\end{tabular}
}
\vspace{-4mm}
\end{table}

To further validate this, we analyze the speedup of query execution times under data drift, and plot the speedup in Figure~\ref{subfig.ablation_bao}, Figure~\ref{subfig.ablation_balsa} and Figure~\ref{subfig.ablation_lero}. 
We observe that Balsa exhibits much higher variance in speedup than Bao, confirming its greater sensitivity to drift.
Interestingly, for certain queries, Balsa even outperforms PostgreSQL by a larger margin, suggesting that fine-grained black-box models can still identify better plans when they generalize well.
We take a closer look at the plans generated by each system for a representative query, q10c. Before the drift, as shown in Figure~\ref{subfig.pg_plan_pre_drift} and Figure~\ref{subfig.bao_plan_pre_drift}, Bao improves performance by applying the hint `SET enable\_indexscan TO off'. 
However, after the drift (Figure~\ref{subfig.pg_plan_post_drift} and Figure~\ref{subfig.bao_plan_post_drift}), Bao switches to disabling nested loop joins with `SET enable\_nestloop TO off'. 
Despite this change, the join order selected by Bao remains largely similar to that of PostgreSQL because Bao relies on PostgreSQL's initial plan structure.
In contrast, Balsa directly predicts join orders via its model, which allows it to produce plans that differ significantly from both PostgreSQL and Bao. 
Notably, as shown in Figure~\ref{subfig.balsa_plan}, Balsa generates the same plan before and after the drift, suggesting that the model failed to adapt to distributional changes. 
This static behavior under drift likely leads to mispredicted plans and performance degradation.

\noindent
\expmessage{White-box approaches can achieve better robustness. Black-box approaches, with larger search spaces, offer greater potential robustness but are often limited by current model capacity.}
\subsubsection{Regression model vs. Ranking model}
\label{subsec:eval:abl:ranking}
We further study the effect of the regression model and the ranking model on the performance of white-box approaches.
We implement Bao+R, a variant of Bao that employs Lero's ranking model while retaining Bao's hint-based plan search strategy.
Therefore, comparing the performance degradation of Bao+R and Bao can directly evaluate the robustness of the ranking model.
As shown in Figure~\ref{subfig.ablation_ranking}, Bao+R suffers less performance degradation than Bao, suggesting that the ranking model improves robustness under drift.
However, because the ranking model does not explicitly estimate predicate costs, its performance under non-drift scenarios is less competitive compared to Bao's regression-based model.

\noindent\expmessage{Ranking-based models tend to be slightly more robust under data drift but may struggle to achieve optimal performance in stable settings.}




\begin{figure}[t]
\centering
\subfigure[Ranking Model]{
\includegraphics[height=0.13\textwidth]{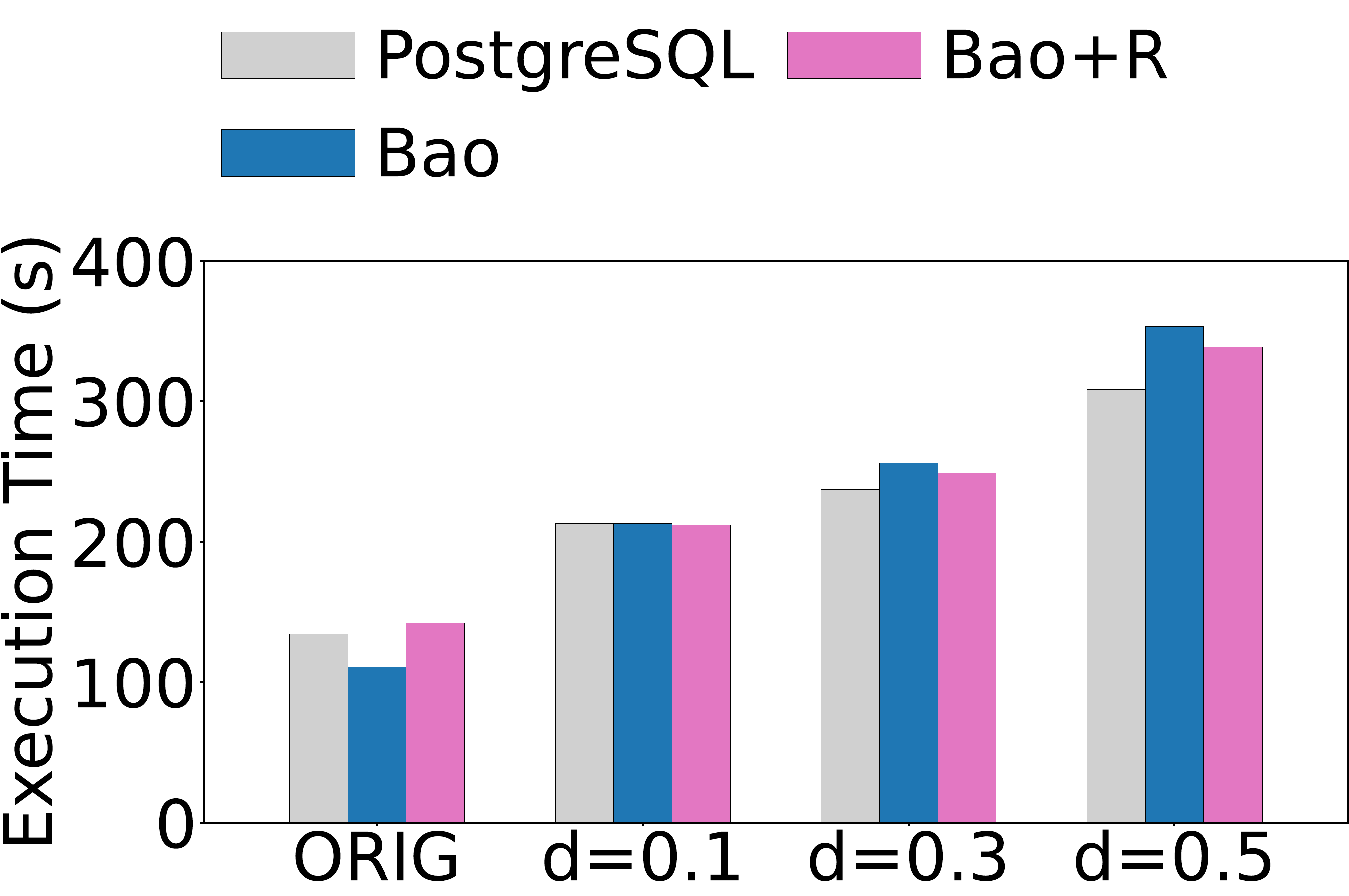}
\label{subfig.ablation_ranking}
}
\subfigure[Data Size]{
\label{subfig.lqo_data_size}
\includegraphics[height=0.13\textwidth]{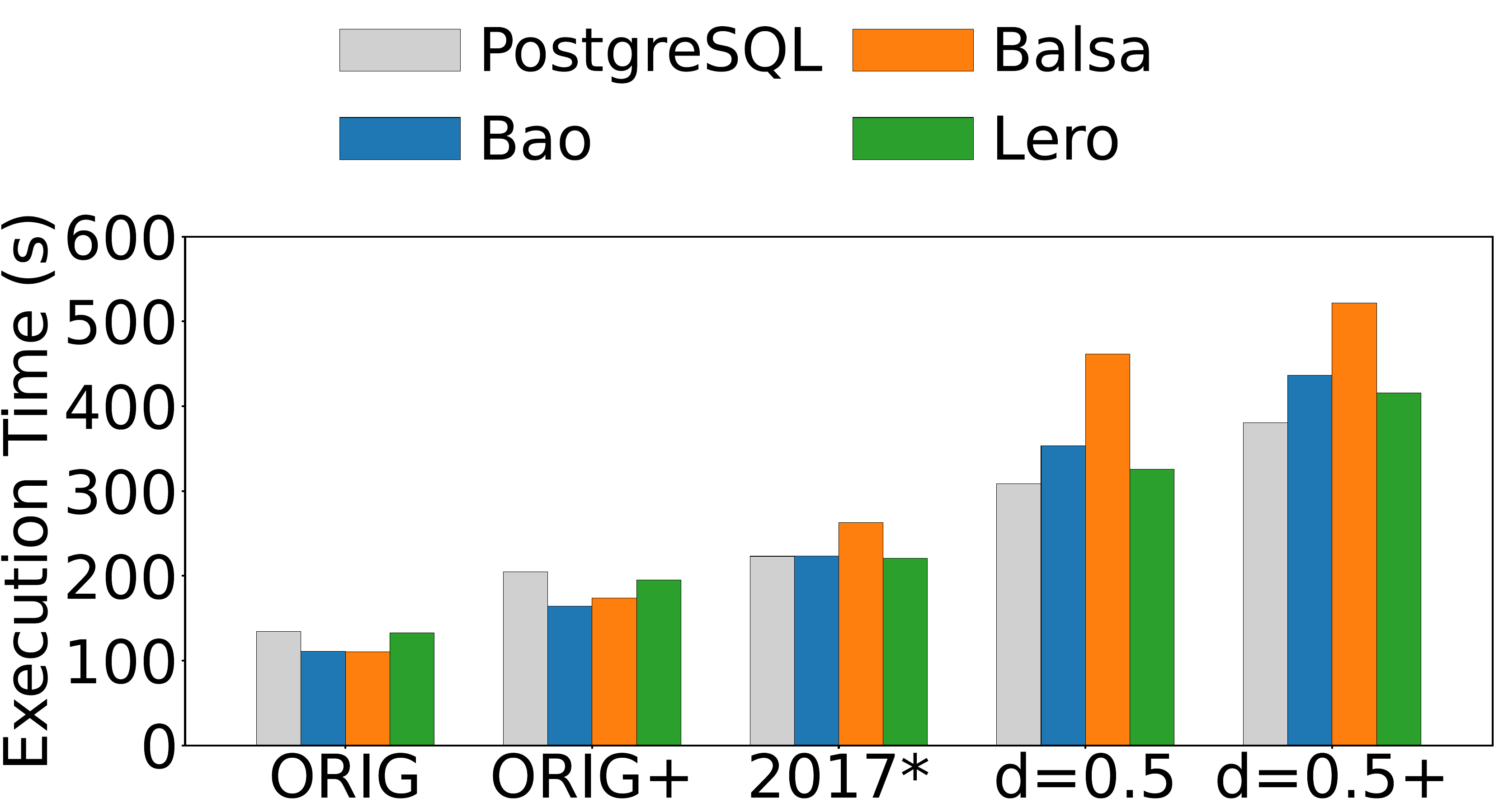}
}
\subfigure[Comparison with D-LQO]{
\label{subfig.lqo_extended}
\includegraphics[height=0.13\textwidth]{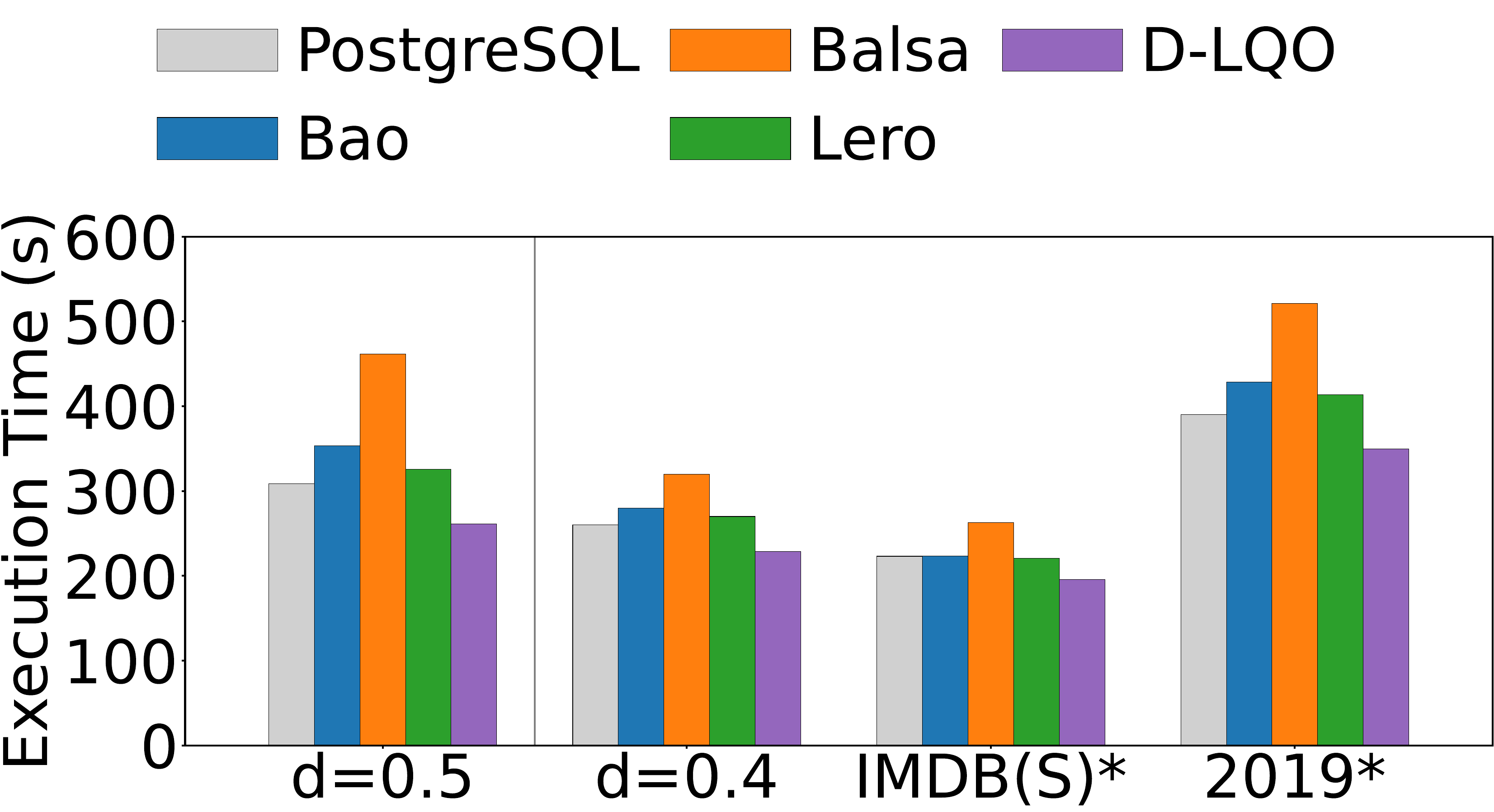}
}
\subfigure[Periodical Retraining]{
\label{subfig.lqo_retrain}
\includegraphics[height=0.13\textwidth]{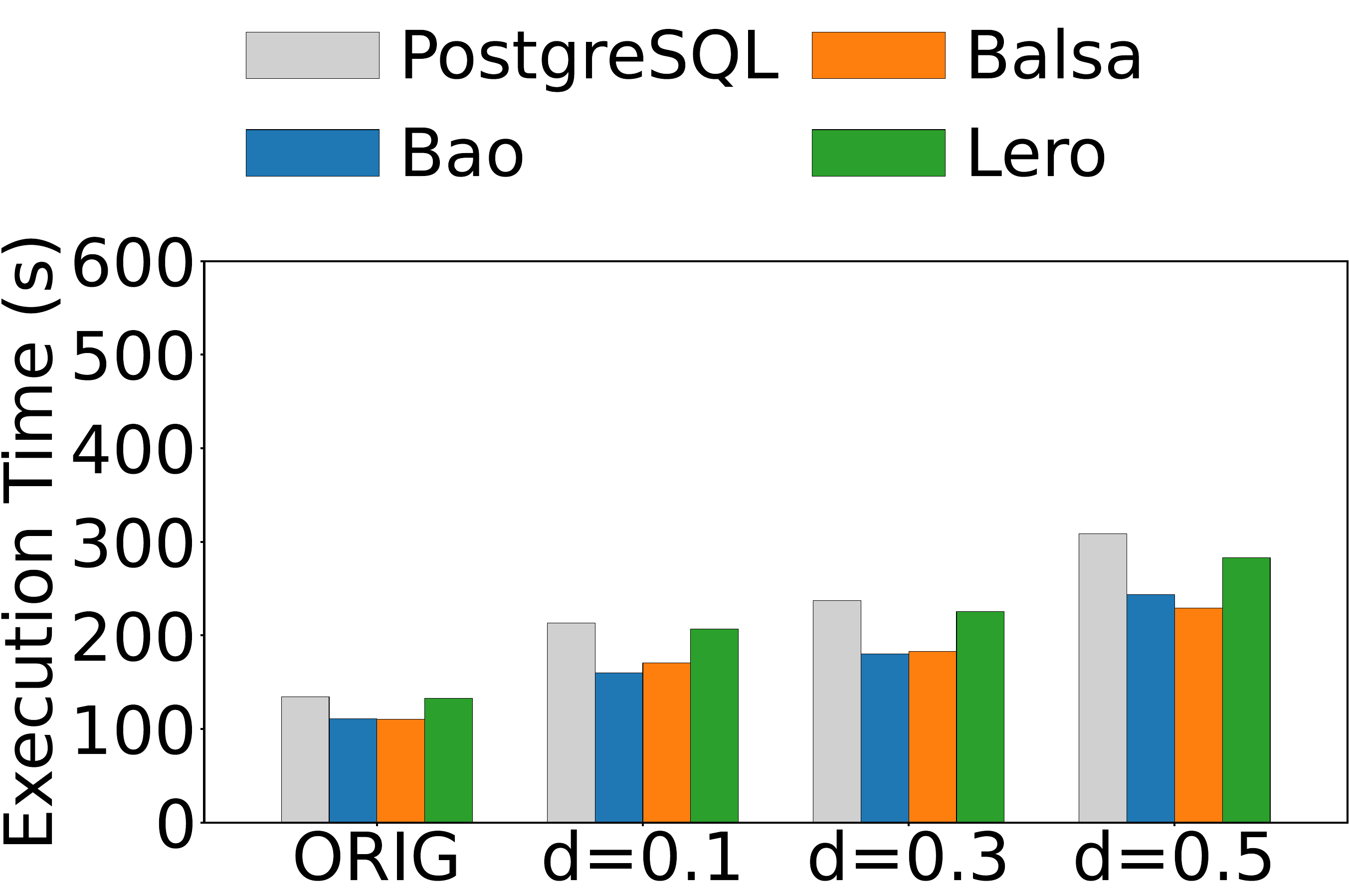}
}
\vspace{-4mm}
\captionsetup{justification=raggedright}
\caption{Analysis of Learned Query Optimizers under Various Synthetic Data Drift}
\vspace{-4mm}
\label{fig.eval.lqo_data_drift_others}
\end{figure}

\subsubsection{Experiments under synthesized data drift with varying data sizes}
\label{subsec:data_size}
We also evaluate the individual effects of data distribution drift and dataset size on learned query optimizers.
We construct ORIG+ by increasing the dataset size of ORIG to match that of 2017, while minimizing distribution drift.
Comparing ORIG+ to ORIG isolates the effect of data size, while comparing 2017* to ORIG+ indicates the impact of distribution drift alone.
As shown in Figure~\ref{subfig.lqo_data_size}, both data scale and distribution drift affect the performance of learned query optimizers.
In particular, under pure distribution drift (\emph{i.e.}, $d$=0.5), we observe clear performance degradation, indicating that learned optimizers are sensitive to distributional changes.
We further introduce $d$=0.5+, a variant of drifted IMDB data $d$=0.5 with the data size matching 2017*.
Compared to the standard distribution drift $d$=0.5, $d$=0.5+ leads to more pronounced performance degradation, suggesting that distributional drift is the dominant factor and its negative effects can be amplified by larger data sizes.

\noindent\expmessage{While data size contributes to performance variance, distribution drift more directly affects the models' robustness.}

\subsubsection{Experiments for learned query optimizers with dynamic structures}
\label{subsec:extended_drift}
The results presented thus far indicate that
both the model type and the plan search strategy have an effect on learned query optimizer robustness.
Therefore, we investigate the possibility of a learned query optimizer that is allowed to dynamically select a model based on the system conditions. 
We
implement a prototype of a dynamic learned query optimizer, D-LQO, which uses a coarse-grained router to select the optimizer (trained on non-drift data) expected to yield the lowest execution time for a given query under a specific drift.
The router is trained using query encodings as inputs and their respective execution times under different drift levels ($d$=0.1, 0.3, 0.5) in Figure~\ref{subfig.lqo_drift_factor}, as supervision. 
As shown in Figure~\ref{subfig.lqo_extended}, D-LQO achieves an 15.01\% performance improvement over PostgreSQL without retraining under $d$=0.5. 
To further study the robustness, we evaluate D-LQO on three unseen drift scenarios, $d$=0.4, IMDB(S)*, and 2019*, as specified in Section~\ref{subsec:gen_extended_drift}.
As observed in Figure~\ref{subfig.lqo_extended}, D-LQO still maintains up to 12.13\% improvement in these cases.
This robustness stems from the fact that learned query optimizers respond differently to drift. 
By learning these patterns, D-LQO selects the best-performing optimizer under a given drift.

\noindent\expmessage{A learned query optimizer with dynamic structures can deliver robustness without retraining.}

\subsubsection{Performance with periodical retraining}
\label{subsec:lqo_retraining}
We evaluate the effectiveness of periodic retraining in learned query optimizers under data drift.
In this setting, we simulate an online deployment by continuously injecting drifted data and retraining the model after a fixed number of query executions.
As shown in Figure~\ref{subfig.lqo_retrain}, periodic retraining helps the learned optimizers to regain their performance advantage over PostgreSQL.
This stems from the model's gradual adaptation to the evolving data distribution. 
As the system executes queries on drifted data, it accumulates additional knowledge reflective of the new distribution, which is then leveraged during periodic retraining to restore performance.

\noindent
\expmessage{When drift is smoothly injected and resources are sufficient, retraining is a viable strategy to enhance robustness, as its cost can be amortized across executions.}

\subsection{Experiments on Updatable Learned Indexes under Synthetic Data Drift}
\label{subsec:eval:indexes}

We study the performance of updatable learned indexes with varying synthesized data drift.
Since all the evaluated learned indexes are in-memory, we include the in-memory B$^+$-tree~\cite{DBLP:journals/pvldb/WongkhamLLZLW22} as a baseline.
We use the Facebook~\cite{DBLP:conf/infocom/GjokaKBM10} dataset to conduct our experiments. Based on it, we generate an original dataset (ORIG) and three drifted versions with drift factors $d$= 0.1, 0.3, and 0.5. 
Each drifted dataset is kept the same size as ORIG. 
We initialize each learned index using ORIG, and then incrementally apply insertions and deletions to introduce data drift. 
After applying drift, we perform point lookups to evaluate performance under different drift levels.


\begin{figure}[t]
\centering
\subfigure[Drift Factor]{
\includegraphics[height=0.15\textwidth]{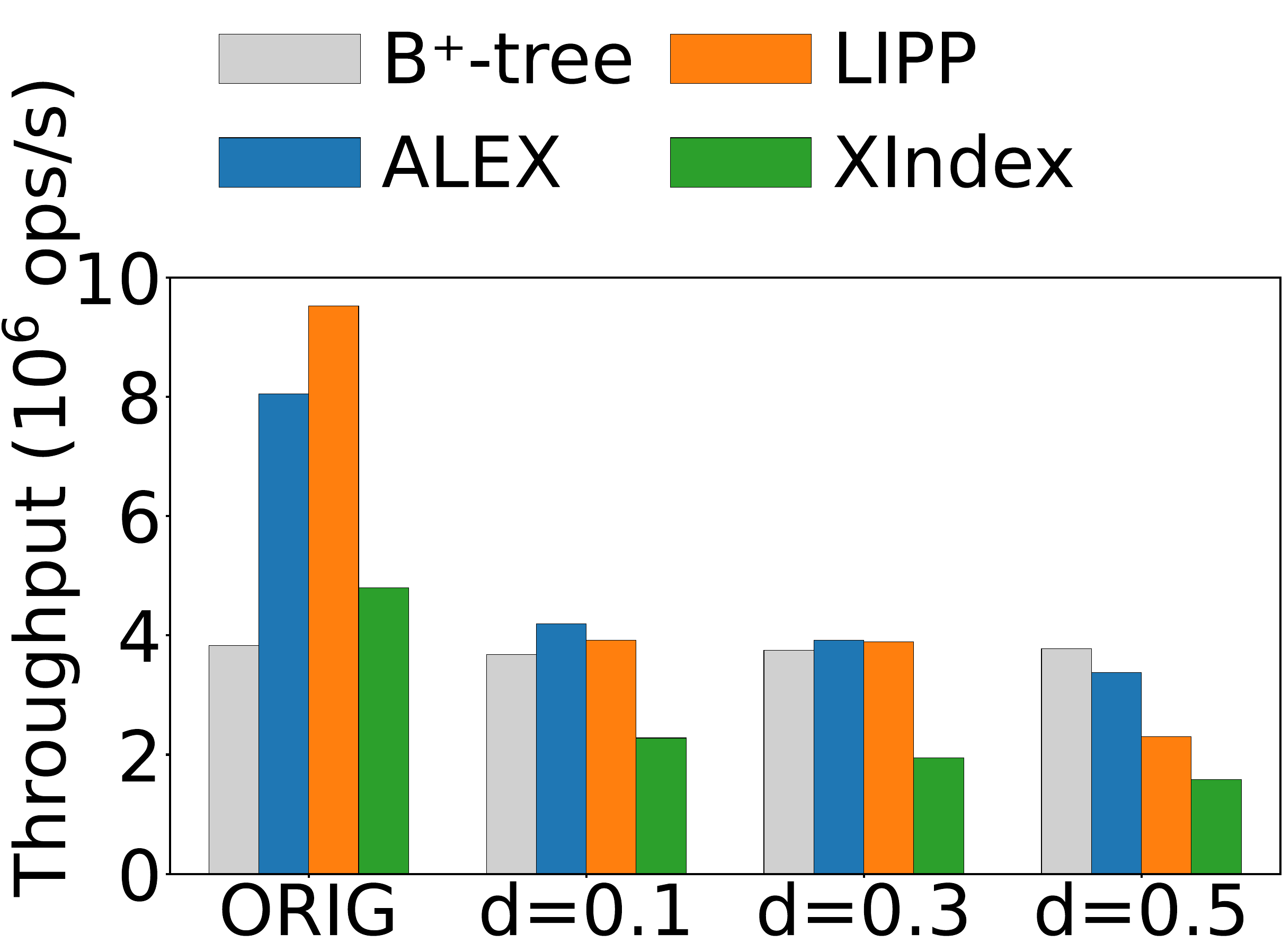}
\label{Fig.index.fb}
}\hspace{3mm}
\subfigure[Comparison with A-LI]{
\includegraphics[height=0.15\textwidth]{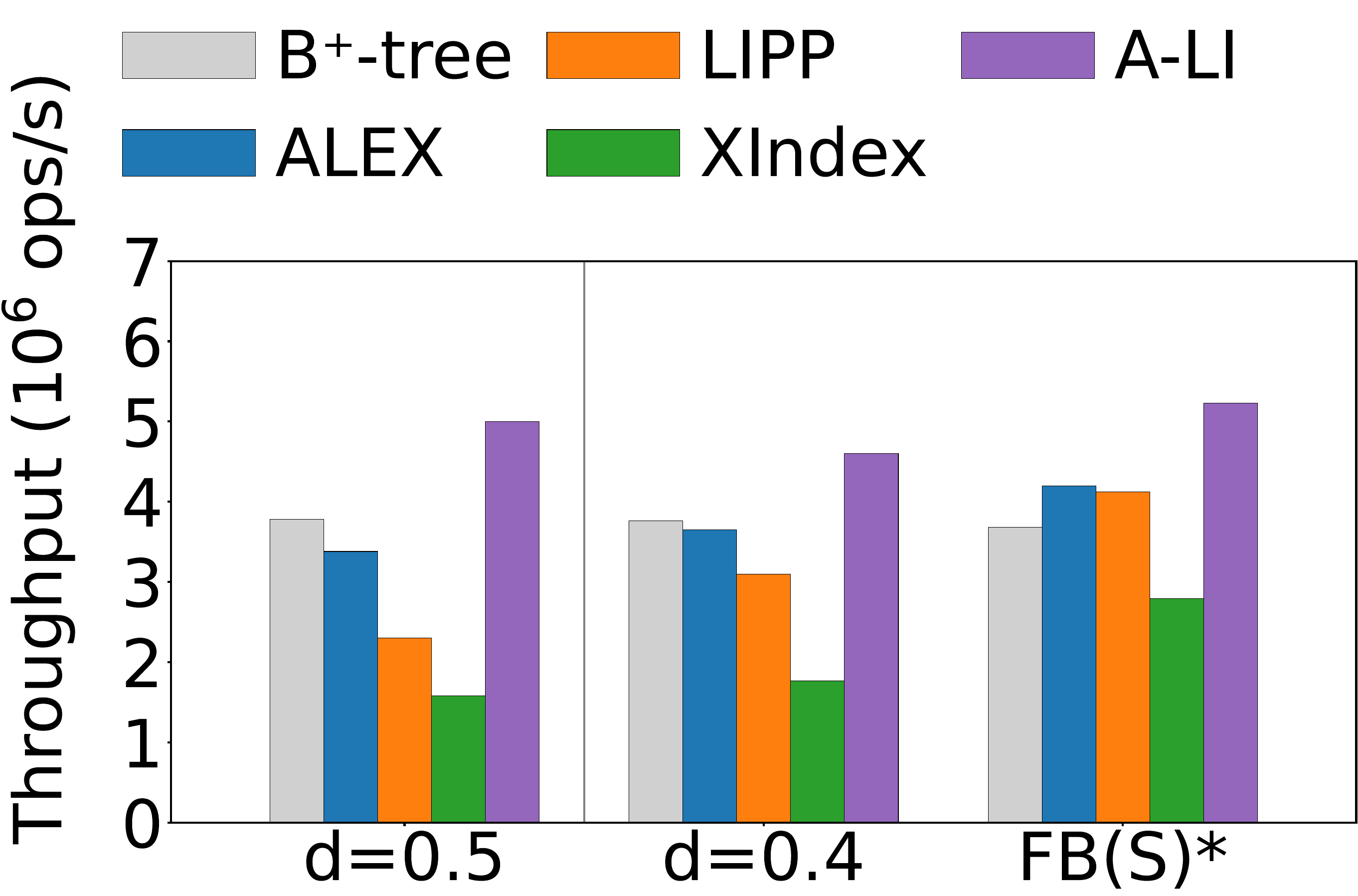}
\label{Fig.index.ali}
}
\vspace{-4mm}
\captionsetup{justification=raggedright}
\caption{Throughput of Updatable Learned Indexes under Synthetic Data Drift}
\label{fig.index_eval_imdb}
\vspace{-6mm}
\end{figure}

\begin{table}[t]
\renewcommand\arraystretch{1.1}
\vldbrevision{
\caption{\vldbrevision{Detailed Metrics of Updatable Learned Indexes}}
\label{tab:lidx_eval:acc}
\vspace{-2mm}
\resizebox{0.86\linewidth}{!}{%
\begin{tabular}{llllll}
\toprule
Competitors & Detailed Metrics    & ORIG & $d$=0.1 & $d$=0.3 & $d$=0.5 \\
\hline
B$^+$-tree        & \multirow{3}{*}{Avg Search Bound} & 3.467 & 3.507 & 3.630 & 3.652 \\
ALEX   &                      & 3.428 & 4.090 & 4.126 & 4.222 \\
XIndex      &                      & 4.556 & 5.179 & 11.620 & 12.842 \\
\hline
LIPP       & Avg Search Depth & 2.279 & 2.837 & 2.976 & 3.229 \\
\bottomrule
\end{tabular}
}
}
\vspace{-4mm}
\end{table}



\subsubsection{Experiments under synthesized data drift with various drift factors}
We plot the overall throughput in Figure~\ref{fig.index_eval_imdb}.
As data drift increases, all learned indexes exhibit degraded performance, while the B$^+$-tree maintains relatively stable throughput. 
This is due to the fact that learned indexes rely on trained models to guide key placements, but the model accuracy deteriorates as the data distribution changes, leading to higher search costs. 
ALEX and LIPP insert data in-place, which leads to higher key lookup overhead as the data distribution evolves. 
XIndex, though designed to handle updates via a separate buffer, suffers from increased buffer size under drift, leading to additional lookup overhead. 
In contrast, B$^+$-tree does not rely on data distribution modeling; its structural properties remain effective regardless of drift, thus preserving stable search bounds and throughput.
Among the learned indexes, LIPP experiences the most significant performance drop due to its chaining-based conflict resolution mechanism, which triggers frequent node splits. 
These splits deepen the tree structure, further increasing search cost.
To validate these observations, Table~\ref{tab:lidx_eval:acc} reports the search bounds on leaf nodes. 
A higher search bound and tree height indicate reduced model accuracy. 
We observe that the search bounds of ALEX, LIPP, and XIndex increase consistently with the degree of drift, aligning with their respective throughput degradation trends.

\noindent\expmessage{For learned indexes, heterogeneous structures generally offer greater stability under drift compared to homogeneous ones, but may fail to achieve optimal performance in non-drift or well-sampled scenarios.}

\subsubsection{Experiments for learned index with dynamic structures}
\label{subsec:lidx_extended}
To investigate how to enhance the robustness of learned indexes, we build A-LI, a learned index where each node independently adopts different design choices, including node size and split strategy. 
We employ a model to predict the optimal configuration parameters for a given drift level.
We execute point lookups on data with drift factors ($d$=0.1,0.3,0.5) using various parameter settings and collect the resulting throughout as supervision signals to train the model.
In our experiments, we collect 1000 samples.
Apart from $d$=0.5, we evaluate the performance of A-LI under $d$=0.4 and FB(S)*, a variant of the Facebook dataset generated using the STACK drifter.
As shown in Figure~\ref{Fig.index.ali}, A-LI achieves up to 32.18\% higher throughput than the next-best learned index under $d$=0.5, and up to 24.45\% improvement under unseen drifts.
A-LI adapts its structure by applying different split strategies and node sizes across nodes according to local drift characteristics. In particular, nodes experiencing stronger drift are assigned larger node sizes, which helps reduce search error and maintain high performance.
Under $d$=0.5, A-LI achieves an average search error of 3.037, lower than all evaluated baselines.
These results demonstrate that a learned index with an adaptive structure can remain efficient under drift. 
Recent work, such as SELIX~\cite{selix}, further explores this direction.

\noindent
\expmessage{Designing adaptive learned indexes that employ different structures for different nodes or subtrees can be a promising direction to improve robustness.}

\extended{
\begin{figure}
    \centering
    \includegraphics[width=0.4\textwidth]{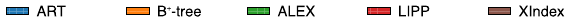} \\
    \subfigure[Throughput]{
    \label{Fig.index.concurrent.imdb.throughput}
    \includegraphics[width=0.22\textwidth]{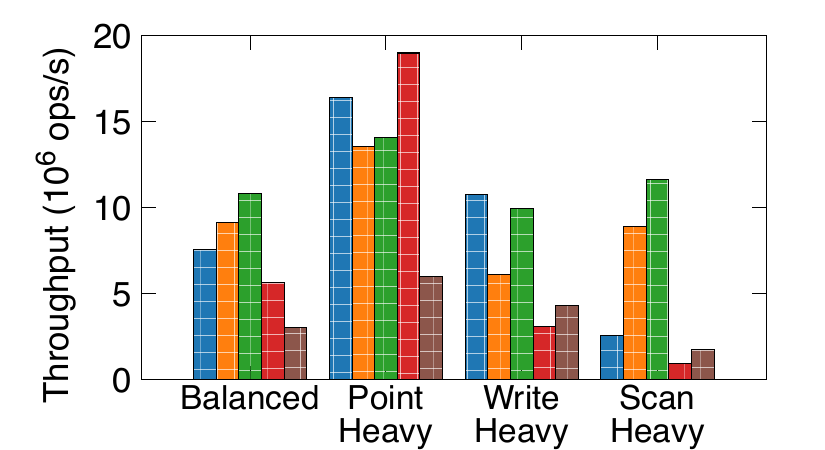}
    }
    \subfigure[Latency]{
    \label{Fig.index.concurrent.imdb.latency}
    \includegraphics[width=0.22\textwidth]{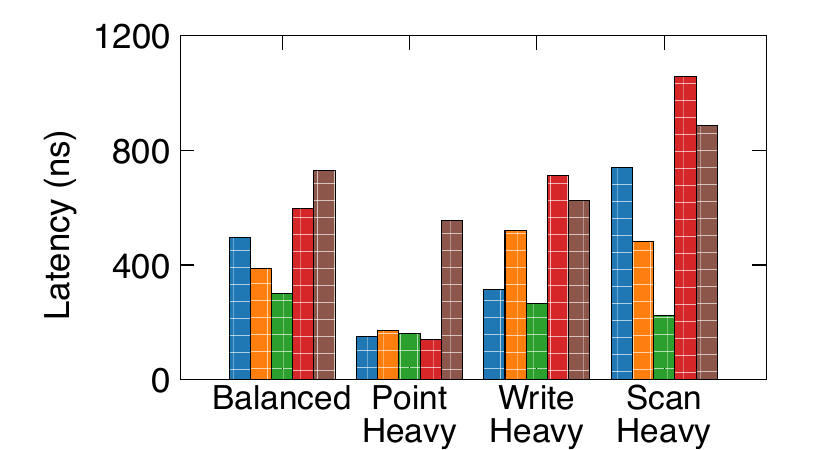}
    }
    \vspace{-4mm}
    \caption{Performance of Learned Indexes with Concurrent Read/Write Access}
    \label{Fig.index.concurrent.imdb}
    \vspace{-2mm}
\end{figure}
}

\extended{\subsubsection{Experiments on concurrent read/write access under mixed workloads.}
We evaluate the performance of learned indexes under concurrent read and write workloads.
Based on default operations specified in Section~\ref{subsec:workloaddataset}, we construct 4 mixed workloads: (1) Balanced workload, which
contains 33\% inserts to put keys into indexes, 33\% point reads to
randomly lookup a single key, and 33\% range scans to randomly read keys with the default scan length.
(2) Point-heavy workload, which consists of 90\% point reads, 5\% writes, and 5\% range scans.
(3) Write-heavy workload, which
contains 90\% writes, 5\% point reads, and 5\% range scans.
(4) Scan-heavy workload, which consists of 90\% range scans, 5\% writes, and 5\% point reads.
Each experiment executes a total of 4M operations using 4 concurrent threads.
We adopt the implementation of ART and B$^+$-tree with optimistic lock coupling (OLC)~\cite{DBLP:conf/damon/LeisSK016} as provided in~\cite{DBLP:conf/sigmod/WangPLLZKA18}, and use the concurrent implementations of ALEX and LIPP as provided in~\cite{DBLP:journals/pvldb/WongkhamLLZLW22}, to conduct our experiments.

We report the throughput and latency across different mixed workloads in Figure~\ref{Fig.index.concurrent.imdb}.
We can observe that, under the balanced workload, ALEX outperforms the other indexes due to its relatively stable performance across writes, point reads, and scans. 
Regarding the point-heavy workload, LIPP archives the best performance because of its precise point-lookup technique.
Learned indexes cannot outperform ART under the write-heavy workload, primarily due to their inherent overhead for model updates and maintenance.
In addition, ALEX demonstrates superior performance under the scan-heavy workload for similar reasons as discussed in O1.
}



\subsection{Experiments under Workload Drift}
\label{subsec:eval:extended}
\label{subsec:eval:workload_drift}

\subsubsection{Experiments on learned query optimizer}
\label{subsec:eval:ablation_workload_drift}

We apply different drift factors ($d \in {0.1, 0.3, 0.5}$) to the join patterns.
We use the generated drifted workloads to train learned query optimizers, and evaluate them on the original JOB workload, ensuring the performance can be compared between different drifts.
As observed in Figure~\ref{subfig.ablation_lqo_workload}, Balsa exhibits the most pronounced performance degradation, due to its reliance on a black-box plan enumeration process.
When training join patterns no longer align with the test workload, Balsa's model produces inaccurate plan predictions, resulting in suboptimal performance.
This is consistent with our finding in Section~\ref{subsec:eval:lqo_drift_factor}.


\extended{
We then fix the drift factor $d$ to 0.5 and conduct experiments on workloads with join pattern drift and predicate drift to investigate the performance of learned query optimizers under different types of workload drift.
Each learned optimizer is trained on the original data with the drifted workload and then evaluated using the same set of testing queries.
The results in Figure~\ref{fig.learned_eval_job5} show that
HybridQO is more sensitive to join pattern drift than Bao.
This is expected due to the differences in how hints are utilized in HybridQO and Bao.
HybridQO relies on join-order-based hints, making it more susceptible to changes in join patterns, whereas Bao uses more general hints that incorporate both join patterns and predicates, allowing it to adapt better to such drift.
}

\subsubsection{Evaluation on learned CC}

We test Polyjuice~\cite{DBLP:conf/osdi/WangDWCW0021}, with two traditional algorithms, 2PL and OCC~\cite{DBLP:conf/sigmod/0005YBPD22}, and a hybrid CC algorithm, IC3~\cite{DBLP:conf/sigmod/WangMCYCL16} under varying workload drift.
We use an initial arrival rate simulated with a single client thread, and increase the drift factors to simulate transaction arrival rates.
Polyjuice is first trained on the initial arrival rate and then retrained after each arrival rate drift, following the original retraining procedure.
We execute the default TPC-C~\cite{tpcc} transactions and plot the throughput in Figure~\ref{Fig.cc.chbenchmark.1}.
As drift increases, Polyjuice achieves higher throughput due to the increased arrival rate.
However, it requires more time to converge to an efficient policy under higher drift, as greater drift increases learning complexity due to limited prior knowledge
In addition, Polyjuice’s retraining is time-consuming, hindering its responsiveness to continuous drift in arrival rates.
These results demonstrate that learned CC still has room for improvement in fast adaptation, as discussed in recent work~\cite{neurcc}.

\extended{
To validate this point, we perform additional experiments to assess CCaaLF~\cite{DBLP:journals/corr/abs-2503-10036}, a recently proposed learned CC algorithm that enables fast adaptation under dynamic workloads.
We rerun our experiments on both the default workloads and CH-Benchmark workloads under arrival rate drift ($d=0.5$), and plot the results in Figure~\ref{fig.ccaalf}.
As observed, CaaLF achieves higher throughput and faster convergence than Polyjuice.
This improved adaptability is mainly due to the fact that NeurCC~\cite{neurcc} explicitly models transaction dependencies in action selection, enabling it to dynamically adjust its policy based on the current system state.
}

\noindent\expmessage{Fast adaptation is crucial for learned CC to be practical in real-world scenarios, since transactions can be completed in milliseconds or less.}

\extended{We conduct additional experiments using the CH-Benchmark workload under transaction arrival rate drift. 
Specifically, we run our evaluations on the transactional workloads from CH-Benchmark with varying the drift factor across 0.1, 0.3, and 0.5.
The transaction throughput results are presented in Figure~\ref{fig.cc_eval_tpcc}. 
As shown in the figure and consistent with observations in Figure~\ref{fig.cc_eval}, Polyjuice achieves higher throughput but requires longer convergence time to reach an optimal policy as the drift factor increases.}

\begin{figure}[t]
\centering
\subfigure[{Learned Query Optimizer}]{
\includegraphics[height=0.15\textwidth]{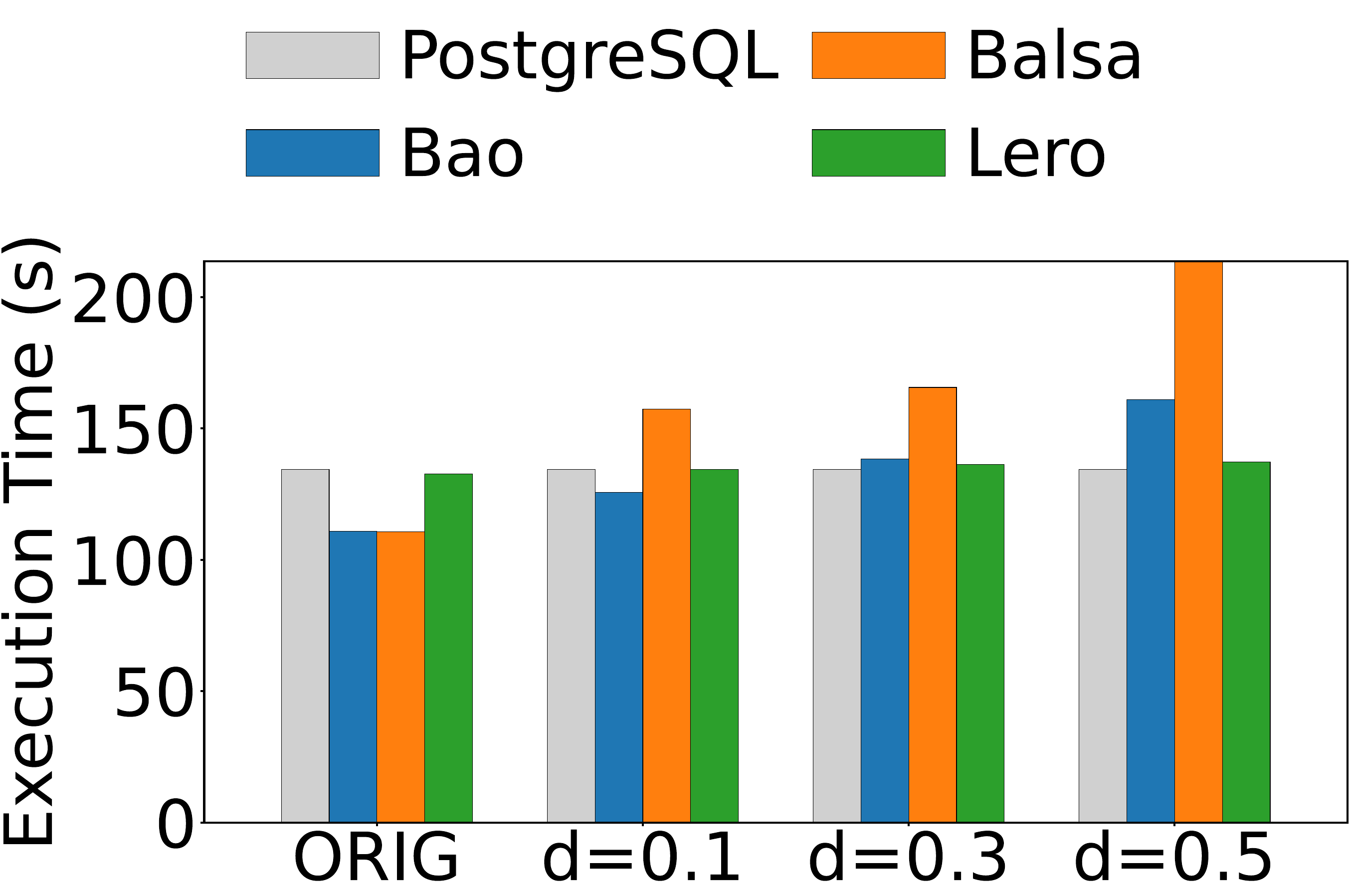}
\label{subfig.ablation_lqo_workload}
}\hspace{3mm}
\subfigure[Learned CC]{
\label{Fig.cc.chbenchmark.1}
\includegraphics[height=0.15\textwidth]{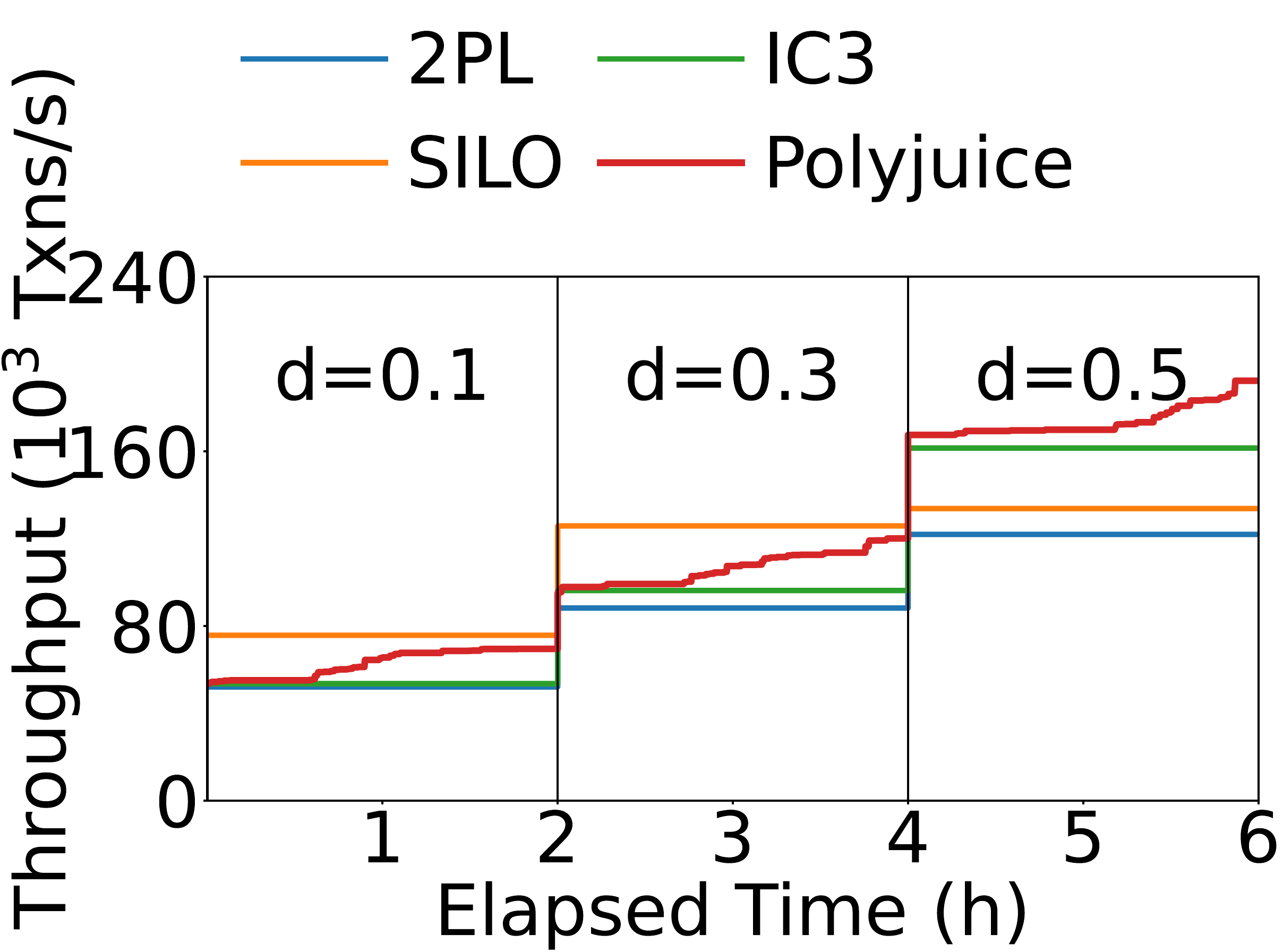}
}
\extended{
\subfigure[Execution Time - Join Pattern Drift and Predicate Drift]{
\label{fig.learned_eval_job5}
\includegraphics[width=0.42\textwidth]{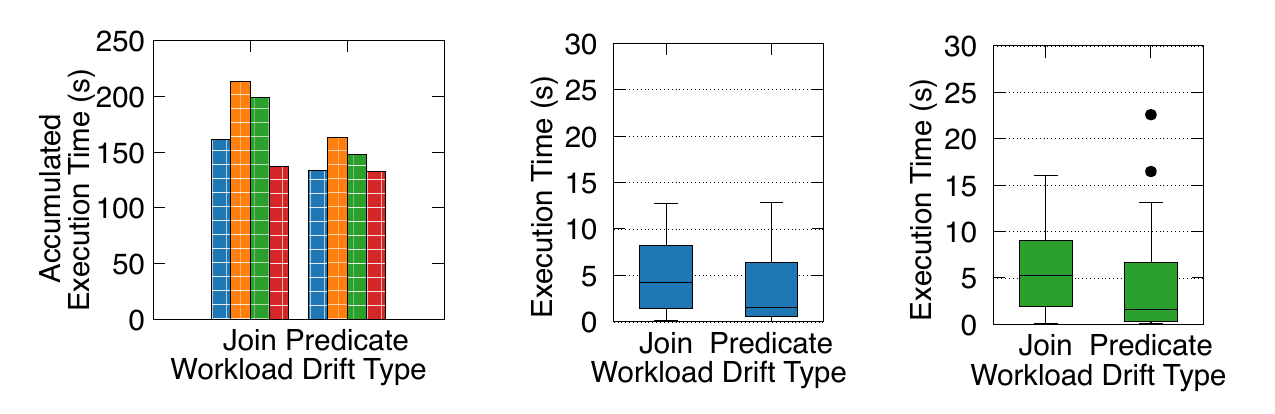}
}}
\vspace{-4mm}
\captionsetup{justification=raggedright}
\caption{Evaluations under Workload Drift}
\label{fig.lqo_eval_on_real_data_drift}
\vspace{-4mm}
\end{figure}


\extended{
\begin{figure}[t]
\centering
\subfigure[Default Workload]{
\label{Fig.cc.default.throughput}
\includegraphics[width=0.22\textwidth]{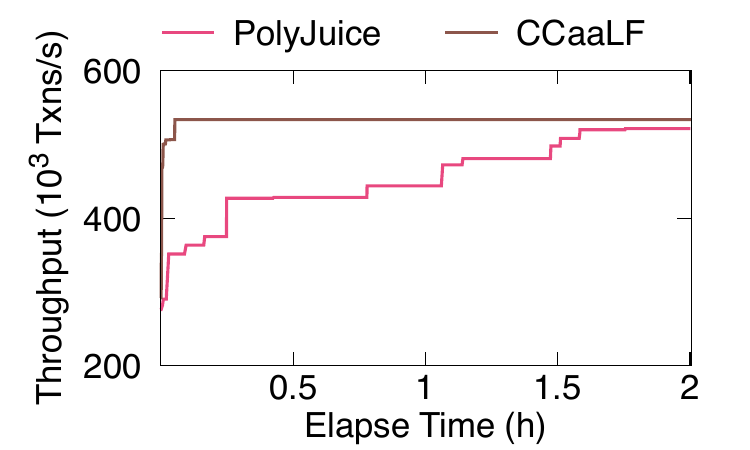}
}
\subfigure[CH-Benchmark Workload]{
\label{Fig.cc.chbenchmark.throughput}
\includegraphics[width=0.22\textwidth]{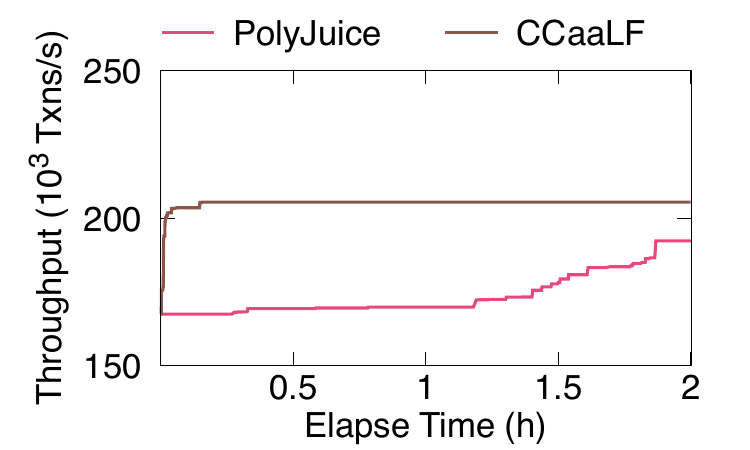}
}
\vspace{-4mm}
\captionsetup{justification=raggedright}
\caption{CCaaLF vs PolyJuice (Arrival Rate Drift with Drift Factor $d$=0.5)}
\label{fig.ccaalf}
\vspace{-2mm}
\end{figure}
}

\section{Related Work} \label{sec:related}


\paragraphtitle{Data Generators}
Existing data generators for tabular data typically focus on modeling the distribution of a static table~\cite{DBLP:journals/sigmod/RGPW24,DBLP:conf/icml/KotelnikovBRB23,DBLP:conf/iclr/Zhang0SSQFRK24}.
To improve flexibility, some recent approaches~\cite{DBLP:conf/nips/XuSCV19,DBLP:journals/pacmmod/LiuFTLD24} propose conditional generation models that support data synthesis under given schema constraints or workload patterns.
However, these methods do not explicitly incorporate distribution drift into the generation process.
As a result, their conditional mechanisms are not designed to support effective generation of evolving data distributions.
In contrast, the drift generator in \dbname explicitly models distributional changes, enabling the synthesis of data with controlled drift.
This further enables fine-grained evaluation of learned database components under evolving drift conditions.

\paragraphtitle{Benchmarks for Learned Database Components}
Multiple benchmarks have been devised for 
various learned database components~\cite{DBLP:journals/pvldb/HanWWZYTZCQPQZL21,DBLP:journals/corr/abs-2408-16170,DBLP:journals/pacmmod/LanBCB23,DBLP:journals/pvldb/SunZL23,DBLP:journals/pacmmod/MoCWCB23,DBLP:journals/pvldb/0006CJLXWLWZZWR22,DBLP:journals/pvldb/Ma0L0M22}.
Among them, several recent studies have attempted to benchmark learned database components under data and workload drift~\cite{DBLP:journals/pvldb/NegiWKTMMKA23,DBLP:journals/pvldb/LehmannSS24,DBLP:conf/sigmod/KimJSHCC22,DBLP:journals/pvldb/SunZL23,DBLP:journals/pvldb/WongkhamLLZLW22}.
These efforts typically treat drift in a case-by-case manner, and therefore only partially reveal how different types and degrees of drift affect system performance.
\dbname, in contrast, provides controllable and effective drift-aware data and workload generation, enabling a more comprehensive performance evaluation under diverse drift scenarios. 
As more learned database components emerge, such as cardinality estimators based on more generalizable models~\cite{DBLP:journals/pvldb/LiWZDZ023,DBLP:journals/corr/abs-2505-04404,DBLP:conf/cidr/KipfKRLBK19}, and as end-to-end learned DBMSs begin to integrate these components~\cite{neurdb-scis-24}, \dbname serves as a valid and general benchmark suite to support the continued improvement of their robustness.

\paragraphtitle{Benchmarks with Distribution Drift}
Several benchmarks have been proposed to evaluate performance under distribution drift.
DSB~\cite{DBLP:journals/pvldb/DingCGN21}, a standard benchmark for decision-making applications, enables evaluations under distribution drift using TPC-DS.
However, it is restricted to exponential distributions and does not quantitatively measure drift as \dbname.
In addition, concept drift modeling~\cite{DBLP:conf/kes/PorwikD22,DBLP:conf/aaai/LiYLX022,DBLP:conf/eurosys/BotherSGK23,DBLP:conf/icml/KohSMXZBHYPGLDS21,DBLP:conf/nips/GardnerPS23,DBLP:journals/pacmmod/LuoXTJC25,DBLP:journals/pacmmod/BotherRGHMTK25} has gained traction for evaluating ML-based prediction and analysis tasks.
Although they are not directly applicable to database benchmarking,  these methods also treat drift as distribution drift.
Emerging dynamic benchmarks~\cite{DBLP:conf/dbtest-ws/BensonBBLLPRSST24,DBLP:conf/nips/GardnerPS23} show promising potential.
As a dynamic benchmark,
 \dbname 
introduces measurable and controllable data and workload drift based on real drifts to enable systematic 
performance evaluations.

\section{Conclusion} \label{sec:conclusion}

This paper introduced \dbname, a new benchmark suite designed to evaluate end-to-end learned DBMSs containing all learned components under controllable data and workload drift.
We defined the drift factor to effectively quantify drift and introduced a novel drift-aware data and workload generation framework.
Extensive experimental results show the effectiveness of \dbname in generating realistic drift, and 
demonstrate that different learned database components withstand and adapt to data and workload drift
to varying degrees, depending on their specific design choices.
Through in-depth analysis, we offer recommendations for enhancing the robustness of learned database components.




\bibliographystyle{ACM-Reference-Format}
\balance
\bibliography{main-bibliography}

\end{document}